\documentclass[11pt,a4paper]{article}
\usepackage[american]{babel}
\usepackage[utf8]{inputenc}
\usepackage{amsmath}
\numberwithin{equation}{section}
\usepackage{mathtools}
\usepackage{amsthm}
\usepackage{amssymb}
\usepackage{mathabx} 
\usepackage{bbm}
\usepackage[pdfborder={0 0 0}]{hyperref}
\usepackage[toc]{appendix}
\usepackage{mathrsfs}
\usepackage[nameinlink,capitalise]{cleveref}
\usepackage{aliascnt}
\allowdisplaybreaks 

\newcommand{\cA}{{\mathcal A}}\newcommand{\cC}{{\mathcal C}}
\newcommand{\cF}{{\mathcal F}}

\newcommand{\cJ}{{\mathcal J}}\newcommand{\cK}{{\mathcal K}}\newcommand{\cL}{{\mathcal L}}
\newcommand{\cN}{{\mathcal N}}\newcommand{\cO}{{\mathcal O}}

\newcommand{\cS}{{\mathcal S}}\newcommand{\cT}{{\mathcal T}}
\newcommand{\cV}{{\mathcal V}}
\newcommand{\cZ}{{\mathcal Z}}

\newcommand{\BC}{{\mathbb C}}

\newcommand{\BN}{{\mathbb N}}
\newcommand{\BR}{{\mathbb R}}

\usepackage{dsfont} 

\usepackage{mathrsfs}

\newcommand{\sH}{{\mathscr H}}
\newcommand{\sK}{{\mathscr K}}

\newcommand{\sS}{{\mathscr S}}

\newcommand{\rmd}{{\mathrm d}}\newcommand{\rme}{{\mathrm e}}
\newcommand{\rmi}{{\mathrm i}}

\def\nn{\nonumber}

\newcommand\pscal[1]{{\ensuremath{\left\langle #1 \right\rangle}}}

\newcommand{\tphi}{\tilde{\varphi}}
\newcommand{\tl}{\Theta_\Lambda}

\theoremstyle{definition}

\newtheorem{definition}{Definition}[section]

\newaliascnt{condition}{definition}
\newtheorem{condition}[condition]{Condition}
\aliascntresetthe{condition}

\newaliascnt{assumption}{definition}
\newtheorem{assumption}[assumption]{Assumption}
\aliascntresetthe{assumption}

\newaliascnt{notation}{definition}
\newtheorem{notation}[notation]{Notation}
\aliascntresetthe{notation}

\theoremstyle{remark}

\newaliascnt{rem}{definition}
\newtheorem{rem}[rem]{Remark}
\aliascntresetthe{rem}

\theoremstyle{plain}

\newaliascnt{theorem}{definition}
\newtheorem{theorem}[theorem]{Theorem}
\aliascntresetthe{theorem}

\newaliascnt{lemma}{definition}
\newtheorem{lemma}[lemma]{Lemma}
\aliascntresetthe{lemma}

\newaliascnt{prop}{definition}
\newtheorem{prop}[prop]{Proposition}
\aliascntresetthe{prop}

\newaliascnt{cor}{definition}
\newtheorem{cor}[cor]{Corollary}
\aliascntresetthe{cor}

\begin{document}

\title{Effective Dynamics for the Bose Polaron in the Large-Volume Mean-Field Limit}
\date{}
\author{
Jonas Lampart\thanks{Universit\'e Bourgogne Europe, CNRS, Laboratoire Interdisciplinaire Carnot
de Bourgogne ICB UMR 6303, 21000 Dijon, France. \texttt{jonas.lampart@u-bourgogne.fr}}, Peter Pickl\thanks{Fachbereich Mathematik, University of Tübingen, Auf der Morgenstelle 10, 72076 Tübingen, Germany.
\texttt{p.pickl@uni-tuebingen.de}}, Siegfried Spruck\thanks{Institute for Analysis, Karlsruhe Institute of Technology, Englerstraße 2, 76131 Karlsruhe, Germany. \texttt{siegfried.spruck@kit.edu}}}

\maketitle\textbf{}

\section*{Abstract}
We consider the dynamics of the Bose polaron system, a dense quantum gas consisting of $N$ bosons evolving in $\mathbb{R}^3$ in the presence of an impurity particle. The system is studied in the mean-field scaling with initially high density $\rho$ and large volume $\Lambda$ of the gas. In the initial state, almost all bosons are in the Bose-Einstein condensate, with a few excitations. We derive from the microscopic dynamics, in the joint limit of large densities and volumes, with the constraint $\Lambda^3 \ll \rho$, the effective description by the translation-invariant Bogoliubov-Fröhlich Hamiltonian, which couples the quantum field of excitations linearly to the impurity particle.

\tableofcontents

\section{Introduction}

Understanding the emergence of effective theories in quantum many-body systems from their microscopic counterparts is a central challenge in mathematical physics.
In this work, we study an interacting system consisting of a Bose gas and a single impurity particle.
This impurity, often referred to as a tracer particle, is commonly used in experiments to probe various properties of the gas, such as its spatial density distribution \cite{SHD10} and superfluid behavior \cite{GMDA24}. Beyond these, it has found numerous additional experimental applications \cite{ZPSK10}.

 In this work, we assume that the bosonic part of the system exhibits Bose-Einstein condensation, where a macroscopic number of bosons occupy the same quantum state, called the condensate. Particles outside the condensate are referred to as excitations, and their number fluctuates due to the interacting nature of the system. These excitations were first modeled by Bogoliubov \cite{Bog47} using a quantum field theory generated by a quadratic Hamiltonian, now known as the Bogoliubov Hamiltonian. The validity of the Bogoliubov approximation has been studied extensively in the mathematical physics literature. Rigorous results for the dynamics can be found in \cite{GMM10, GMM11, LNS15,BOS15, NaNa17,MPP19,PPS20,COS24,BCS17} and for the static case in  \cite{Sei11,GrSe13,LNSS15, YaYi09, DeNa14, BBCS19, BCS21, NaTr23, BSS22, FoSo20, FoSo23, BPS21, BPPS22, Bro25}.

When an impurity is introduced into the system, it interacts with the surrounding bosonic gas. Such impurity–boson interactions can be effectively described by the Bogoliubov–Fröhlich Hamiltonian, which linearly couples the impurity—via creation and annihilation operators—to a Bogoliubov-type field of excitations \cite{GrDe16, GMDA24}.
The rigorous derivation of the Bogoliubov-Fröhlich Hamiltonian has been established on a unit volume with periodic boundary conditions. 
In the static case, Myśliwy and Seiringer derived the effective Hamiltonian in the mean-field regime \cite{MySe20}, which Lampart and Triay extended to the dilute regime (the Gross-Pitaevskii scaling) \cite{LaTr25}. These results also imply convergence of the dynamics in an appropriate sense, as shown in \cite[Corollary~1.6]{LaTr25}. 
A direct result on the mean-field dynamics with weaker assumptions on the potentials and an explicit convergence rate was obtained by Lampart and Pickl in \cite{LaPi22}.   
In a different scaling regime involving a heavy impurity in a Bose gas, other results yield effective dynamics modeling the tracer as a classical particle \cite{FrGa13, DFPP14}.

In this article, we consider, as in \cite{LaPi22}, a dense regime, meaning that the interaction range of a typical particle overlaps with many others. 
The main novelty lies in removing the periodic boundary conditions, which results in a non-constant condensate evolving in $\BR^3$.
We assume that the condensate initially varies on the large length scale $\Lambda^{1/3}\gg1$, which defines the volume $\Lambda$ of the gas, while both, the impurity–boson and boson–boson interaction ranges, are of order one. Consequently, the condensate appears nearly constant on all interaction length scales.
  These changes introduce significant new mathematical challenges and bring us closer to a physically realistic model. 
   
   In this setting, we prove $L^2$-norm convergence of the full $N$-body wave function $\psi_{t}^\Lambda$, evolving under the microscopic Hamiltonian $H_\rho$ \eqref{eq:Hamiltonian}, to the solution generated by the Bogoliubov-Fröhlich Hamiltonian $H^{\rm BF}_\Lambda$. The convergence holds in the joint limit of large densities $\rho=\frac{N}{\Lambda}$ and volumes $\Lambda$,  assuming the relation $\rho^\alpha=\Lambda$, $0<\alpha<1/3$, on $\BR^3$. Moreover, we provide an explicit convergence rate in \cref{thm:FullDynamics-BF-estimate}. 
    The effective Hamiltonian $H^{\rm BF}_\Lambda$ still depends on $\Lambda$ through condensate contributions and the Bogoliubov Hamiltonian, which creates excitations on the whole condensate volume. 
   After extracting a divergent number of excitations from $H^{\rm BF}_\Lambda$, we show that its infinite-volume limit exists. The resulting limiting Hamiltonian $H_\infty^{\rm BF}$ provides a translation-invariant effective description of the system, independent of the scaling parameters $\rho$ and $\Lambda$, and admits the explicit Bogoliubov dispersion relation (see \eqref{eq:Bogoliubov-Dispersion}).  A precise statement on the result is given in \cref{thm:Infty-Volume-Dynamics}.

A rigorous proof of the validity of the Bogoliubov-Fröhlich Hamiltonian is particularly relevant in light of its widespread use in the study of the Bose polaron, where the Bogoliubov-Fröhlich Hamiltonian serves as the underlying model for capturing quasiparticle behavior (see \cite{GMDA24} for an overview). Notably, recent mathematical results have demonstrated the existence of a stable quasiparticle for the translation-invariant Bogoliubov-Fröhlich Hamiltonian \cite{HiLa24}.

\subsection{Definition of the Model}

\paragraph{Microscopic Hamiltonian.}
We study the dynamics of a quantum gas consisting of $N$ bosons evolving in $\mathbb{R}^3$ in the presence of an impurity particle.
We assume that the Bose gas occupies a large initial volume $\Lambda\geq1$ and has high initial density $\rho=\frac{N}{\Lambda}$. Moreover, we impose the scaling relation $\rho^\alpha=\Lambda$, $\alpha\geq0$. The case $\alpha=0$ is the usual mean-field regime, whereas $\alpha=\infty$ corresponds to the thermodynamic limit. The system's microscopic dynamics is governed by the Schrödinger equation
\begin{align}
	\rmi \partial_t \psi_t^\Lambda &= H_\rho \psi_t^\Lambda\,,\quad
						\psi_0^\Lambda  = \psi^\Lambda
\end{align}
with the Hamiltonian 
\begin{align}
	H_\rho = - \frac{\Delta_x}{2m} - \sum\limits_{i=1}^{N} \frac{\Delta_{y_i}}{2} + \frac{1}{\rho}\sum\limits_{1\leq i<j\leq N} V(y_i- y_j) +\frac{1}{\sqrt{\rho}}\sum\limits_{i=1}^{N} W(x-y_i) \label{eq:Hamiltonian}
\end{align}
acting on the Hilbert space 
$$L^2(\mathbb{R}^{3}_x)\otimes L^2_{\text{sym}}(\mathbb{R}^{3N}_y)\,.$$
 Here, $x$ denotes the position of the impurity, $y_i$ the positions of the bosons, $m$ the impurity's mass in units of masses of the gas particles, and $\Delta$ the Laplace operator. 
 The interactions are weak of mean-field type with both the boson-boson interaction potential $V\in L^\infty(\BR^3,\BR)$ and boson-impurity interaction potential $W\in L^\infty(\BR^3,\BR)$ even and rapidly decreasing (see \cref{Assumption:Initial-datum-and-potential} for a precise definition). The scaling factor $\frac{1}{\rho}$ of the potential $V$ is chosen as a mean-field scaling  and the $\frac{1}{\sqrt{\rho}}$ scaling for $W$  ensures that the impurity-excitation interaction remains of $\cO(1)$ (see \cref{sec:Tracer-localization} for details). The Hamiltonian \eqref{eq:Hamiltonian} in the same scaling, without tracer particle, was first introduced in \cite{DeNa14} and later studied in \cite{DFPP16,PPS20}.

\paragraph{Effective Hamiltonian.}
We decompose the system’s dynamics into contributions from the condensate and its excitations, i.e., particles outside the condensate. The goal is to describe the excitation dynamics through an effective theory. This is achieved via the Bogoliubov approximation, which leads to the Bogoliubov Hamiltonian \cite{Bog47}.
Including the tracer particle and passing to the infinite-volume limit  then leads to the translation-invariant Bogoliubov–Fröhlich Hamiltonian $H_\infty^{\rm BF}$. This Hamiltonian provides an effective description of the coupled dynamics of the tracer and collective excitations, known as phonons. In Fourier representation for the phonons, it takes the form
\begin{align}
	\widehat{H}_\infty^{\rm BF} 
	&= \rmd\Gamma(\omega)  -\frac{\Delta_x}{2m} + a\left(   \sqrt{\tfrac{p^2}{2\omega}} \widehat{W} \rme^{-\rmi px}\right) + a^*\left( \sqrt{\tfrac{p^2}{2\omega}} \widehat{W} \rme^{-\rmi px} \right) \,, \label{eq:H-BF-Infty-Volume}
\end{align}
acting on $L^2(\BR^3_x,\cF(L^2(\BR^3)))$, where the symmetric Fock space $\cF(L^2(\BR^3))$ is used to describe excitations.
The Bogoliubov dispersion relation is given by
\begin{align}
	\omega(p)=\sqrt{\tfrac{p^4}{4} + p^2 \widehat{V}(p) (2\pi)^{3/2}}\,, \label{eq:Bogoliubov-Dispersion}
\end{align} 
where we assumed that $\widehat{V}\geq0$. By the commutator theorem \cite[Theorem~X.37]{ReSi75ii}, it is readily verified that $\widehat{H}_\infty^{\rm BF} $ is essentially self-adjoint on every core of $-\frac{\Delta_x}{2m}+ \rmd \Gamma(p^2+1)+1$.
This model describes an impurity linearly coupled to a freely evolving field of phonons, in a homogeneous condensate with density equal to one \cite{GrDe16}.

The translation-invariant and time-independent Hamiltonian $H_\infty^{\rm BF}$ can then be fiber-decomposed in the total momentum of the system as in \cite{HiLa24}, to obtain an effective dispersion relation for the impurity particle. This provides direct information about its quasiparticle behavior.

\subsection{Mean-Field Description of the Condensate}
In the mean-field regime, the dynamics of the condensate wave function is effectively described by the Hartree equation:
 \begin{align}
 	\rmi\partial_t\varphi_t^\Lambda &= h_t^\Lambda  \varphi_t^\Lambda \,,\quad \varphi_{t=0}^\Lambda=\varphi_0^\Lambda \label{eq:Hartree} 							
 \end{align}
 with 
 \begin{align}
 	h_t^\Lambda=h[\varphi_t^\Lambda]&= -\frac{\Delta}{2} + V\ast |\varphi_t^\Lambda|^2 -\mu_t^\Lambda\,, \label{eq:Hartree-Ham} 
 \end{align}
where the convolution term $V\ast|\varphi_t^\Lambda|^2(y)= \int V(x-y) |\varphi_t^\Lambda|^2(x) dx$ accounts for the mean-field interaction between bosons. The constant $\mu_t^\Lambda\in \BR$ can be freely chosen, as it only affects the global phase of $\varphi_t^\Lambda$. We choose the normalization $\Vert \varphi_0^\Lambda\Vert_2=\Lambda^{1/2}$, and, following the convention of  \cite{LNS15}, set
\begin{align}
\mu_t^\Lambda := \frac{1}{2}\pscal{\frac{\varphi_t^\Lambda}{\Lambda^{1/2}}, V\ast |\varphi_t^\Lambda|^2 \frac{\varphi_t^\Lambda}{\Lambda^{1/2}}}\,. \label{eq:Mu}
\end{align}  
For initial data $\varphi_0^\Lambda\in H^\infty(\BR^3)$, standard arguments using Duhamel's formula yield a unique global solution $\varphi^\Lambda\in C^1\big(\BR,H^\infty(\BR^3)\big)$ of the Hartree equation, which we refer to as the condensate wave function.
 	   The interaction between the impurity and the condensate is not included in \eqref{eq:Hartree-Ham}, as it is subleading in the regime of large density $\rho$.  The condensate, consisting of $\mathcal{O}(N)$ bosons, thus evolves independently of the single impurity.
The rigorous justification of the Hartree description in the setting considered here without the impurity is given in \cite{DFPP16, PPS20}.

\subsection{Excitation Dynamics}

\paragraph{Excitation Representation.}
To effectively describe excitations out of the condensate, we use the excitation representation \cite{LNSS15}. In this framework, a given $N$-body wave function $\psi\in L^2(\mathbb{R}^3)^{\otimes_s N}$ is decomposed into a component in the direction of the condensate $\varphi_t^\Lambda$ and a component orthogonal to it. 
To formalize this approach and the special role of the condensate, we introduce the orthogonal projection $P_t^\Lambda:L^2(\mathbb{R}^3) \to L^2(\mathbb{R}^3)$ onto the condensate via
\begin{align}
P_t^\Lambda\psi :=\left|\frac{\varphi_t^\Lambda}{\Lambda^{1/2}}\right\rangle \left \langle\frac{\varphi_t^\Lambda}{\Lambda^{1/2}}\right| \psi:= \frac{\varphi_t^\Lambda}{\Lambda^{1/2}} \pscal{\frac{\varphi_t^\Lambda}{\Lambda^{1/2}} , \psi} \,. \label{eq:Condensate-Projection} 
\end{align} 
The projection onto excitations is given by $Q_t^\Lambda=1-P_t^\Lambda$.
Any $\psi\in L^2(\mathbb{R}^3)^{\otimes_s N}$ can be decomposed into
\[\psi= (P_t^\Lambda+Q_t^\Lambda)^{\otimes N}\psi= \sum\limits_{k=0}^N \left(\frac{\varphi_t^\Lambda}{\Lambda^{1/2}}\right)^{\otimes N-k} \otimes_s (\chi_t^\Lambda)^{(k)}\]
with unique $(\chi_t^\Lambda)^{(k)} \in \big(\{\varphi_t^\Lambda\}^\perp\big) ^{\otimes_s k}$. 
To analyze excitations out of the condensate, we define the excitation map $U_t^\Lambda:= U(\varphi_t^\Lambda)  : L^2(\mathbb{R}^3)^{\otimes_s N} \to \mathscr{F}_{+,t}^\Lambda$ by
\begin{align}
	 U( \varphi_t^\Lambda )\psi = \bigoplus\limits^N_{k=0} (\chi_t^\Lambda)^{(k)}\,, \label{eq:Excitation-map}
\end{align}
 which acts on $N$-particle wave functions and maps them into the excitation space $\mathscr{F}_{+,t}^\Lambda:=\bigoplus\limits^\infty_{k=0} \big(\{\varphi_t^\Lambda\}^\perp\big)^{\otimes_s k}$.
Note that $U_t^\Lambda $ is an isometry.
We can now introduce the excitation Hamiltonian
 \begin{align}
 	H_{\rho}^{\rm{ex}}(t) =  U_t^\Lambda  H_\rho (U_t^\Lambda )^* + \rmi  (\partial_tU_t^\Lambda ) (U_t^\Lambda )^*  \,, \label{eq:Ex-Ham}
\end{align}
which describes the microscopic dynamics in the excitation space, by satisfying
	\begin{align}
		\rmi \partial_t \psi_{t}^\Lambda = H_\rho \psi_{t}^\Lambda \quad \Leftrightarrow \quad \rmi \partial_t U_t^\Lambda  \psi_{t}^\Lambda =  H_{\rho}^{\rm ex}(t) U_t^\Lambda  \psi_{t}^\Lambda \,. 
	\end{align}

\paragraph{Bogoliubov-Fröhlich Hamiltonian.} To motivate the definition of $H_\infty^{\rm BF}$, we first introduce the volume-dependent Bogoliubov-Fröhlich Hamiltonian $H_\Lambda^{\rm BF}(t)$ and show that we can make sense of its infinite-volume limit, giving rise to $H_\infty^{\rm BF}$. 

We start by deriving the Bogoliubov Hamiltonian from $H_{\rho}^{\rm{ex}}(t)$ by expressing it in second quantization with a basis that includes the condensate wave function. Replacing all creation and annihilation operators of the condensate by $\sqrt{N}$, and neglecting terms that are small when the number of excitations is negligible compared to the total particle number $N$ then yields the Bogoliubov Hamiltonian. 
 Applying this procedure to all parts of the excitation Hamiltonian that are independent of the tracer leads to the effective Bogoliubov Hamiltonian, describing the excitation dynamics without an impurity. It acts on the Fock space $\cF(L^2(\BR^3))$ and has the form
\begin{align}
	H^{\rm Bog}_\Lambda(t) = \rmd\Gamma\big(h_t^\Lambda+K_1^\Lambda(t)\big) + \frac{1}{2} \sum_{m,n\geq0} \big( (K_2^\Lambda(t)J)_{mn} a_m^*a_n^* + \mathrm{h.c.} \big)\,, \label{eq:HBog-Definition}
\end{align}
where $h_t^\Lambda$ is the mean-field Hamiltonian from the Hartree equation  \eqref{eq:Hartree}, $\mathrm{h.c.}$ denotes the Hermitian conjugate and the operators $K_1^\Lambda(t)$ and $K_2^\Lambda(t)$ are defined by
	\begin{align}
		K_1^\Lambda(t)&:=Q_t^\Lambda \tilde{K}_1^\Lambda(t)Q_t^\Lambda \,, \quad &\tilde{K}_1^\Lambda(t):L^2(\BR^3)\to L^2(\BR^3) \,, \label{eq:K1-Def}\\
		K_2^\Lambda(t)&:=Q_t^\Lambda \tilde{K}_2^\Lambda(t) J Q_t^\Lambda J^*\,, \quad &\tilde{K}_2^\Lambda(t):(L^2(\BR^3))^* \to L^2(\BR^3)\,,  \label{eq:K2-Def}
	\end{align}
	with
	\begin{align}
		[\tilde{K}_1^\Lambda(t) \psi](x) &:= \int [\varphi_t^\Lambda(x)V(x-y)(\varphi_t^\Lambda)^*(y)] \psi(y) dy \,, \label{eq:K1-Def-2}\\
		[\tilde{K}_2^\Lambda(t) J\psi](x) &:= \int [\varphi_t^\Lambda(x) \varphi_t^\Lambda(y) V(x-y)] \psi^*(y) dy \,,\label{eq:K2-Def-2}
	\end{align}
	where $J$ maps $L^2(\BR^3)$ into its dual with $J\psi=\pscal{\psi,\,.\,}_{L^2}$.
	 Note that $K_1^\Lambda\in C^1(\BR,\text{HS}(L^2))$, $K_2^\Lambda \in C^1(\BR, \text{HS}((L^2)^*,L^2))$, $K_1^\Lambda(t)$ self-adjoint and we have $K_2^\Lambda(t)^*= J K_2^\Lambda(t)J$.
The operators $a_m^*$ and $a_m$ denote creation and annihilation operators of $u_m$, where $\{u_m(t)\}_{m\in\BN}$ is a time-dependent orthonormal basis of $L^2(\BR^3)$ including the normalized condensate $u_0(t):=\varphi_t^\Lambda/\Lambda^{1/2}$. 
We set $(K_2^\Lambda(t)J)_{mn} =\pscal{u_m,K_2^\Lambda(t)J u_n}$.
To prove the validity of the Bogoliubov dynamics, it is essential to control the number of excitations (see \cref{sec:Bogoliubov-Approximation}).

Repeating the procedure above for the missing tracer contributions, we obtain the finite-volume Bogoliubov-Fröhlich Hamiltonian acting on the space $L^2(\BR^3)\otimes \cF(L^2(\BR^3))= L^2( \BR^3,\cF(L^2(\BR^3)))$,  effectively modeling the dynamics of the full system:
\begin{align}
	H^{\rm BF}_\Lambda(t)= -\frac{\Delta_x}{2m} +  a ( Q_t^\Lambda W_x \varphi_t^\Lambda) + a^*( Q_t^\Lambda W_x \varphi_t^\Lambda) + H^{\rm Bog}_\Lambda(t)  \,, \label{eq:HBF-Ham}
\end{align}
where $W_x(y):=W(x-y)$.
This model describes an impurity linearly coupled to the field of excitations, while the excitation-excitation interaction is quadratic in the creation and annihilation operators. 

Note that in $H^{\rm BF}_\Lambda$ the mean-field contribution of the condensate–impurity interaction, $\rho^{1/2} W \ast |\varphi_t^\Lambda|^2(x)$, is omitted. In fact, given that the condensate is sufficiently flat near the impurity at the initial time, it can be approximated by a constant, so that $\rho^{1/2} W\ast |\varphi_t^\Lambda|^2(x) \sim \rho^{1/2} W\ast 1$. Thus, it can be neglected in the dynamics, as shown in \cref{sec:Control-tracer-condensate-mean-field-interaction}. This omission is crucial: a non-constant mean-field contribution would dominate the impurity's dynamics, masking its interaction with the excitations and potentially causing it to escape the Bose gas on short timescales (see \cref{rem:Flat-condensate}).

\begin{rem}[Bogoliubov-Fröhlich Dynamics]\, \label{rem:Bog-BF-Ham}
The differential equation
	$
	\rmi \partial_t \psi_{\Lambda,t}^{\mathrm{BF}} ={}H^{\rm BF}_\Lambda(t) \psi_{\Lambda,t}^{\mathrm{BF}}$, $\psi_{\Lambda,t=0}^{\mathrm{BF}}= \psi_0
	$ has in a weak sense the unique global solution given by 
	$
	\psi_{\Lambda,t}^{\mathrm{BF}}={}U^{\rm BF}_\Lambda(t,0) \psi_0$, $\forall \psi_0\in L^2(\BR^3,\cF(L^2)),
	$
	where $U^{\rm BF}_\Lambda(t,0)$ is the unitary propagator of $H^{\rm BF}_\Lambda(t)$ (see \cite[Theorem 8]{LNS15} and \cite[Appendix D]{Spr25}).
	The dynamics generated by $H^{\rm BF}_\Lambda(t)$ leaves the excitation space invariant: 
	\begin{align*}
		U^{\rm BF}_\Lambda(t,t_0)\big(L^2(\BR^3,\cF(\{\varphi_{t_0}^\Lambda\}^{\perp}))\cap L^2(\BR^3,Q(\cN))\big) &\subset L^2(\BR^3,\cF(\{\varphi_t^\Lambda\}^{\perp}))\,,
	\end{align*}
	which can be proven analogously to \cite[Theorem 7]{LNS15}.
\end{rem}

\paragraph{Infinite-Volume Dynamics.}
The effective Hamiltonian $H^{\rm BF}_\Lambda(t)$ depends on the volume $\Lambda$ in two ways. First, through the condensate $\varphi_t^\Lambda$, whose initial datum varies on the scale $\Lambda^{1/3}$. Second, through $H^{\rm Bog}_\Lambda(t)$, which creates excitations on the whole condensate volume $\Lambda$.
 	In order to obtain a genuine limiting dynamics, independent of the scaling parameter $\Lambda=\rho^\alpha$, we have to remove the $\Lambda$-dependence from the effective dynamics. 
 	In addition the initial data may also contain a number of excitations that diverges with the volume and must be extracted (see \cref{sec:Infinite-Volume-Dynamics} for details).

Given the right initial conditions on the condensate, especially that it is flat around the origin, we can show that $\varphi_t^\Lambda$ converges locally to a phase $\rme^{-it\nu}$, $\nu\in\BR$ (see \cref{cor:Strong-convergence-infinite-volume})

To study the infinite-volume limit of $H_{\Lambda}^{\rm Bog}(t)$, we introduce Bogoliubov transformations. 
 We call an operator $\cZ\in \cL\big( L^2(\BR^3)\oplus J  L^2(\BR^3)\big)$ a Bogoliubov map if
 		 \begin{align}
 		 \cZ:=\begin{pmatrix}
	c & J^*bJ^* \\
	b & JcJ^* 
\end{pmatrix} \label{eq:def-Q-0}
\end{align} 		  
and if it satisfies the symplectic condition $\cZ S \cZ^*=\cZ^*S\cZ =S$, where $S=\text{diag}(1,-1)$. A Bogoliubov map $\cZ$ is called  unitarily implementable if there exists a unitary $U_{\cZ}$ on the bosonic Fock space $\cF(L^2)$ with
 		 \begin{align}
 		 	U_{\cZ}a(f)U_{\cZ}^*&= a( cf) +a^*(J^*bf) \,, \label{eq:Bog-creation-annihilation} \\
 		 	U_{\cZ}a^*(f)U_{\cZ}^*&=a(J^*bf) + a^*(cf) \,.
 		 \end{align}
 	 We refer to $U_{\cZ}$ as a Bogoliubov transformation, and the states $U_{\cZ}\Omega$ are called quasi-free states (for a detailed discussion of Bogoliubov transformations see \cite{NNS16,Nap18,BPPS22}).

Now, to obtain a candidate for the limit of $H^{\rm Bog}_\Lambda(t)$ we look at its propagator $U_{\Lambda,t}^{\mathrm{Bog}}:=U^{\mathrm{Bog}}_\Lambda(t,0)=U_{\cV_t^\Lambda}$, which is a Bogoliubov transformation (see e.g. \cite[Theorem~2.2]{AKS13} or \cite[Lemma~4.8]{BPPS22}), where $\cV_t^\Lambda\in \cL(L^2(\BR^3)\oplus JL^2(\BR^3))$ is given by
\begin{align}
 \rmi\partial_t \cV_t^\Lambda &=\cA^\Lambda(t)   \cV_t^\Lambda \,, \quad	\cV_0^\Lambda=I \,, \label{eq:V-Matrix-Representation}
\end{align}  with 
\begin{align}
\cA^\Lambda(t) = \begin{pmatrix}
h_t^\Lambda + K_1^\Lambda(t) & -K_2^\Lambda(t) \\
K_2^\Lambda(t)^* & -J( h_t^\Lambda + K_1^\Lambda(t))J^*
\end{pmatrix} \,.  \label{eq:A(t)-Def}
\end{align}
On the level of the Bogoliubov maps $\cV_t^\Lambda$, we are able to take the limit $\Lambda\to \infty$. That is, there exists a limiting Bogoliubov map $\cV_t^\infty$, such that $\cV_t^\Lambda \rme^{\rmi \nu S} \to \cV_t^\infty$ in the strong sense.
 For boson–boson interactions of positive type, $\widehat{V}\geq0$, the limiting dynamics admits a diagonal form in Fourier representation:
\begin{align}
	 \cT\widehat{\cV_t^\infty}\cT^{-1} =   \begin{pmatrix}
	e^{-\rmi t \omega} & 0 \\
	0 & Je^{-\rmi t \omega}J^*
\end{pmatrix}	   \,, \label{eq:Diagonal-V-infty-1}
\end{align}
for a diagonalizing operator 
\begin{align}
  	\cT= \frac{1}{2} 
  	\begin{pmatrix}
  		\tau +\tau^{-1} & (\tau -\tau^{-1}) R \cC  J^*\\
  		J \cC  R (\tau -\tau^{-1}) & J(\tau +\tau^{-1})   J^*
  	\end{pmatrix}  	
  	 \,, \label{def:T-infty-volume}
\end{align}
where $\tau(p) = \sqrt{p^2/(2\omega(p))}$,
$\cC\psi=\psi^*$ and $R\psi(p)=\psi(-p)$ (see for example \cite{BrDe07}). Although $\cV_t^\infty$ itself is not unitarily implementable, the diagonal evolution $e^{-\rmi t \omega} \oplus  J e^{-\rmi t \omega} J^*$ admits a unitary implementation.
We emphasize that $\cT$ does not satisfy Shale’s criterion \cite{Sha62} for unitary implementability on Fock space, since the off-diagonal term $\tau -\tau^{-1}$ fails to be Hilbert–Schmidt. In fact, $\cT$ is not even a Bogoliubov map in the strict sense, as it is unbounded due to the infrared divergence of $\tau^{-1}$. This creates substantial technical difficulties in rigorously conjugating  the many‑body dynamics on Fock space by $\cT$. 

Nevertheless, at a formal level one can proceed as follows: First, take the formal limit $\Lambda\to \infty$ in the $H_\Lambda^{\rm BF} $. Then, assuming hypothetically that $\cT$ were unitarily implementable, one conjugates the limiting Hamiltonian by the corresponding unitary.  This procedure yields a natural candidate for the infinite-volume Bogoliubov–Fröhlich Hamiltonian governing the transformed dynamics, namely $H_\infty^{\rm BF} $ in \eqref{eq:H-BF-Infty-Volume} (see \cref{sec:Infinite-Volume-Dynamics} for more details). The rigorous validity of $H_\infty^{\rm BF} $ is established in \cref{thm:Infty-Volume-Dynamics}.

\section{Main Results} \label{sec:Main-Result}
We now come to the main results of this article, which establish the validity of the Bogoliubov-Fröhlich dynamics in the joint limit of large initial volumes $\Lambda=\rho^\alpha$ and large initial densities $\rho$ for $0<\alpha<1/3$.
We first present the result for the infinite-volume Hamiltonian $H^{\mathrm{BF}}_\infty$, followed by the corresponding statement for the finite-volume case $H^{\mathrm{BF}}_\Lambda (t)$.

\subsection{Infinite-Volume Dynamics}

 To derive the infinite-volume limit, we introduce a unitary implementable approximation $\cZ_0^\Lambda$ of the limiting operator $\widecheck{\cT}$, which diagonalizes the Bogoliubov dynamics.
 We then show that the transformed microscopic dynamics, with initial state $U_{\cZ_0^\Lambda}U_0^\Lambda\psi_0^\Lambda  \to \psi_0^\infty$, converges to the dynamics of the infinite-volume Bogoliubov-Fröhlich Hamiltonian $H_\infty^{\rm BF}$ with initial state $\psi_0^\infty$.

Applying the Bogoliubov transformation $U_{\cZ_0^\Lambda}$ to the microscopic initial state $U_0^\Lambda\psi_0^\Lambda $ yields a representation in which a divergent number of excitations, scaling with the volume, can be extracted.  These excitations are interpreted as arising from the Bogoliubov-transformed vacuum (see \cref{rem:Tech-particle-number-estimate}). 
After extracting this expected divergent contribution from the initial state, it is natural to take the infinite-volume limit $U_{\cZ_0^\Lambda}U_0^\Lambda\psi_0^\Lambda \to \psi_0^\infty$, which is further justified in the proof of  \cref{thm:Infinite-Volume-Dynamics}.

The precise conditions under which $\cZ_0^\Lambda$ approximates $\widecheck{\cT}$ are collected in the following condition, which will later be applied to more general maps $\cZ_0^\infty$ beyond the specific choice $\widecheck{\cT}$.
\begin{condition} \label{con:Infinite-volume-Bog-maps}
Let $\cZ_0^\Lambda $ be a family of unitarily implementable Bogoliubov maps, $\epsilon>0$, and $\varphi_0^\Lambda$ the initial condensate. We say  $\cZ_0^\Lambda $ satisfies \cref{con:Infinite-volume-Bog-maps} with growth rate $\epsilon$ if there exists a constant $C>0$ such that, for all $\Lambda\geq1$ 
\begin{gather}
	\Vert \cZ_0^\Lambda\Vert_{\rm op}\leq C\Lambda^{\epsilon}\,,\quad \Vert \cZ_0^\Lambda (\cZ_0^\Lambda)^*-1\Vert_{\rm HS} \leq C\Lambda^{1/2+\epsilon}\,, \label{eq:Z-Lambda-epsilon-bounds}
\end{gather}
and if $U_{\cZ_0^\Lambda}^*$ leaves $\cF(\{\varphi_0^\Lambda\}^\perp)$ and $Q( \rmd\Gamma(-\Delta+1))$ invariant.  
	
	Moreover, we suppose that there exists a linear operator $\cZ_0^\infty:D(\widecheck{\cT})\subset L^2\oplus JL^2 \to L^2\oplus JL^2$ such that $\cJ \cZ_0^\infty \cJ =\cZ_0^\infty$ with $\cJ=\begin{pmatrix}
		& J^* \\
		J & 
	\end{pmatrix}$, $\cZ_0^\infty \widecheck{\cT}^{-1}$ is bounded, and $\cZ_0^\infty$ is approximated by $\cZ_0^\Lambda$. That is,
	\begin{align}
		\left\Vert \left(\widehat{\cZ_0^\Lambda} - \widehat{\cZ_0^\infty} \right)\cT^{-1} (\tau\oplus J\tau J^*) F\right\Vert_{L^2\oplus JL^2} \to 0\,, \quad \text{as } \Lambda\to \infty \label{eq:Z-Lambda-convergence-to-T-infty}
	\end{align}
	 for all $F\in L^2\oplus JL^2$. Furthermore, we require that the commutator of $\cZ_0^\Lambda-\cZ_0^\infty$ with translations $T_x$, $T_xf(y)= f(y-x)$, converges to zero: 
	 \begin{align}
	 	\sup_{x\in\BR^3} \frac{1}{(1+x^2)^{1/2}} \Big\Vert  \left[ \cZ_0^\Lambda- \cZ_0^\infty,T_x\oplus JT_xJ^*\right] F\Big\Vert_{L^2\oplus JL^2} \to 0\,, \quad \text{as } \Lambda\to \infty \label{eq:Z-Lambda-commutator-translations-convergence}
	 \end{align} for all $F\in L^2\oplus JL^2$.
\end{condition}

An  explicit example of a family $\cZ_0^\Lambda$ that approximates $\widecheck{\cT}$ in the sense of 	\cref{con:Infinite-volume-Bog-maps}, with an arbitrary growth rate $0<\epsilon<1/3$, is constructed in \cref{sec:Construction-of-Z0}, assuming that the initial condensate $\varphi_0^\Lambda$ is real-valued.

We are now in a position to state the main result in the infinite-volume case.

\begin{theorem}[Infinite-Volume Bogoliubov-Fröhlich Dynamics]\label{thm:Infty-Volume-Dynamics}
 Assume that the potentials $V$ and $W$ are Schwartz functions, with $\widehat{V}\geq0$ and $\widehat{V}(0)>0$. For given $0<\alpha<1/3$ choose $\Lambda=\rho^{\alpha}$, $0<\epsilon<1/4\min\{1/6,\, 2/\alpha(1/3-\alpha)\}$, and $n\in\mathbb{N}_+$ with $n>9/4(1/\alpha-2)$.

Let  $\eta\in H^\infty (\BR^3)$, $\Vert \eta\Vert_2=1$, and let the initial data of the Hartree equation be $\varphi_0^\Lambda(y)=\eta(\Lambda^{-1/3}y)$. 
		 Assume that the condensate is flat around the origin, namely for all $\beta\in\BN_+^3$ with $1\leq |\beta| \leq 2n-1$ we have
		\[ \eta(0)=1\,,\quad \partial^{\beta}\eta(0)=0\,.  \]

Furthermore, assume that there exists a family of unitarily implementable Bogoliubov maps $\cZ_0^\Lambda$ that approximates $\widecheck{\cT}$, defined in \eqref{def:T-infty-volume}, in the sense of \cref{con:Infinite-volume-Bog-maps} with growth rate $\epsilon$.

Let $\psi_0^\infty\in L^2(\BR^3,\cF(L^2(\BR^3)))$ and assume that there exist a family of states $\psi_0^\Lambda \in L^2(\BR^3,L^2_{\rm s}(\BR^{3N}))$  with $U_{\cZ_0^\Lambda}U_0^\Lambda\psi_0^\Lambda  \to  \psi_0^\infty $ in $L^2(\BR^3,\cF(L^2(\BR^3))$ as $\rho^\alpha=\Lambda\to \infty$, where  $U_t^\Lambda $ denotes the excitation map.

Then we have that for all times $T\geq0$,  
 \begin{align*}
 	\sup_{t\in [-T,T]}\left\Vert  \rme^{\rmi \nu_t^\Lambda}    U_{\widecheck{\cT}\cV_t^\infty\widecheck{\cT}^{-1}} U_{\cZ_0^\Lambda} U_{\cV_t^\Lambda}^* U_t^\Lambda  \rme^{-\rmi t H_\rho}\psi_0^\Lambda  - \rme^{-\rmi t H_\infty^{\rm BF}} \psi_0^\infty \right\Vert_{L^2(\BR^3, \cF(L^2))} \to 0
 \end{align*}
as $\rho^\alpha=\Lambda\to \infty$, where we set
 $\nu_t^\Lambda=\int^t_0(\rho^{1/2}\int W- \mu_s^\Lambda )ds\in\BR$, and $\mu_t^\Lambda\in\BR$ is given by \eqref{eq:Mu}. 
By $U_{\cV_t^\Lambda}$ we denote the propagator of $H^{\rm Bog}_\Lambda(t)$, and the diagonal Bogoliubov map $\widecheck{\cT}\cV_t^\infty\widecheck{\cT}^{-1}$ is given in \eqref{eq:Diagonal-V-infty-1}.
\end{theorem}
\begin{proof}[Proof of \cref{thm:Infty-Volume-Dynamics}]
	This theorem follows directly from more general \cref{thm:Infinite-Volume-Dynamics}, proven later in the article, and the fact that $\widecheck{\cT}$ diagonalizes the infinite-volume limit of $\cV^\Lambda_t$ (see \eqref{eq:Diagonal-V-infty-1}). Further details of the proof can be found in \cref{sec:Infinite-Volume-Dynamics}.
\end{proof}
\begin{rem}
 Note that $\widecheck{\cT}$ diagonalizes the limit of $\cV_t^\Lambda$ (see \eqref{eq:Diagonal-V-infty-1}), and that it is neither unitarily implementable nor bounded. The Bogoliubov maps $\widecheck{\cT}_t^\Lambda=  \widecheck{\cT}\cV_t^\infty\widecheck{\cT}^{-1} \cZ_0^\Lambda (\cV_t^\Lambda)^{-1}$ are used as a unitarily implementable approximation of $\widecheck{\cT}$. Since, formally $\widecheck{\cT}^{-1} \cZ_0^\Lambda$ converges to the identity and  $(\cV_t^\Lambda)^{-1}$ to $(\cV_t^\infty)^{-1}$, we also have the convergence of $\cT_t^\Lambda$ to the time-independent $\cT$. A rigorous proof of this convergence follows from the convergence of $\cZ_0^\Lambda$ in \cref{con:Infinite-volume-Bog-maps} together with the convergence of $\cV_t^\Lambda$ shown in \cref{lem:Approx-V-by-V-infty-Conv-Rates}.
		
		 The action of $\widecheck{\cT}_t^\Lambda$ can be interpreted in three steps: first, it removes the Bogoliubov time evolution from the dynamics; next it changes the reference state via $\cZ_0^\Lambda$; finally, it reintroduces the Bogoliubov dynamics, but now in the infinite-volume limit, through $\widecheck{\cT}\cV_t^\infty\widecheck{\cT}^{-1}$. 
\end{rem}

One can generalize the assumptions on the potential and the condensate in \cref{thm:Infty-Volume-Dynamics}. For the more general but technical conditions, see \cref{thm:Infinite-Volume-Dynamics}.

\paragraph{The Scaling and Comparison to the $\beta$-Scaling.}  
The considered scaling with $\rho^\alpha=\Lambda$, $\alpha>0$, introduced in \cite{DeNa14,PPS20} provides a framework for approaching the thermodynamic limit in the dense regime. The usual mean-field scaling is obtained at $\alpha=0$, while the thermodynamic limit corresponds to $\alpha=\infty$. In contrast, most of the literature focuses on the dilute regime, known as the $\beta$-scaling. 
 To compare with the standard $\beta$-scaling,  we place the system in a box of volume $\Lambda$ and rescale it to unit volume, yielding for the bosonic part of the system the Hamiltonian 
  \begin{align}
  		 - \sum\limits_{i=1}^{N} \frac{\Delta_{y_i}}{2N^{2\beta}} + N^{3\beta-1}\sum\limits_{1\leq i<j\leq N} V\big(N^{\beta}(y_i- y_j)\big)\,, \, \quad y_i\in [-1/2,1/2]^3\,, \label{eq:beta-scaling-comparison}
  \end{align}
  with $\beta=\frac{\alpha}{3(1+\alpha)} $.
  Thus, our scaling is not directly comparable to the $\beta$-scaling in the literature. While the interaction potential can be written in the same way,  \eqref{eq:beta-scaling-comparison} involves a semiclassically scaled Laplacian.

\subsection{Finite-Volume Dynamics} \label{sec:conditions}

While \cref{thm:Infty-Volume-Dynamics} assumes simplified conditions, we consider for $H^{\mathrm{BF}}_\Lambda (t)$ in a more general framework with weaker assumptions discussed below.

\subsubsection{Conditions}

\paragraph{Conditions on the Potentials.}
Both potentials may depend on $\Lambda$ and $\rho$  provided that all bounds below are satisfied uniformly. However, this dependence is not captured in our notation.

\begin{assumption}[Assumptions on the Potentials]\label{Assumption:Initial-datum-and-potential}
We denote the boson-boson interaction potential by $V\in L^1(\mathbb{R}^3, \mathbb{R})$. The potential $V$ is even and for all $k\in\BN_0$ there exists a constant $C>0$ such that for all densities $\rho\geq1$ and all volumes $ \Lambda\geq1$
\begin{align}
	\Vert \widehat{V}\Vert_1 +  \big\lVert\thinspace \lvert y \rvert^k V\,\big\rVert_{1}  + \big\lVert\thinspace \lvert y \rvert V\,\big\rVert_{\infty} \leq C \,.
	 \label{eq:Potential-estimates-1}
\end{align}
We denote the boson-impurity interaction potential by $W\in  L^\infty(\BR^3,\BR)$. The potential $W$ is even and $\forall  k\in\BN_0$ $\exists C>0$ such that $\forall \Lambda,\rho\geq1$
\begin{gather}
		\Vert  W\Vert_{H^2}  + \big\Vert\, |y|^k  W\, \big\Vert_{\infty} \leq C\,. \label{eq:Potential-estimates-2}
	\end{gather}
	We say $V$ and $W$ satisfy \cref{Assumption:Initial-datum-and-potential}$_{M}$ for $M\in\BN_+$ if the assumptions above are satisfied and $\exists C>0$ such that $\forall \Lambda,\rho\geq1$
	\begin{gather}
		 \Vert W\Vert_{W^{M,\infty}(\BR^3,\BR)} + \Vert (1+y^2-\Delta)^{1/4}  \partial^\beta W\Vert_{L^2(\BR^3,\BR)}\leq C \quad \text{for all }|\beta|\leq M\,.
	\end{gather}
\end{assumption}
\begin{rem} $\,$ \label{Remark:Initial-datum-and-potential}
\begin{itemize}
	\item \cref{Assumption:Initial-datum-and-potential}$_{1}$ for $M=1$ is assumed throughout the whole article without being mentioned explicitly. Whenever additional regularity of the boson–tracer interaction potential $W$ is required (i.e., $M >1$), we refer to this explicitly by \cref{Assumption:Initial-datum-and-potential}$_{M}$. Note that we use $W\in W^{1,\infty}\cap H^2$ for the well-posedness of the Bogoliubov-Fröhlich dynamics stated in \cref{rem:Bog-BF-Ham}.
	\item  There exists a decreasing function $k_0$ of $\alpha$, with $\alpha$ given by $\rho =\Lambda^{\alpha}$, such that it suffices to assume \eqref{eq:Potential-estimates-1} and \eqref{eq:Potential-estimates-2} only for all $k\leq k_0$.
\end{itemize}
\end{rem}

\paragraph{Conditions on the Condensate.} 
In contrast to \cref{thm:Infty-Volume-Dynamics} we do not need to require that the condensate wave function is a rescaled function varying on the scale $\cO(\Lambda^{1/3})$. 
In the following, we weaken this condition by only requiring that certain $L^p$-norms are consistent with the scaling behavior of a rescaled condensate.

\begin{condition} [Initial Condition of the Condensate]
  \label{con:Initial condition}
For all volumes $\Lambda\geq1$, let the condensate wave function $\varphi_0^\Lambda\in H^\infty(\BR^3)$.
We say $\varphi_0^\Lambda$ satisfies \cref{con:Initial condition} if there exists a constant $ C>0$ such that for all $\Lambda\geq1$ 
\begin{gather}
	    \left\Vert
       \widehat{\,\varphi_{0}^\Lambda\,}  \right\Vert_{1} \leq C\,, 
	\quad \Vert \varphi_0^\Lambda \Vert_{2} = \Lambda^{1/2} \,,  \label{eq:conditions-phi} \\
	\Vert  \widehat{\nabla\varphi_0^\Lambda} \Vert_{1} \leq C \Lambda^{-\frac{1}{3}} \,, \quad \Vert  \nabla \varphi_0^\Lambda \Vert_{2} \leq C \Lambda^{\frac{1}{2}-\frac{1}{3}} \,, \quad \Vert  \Delta\varphi_0 ^\Lambda\Vert_{2} \leq C \Lambda^{\frac{1}{2}-\frac{2}{3}}\,. \label{eq:Initial Derivative phi}
\end{gather}
\end{condition}
Next, we specify a condition that tracks the required bounds on the derivatives of the condensate.

\begin{condition}[Initial Condition for higher Derivatives of the Condensate]\label{con:Initial condition derivative condensate} 
For all volumes $\Lambda\geq1$, let the condensate $\varphi_0^\Lambda\in H^\infty(\BR^3)$.
\begin{itemize}
\item[i)]We say $\varphi_0^\Lambda$ satisfies \cref{con:Initial condition derivative condensate}i)$_{k}$ for a given $k\in \mathbb{N}_0$ if for all $\beta\in \mathbb{N}_0^3$ with $0\leq |\beta|\leq k$, there exists a constant $ C>0$ such that for all $\Lambda\geq1$ 
\begin{align}
	\Vert  \partial^{\beta}\varphi_0^\Lambda \Vert_{\infty} \leq C \Lambda^{-|\beta|/3} \,. \label{eq:DPhi0Infty} 
\end{align}
\item[ii)] We say $\varphi_0^\Lambda$ satisfies \cref{con:Initial condition derivative condensate}ii)$_{k}$ for a given $k\in \mathbb{N}_0$ if for all $\beta\in \mathbb{N}_0^3$ with $0\leq |\beta|\leq k$, there exists a constant $C>0$ such that for all $\Lambda\geq1$ 
\begin{align}
	\Vert  \partial^{\beta}\varphi_0^\Lambda \Vert_{2} \leq C \Lambda^{-|\beta|/3 +1/2} \,. \label{eq:DPhi0L2}
\end{align}
\end{itemize}
We say $\varphi_0^\Lambda$ satisfies \cref{con:Initial condition derivative condensate}$_{k}$ for a given $k\in \mathbb{N}_0$ if both \cref{con:Initial condition derivative condensate}i)$_{k}$ and \cref{con:Initial condition derivative condensate}ii)$_{k}$ are satisfied.
\end{condition}
The generalized flatness condition around the origin of the condensate is given below.

\begin{condition}[Condensate, Flat Around the Origin] \label{con:Condensate-flat-around-origin}
For all volumes $\Lambda\geq1$, let the condensate $\varphi_0^\Lambda\in H^\infty(\BR^3)$.
	We say $\varphi_0^\Lambda$ satisfies \cref{con:Condensate-flat-around-origin}$_{k,s}$ for a given $k\in \BN_+$ and $1/3>s>0$ if for all $0\leq |\beta|\leq k-1$, there exists a constant $ C>0$ such that for all $\Lambda\geq1$
	\begin{align}
		|\partial^{\beta}(\varphi_0^\Lambda(0)-1)|\leq C\Vert \partial^{\beta}(\varphi_0^\Lambda-1)\Vert_\infty \Lambda^{-(k-|\beta|)(1/3-s)}  \label{eq:Faltness-around-origin}
	\end{align}
\end{condition}

\begin{rem}\leavevmode
\begin{itemize}
	\item[i)] For \( \beta=0 \), the bound in \eqref{eq:Faltness-around-origin} gives us that the condensate is approximately equal to 1 at the origin, using 
 $\Vert \varphi_0^\Lambda\Vert_\infty\leq C$ (\cref{con:Initial condition}).
 If we consider \( \varphi_0^\Lambda \) as a rescaled function,  
    $\varphi_0^\Lambda(y) = \eta(\Lambda^{-1/3}y)$,
 then \cref{con:Condensate-flat-around-origin}$_{k,s}$ reduces to  
 \( \partial^{\beta}(\eta(0)-1) = 0 \), $\forall |\beta|\leq k-1$. 
 In general, the inequality \eqref{eq:Faltness-around-origin} implies that the derivatives of $\varphi_0^\Lambda$ are much smaller close to the origin than their supremum.
	\item[ii)] The flatness of the condensate plays a key role in controlling the dynamics of the tracer particle and has two main effects. First, the interaction of the tracer with condensate particles is subleading compared to its interaction with excitations (see \cref{sec:Control-tracer-condensate-mean-field-interaction} for details). Second, the tracer remains inside the Bose gas over the considered timescale of $\cO(1)$ (see \cref{rem:Flat-condensate} for details).
	\item[iii)]  If \cref{con:Condensate-flat-around-origin}$_{k,s}$ is satisfied for some $k\in\mathbb{N}_+$ then it also holds for all $\tilde{k}\leq k$. 
	Clearly, \cref{con:Condensate-flat-around-origin}$_{k,s}$ cannot be satisfied for negative values of $k$.
\end{itemize}
\end{rem}

In Appendix~\ref{sec:Control-of-the-condensate} we analyze how the listed conditions on the condensate are propagated by the Hartree evolution. In particular, if $\varphi_0^\Lambda (y)=  \eta ( \Lambda^{-1/3} y)$ then one can think of $\varphi_t^\Lambda(y)\approx \eta_t ( \Lambda^{-1/3} y)$ for some $\eta_t\in H^{\infty}$.

\paragraph{Conditions on the Tracer Particle and Excitations.} 

We now specify the initial conditions on $U_0^\Lambda \psi_0^\Lambda=:\psi_{\Lambda,0}^{\rm BF}$ for the excitation dynamics. 
For this we assume that the excitation number is of order one locally and essentially of order $\Lambda$ globally.
That is, only $\cO(1)$ many excitations effectively interact with the tracer. The control of the global excitation number is needed for the Bogoliubov approximation, \cref{thm:tildePsi-psiBF-time-estimate}, and the local one to prove tracer localization, \cref{cor:tracer-position-estimate-for-BF}, both are needed for the finite-volume approximation in \cref{thm:FullDynamics-BF-estimate}.

 	To achieve this we assume that  $\psi_{\Lambda,0}^{\rm BF}$ is a perturbation of a quasi-free state. That is, the Bogoliubov-transformed initial state $ U_{\cZ_0^\Lambda}\psi_{\Lambda,0}^{\rm BF}$ contains at most $\cO(1)$ many excitations, whereas $\psi_{\Lambda,0}^{\rm BF}$ itself may have  order $\Lambda$ excitations.

\begin{condition}[Initial Conditions for the Tracer Particle and the Excitation Number] \label{con:Localized-Tracer}
	For all densities $\rho\geq1$ and volumes $\Lambda\geq1$, let $\psi_0^\Lambda \in L^2(\BR^3,\mathcal{F}(L^2))$.
 		  We say $\psi_0^\Lambda $ satisfies \cref{con:Localized-Tracer}$_{\psi_0^\Lambda ,M,\epsilon}$ with power $M\in \BN_0$ and $\epsilon>0$ if there exist constants $ C_M,C>0$ such that for all densities $\rho\geq1$ and volumes $\Lambda\geq1$ there exists a unitarily implementable Bogoliubov map  $\cZ_0^\Lambda\in \cL(L^2\oplus J L^2)$ such that $\psi_0^\Lambda \in Q\big((-\Delta_x +x^2 +U_{\cZ_0^\Lambda}^*(\cN+1)^2U_{\cZ_0^\Lambda})^M\big)$ and
	\begin{align}
 		\Big\langle\psi_0^\Lambda , \left( (-\Delta_x +x^2)\otimes I + I\otimes U_{\cZ_0^\Lambda}^*(\cN+1)^2U_{\cZ_0^\Lambda}\right)^M \psi_0^\Lambda \Big\rangle  \leq C_{M}  \label{eq:(N+1)-x-Nabla-initial-condition-cor}
 \end{align}
  as well as the following bounds
 	\begin{gather}
 		\Vert \widehat{\cZ_0^\Lambda}\cT^{-1} (\tau\oplus J\tau J^*)\Vert_{\cL(L^2\oplus JL^2)} \leq C \,,  \label{eq:Z0-bounds-1} \\
 		\Vert \cZ_0^\Lambda \Vert_{\cL(L^2\oplus JL^2)}\leq C\Lambda^{\epsilon}\,, \quad \Vert \cZ_0^\Lambda(\cZ_0^\Lambda)^*-1\Vert_{\mathrm{HS}(L^2\oplus JL^2)}\leq C \Lambda^{1/2+\epsilon} \,. \label{eq:Z0-bounds-2}
 	\end{gather}

\end{condition}
\begin{rem}
	Due to elementary properties of the Bogoliubov transformation $U_{\cZ_0^\Lambda}$, it suffices to control only the transformed number of excitations $U_{\cZ_0^\Lambda}^*(\cN+1)U_{\cZ_0^\Lambda}$ in \cref{con:Localized-Tracer}. 
 		 	 Moreover, $U_{\cZ_0^\Lambda}$ can change the global excitation number by $\Vert \cZ_0^\Lambda(\cZ_0^\Lambda)^*-1\Vert_{\rm HS}^2 + \Vert \cZ_0^\Lambda\Vert_{\rm op}^2$ (see \cite[Lemma~4.4]{BPPS22}). The localization of the tracer, being a local phenomenon, only requires bounds on $\Vert \cZ_0^\Lambda\Vert_{\rm op}$, as well as the uniform bound in $\Lambda$ shown in \eqref{eq:Z0-bounds-1}, which is proven later in \cref{cor:tracer-position-estimate-for-BF}. Thus, we interpret  $\Vert \cZ_0^\Lambda(\cZ_0^\Lambda)^*-1\Vert_{\rm HS}^2\sim \cO(\Lambda^{1+2\epsilon})$ as the global excitation number, while the bounds on the operator norm ensure that the local excitation number is of $\cO(1)$.

 		All bounds \eqref{eq:Z0-bounds-1} and \eqref{eq:Z0-bounds-2} hold provided $\Vert \cZ_0^\Lambda\Vert_{\rm op}\leq C$. However, this condition does not appear to be directly compatible with the requirement in \cref{con:Infinite-volume-Bog-maps} that $\widehat{\cZ_0^\Lambda}\cT^{-1} (\tau\oplus J\tau J^*) \to  \widehat{\cZ_0^\infty} \cT^{-1} (\tau\oplus J\tau J^*)$ strongly in $L^2\oplus JL^2$, since the limiting operator, for instance, $\widehat{\cZ_0^\infty}=\cT$, is unbounded.
 		 A more careful construction of a $\cZ_0^\Lambda$ satisfying both the operator and Hilbert–Schmidt bounds, as well as the strong convergence, is given in \cref{sec:Construction-of-Z0}. Under the assumption of strong convergence, the uniform bound \eqref{eq:Z0-bounds-1} is automatically satisfied.
 		 The existence of a $\cZ_0^\Lambda$ satisfying strong convergence and all bounds above is a key ingredient in our proofs.
\end{rem}

\subsubsection{Finite-Volume Approximation}

We are now able to state our main result for the finite-volume $H^{\rm BF}_\Lambda(t)$.

 \begin{theorem}[Finite-Volume Bogoliubov-Fröhlich Dynamics] \label{thm:FullDynamics-BF-estimate}
 For given $0<\alpha<1/3$ and $0<s<1/3$ choose $\Lambda=\rho^{\alpha}$ and $n,k\in \BN_+$ large enough. Assume that the potentials $V$ and $W$ satisfy \cref{Assumption:Initial-datum-and-potential}$_{n}$. 
	\begin{itemize}
		\item[i)](Condensate conditions) Assume that for all $\rho\geq1$ the condensate $\varphi_0^\Lambda$ varies on the scale $\Lambda^{1/3}= \rho^{\alpha/3}$. That is, \cref{con:Initial condition} with additional regularity in the derivatives of $\varphi_0^\Lambda$, namely \cref{con:Initial condition derivative condensate}i)$_{k+2n-1}$ and \cref{con:Initial condition derivative condensate}$_{m}$ for $m=\max\{k,2\}+2$. Furthermore, we require that $\varphi_0^\Lambda$ is flat around the origin, namely \cref{con:Condensate-flat-around-origin}$_{2n,s}$, and assume that there exist $\eta\in H^1(\BR^3)$ and constants $\delta,C>0$ such that $\forall \Lambda\geq1$
\begin{align}
	\Vert \varphi_0^\Lambda(\Lambda^{1/3}\,.\,) -\eta \Vert_2 \leq  C\Lambda^{-\delta} \,.
\end{align}
		\item[ii)](Tracer localization and excitation number bound)
		  For all densities $\rho\geq1$ let $\psi_0^\Lambda \in L^2(\BR^3,H^1_{\rm sym}(\BR^{3N}))$ and set  $ \psi_{\Lambda,0}^{\mathrm{BF}}:=U_0^\Lambda\psi_0^\Lambda $, where $U_t^\Lambda $ is the excitation map from \eqref{eq:Excitation-map}.
		 \\Assume that $\psi_{\Lambda,0}^{\mathrm{BF}}$ is a perturbation of a quasi-free state with localized tracer particle, namely that  $\psi_{\Lambda,0}^{\mathrm{BF}}$ satisfies \cref{con:Localized-Tracer}$_{2n,\epsilon}$ with $\epsilon>0 $ small enough.
\end{itemize}
	
	\noindent Let $\psi_{\Lambda,t}^{\mathrm{BF}}\in L^2(\BR^3,\cF(\{\varphi_t^\Lambda\}^{\perp}))$ be the solution of the effective Bogoliubov-Fröhlich dynamics with initial data $\psi_{\Lambda,0}^{\mathrm{BF}}$ (see \cref{rem:Bog-BF-Ham}).
 Then for all times $T\geq0$, there exists a constant $C>0$ such that for all $\rho \geq1$ 
 \begin{align}
 	\sup_{t\in [-T,T]}\Vert  \rme^{\rmi \nu_t^\Lambda}  U_t^\Lambda  \rme^{-\rmi t H_\rho}\psi_0^\Lambda  -\psi_{\Lambda,t}^{\mathrm{BF}} \Vert \leq C \rho^{\frac{3(1+2\epsilon)\alpha -1}{2}} \,, \label{eq:FullDynmaics-BF-estimate}
 \end{align}
 where $\nu_t^\Lambda=\int^t_0(\rho^{1/2}\int W- \mu_s^\Lambda )ds\in\BR$  and $\mu_t^\Lambda\in\BR$ as in \eqref{eq:Mu}.
 \end{theorem}
 \begin{rem} \label{rem:FullDynamics-BF-estimate}
 We have explicit control on the lower bounds on $n$ and $k$ and upper bound on $\epsilon$ such that \cref{thm:FullDynamics-BF-estimate} is valid:
 		 \begin{gather}
	n\geq  \frac{3}{4(1/3-s)} \bigg( \frac{1}{\alpha} -2 -\frac{s}{3} \bigg)  \,, \quad
	k\geq \frac{(2n-1/2)(1/3 -s) }{ 1/3 + s }\,, \label{eq:tilde-n-and-k-large-enough-condition-main-thm} \\
	0<\epsilon< 1/4 \min\{ \delta,\, s,\, 3/2(1/3-s),\, 1/6,\, 2(1/(3\alpha)-1)\} \,, \label{eq:Condition-epsilon-small-enough}
\end{gather}
where  $\epsilon<1/2(1/(3\alpha)-1)$ ensures that \eqref{eq:FullDynmaics-BF-estimate} converges to zero.
 \end{rem}

 \begin{proof}[Proof of \cref{thm:FullDynamics-BF-estimate}]
 We define $\xi_t^\Lambda$ as the solution  of the intermediate Bogoliubov-Fröhlich dynamics by
  \begin{align} 
 \rmi \partial_t \xi_t^\Lambda &=\left( H^{\rm BF}_{\Lambda}(t)+ \rho^{1/2} W\ast |\varphi_t^\Lambda|^2(x) \right) \xi_t^\Lambda\,, \label{eq:Intermediate-Bog-Froehlich-Dynamics} \\
  \xi_{t=0}^\Lambda &= \psi_{\Lambda,0}^{\mathrm{BF}}\in L^2(\BR^3, Q( \rmd\Gamma(1-\Delta))) \,, \nn
 \end{align}
 which still includes the mean tracer-condensate interaction term. The well-posedness of \eqref{eq:Intermediate-Bog-Froehlich-Dynamics} as well as the invariance of the excitation space under this dynamics, follow by the same arguments as in \cref{rem:Bog-BF-Ham}.
 Then we split 
	\begin{align}
		&\Vert  \rme^{\rmi \int^t_0(\rho^{1/2}\int W- \mu_s^\Lambda )ds}  I\otimes U_t^\Lambda  \rme^{-\rmi H_\rho t}\psi_0^\Lambda -\psi_{\Lambda,t}^{\mathrm{BF}} \Vert \nn \\
		&\leq \Vert  \rme^{-\rmi \int^t_0 \mu_s^\Lambda ds}  I\otimes U_t^\Lambda  \rme^{-\rmi H_\rho t}\psi_0^\Lambda -\xi_t^\Lambda \Vert \label{eq:Bog-Approx}\\
		&\quad + \Vert  \rme^{\rmi t \rho^{1/2} \int W} \xi_t^\Lambda -\psi_{\Lambda,t}^{\mathrm{BF}} \Vert \,. \label{eq:Control-mean-field}
	\end{align}
The effective Hamiltonian  $H^{\rm BF}_{\Lambda}(t)+ \rho^{1/2} W\ast |\varphi_t^\Lambda|^2(x)$ is chosen so that the leading-order contributions generated by the two distinct dynamics in the first term \eqref{eq:Bog-Approx} cancel, thereby rendering it small.

 In the second term \eqref{eq:Control-mean-field} we extract the $\rho$-dependent mean tracer-condensate interaction term $\rho^{1/2} W\ast |\varphi_t^\Lambda|^2$ from the dynamics by approximating it with $\rho^{1/2} \int W$, which is large but constant. For this we use the results from \cref{sec:Control-tracer-condensate-mean-field-interaction}, especially the localization of the tracer particle, given in \cref{Part1:MainThm-new}$_{\gamma}$ (with $\gamma=1/2(3-1/\alpha)>0$ such that $\Lambda^{-\gamma}= (\Lambda^3/\rho)^{1/2}$).

	 Note that the conditions of \cref{thm:tildePsi-psiBF-time-estimate}$_\kappa$ are satisfied for $\kappa=2\epsilon$, since $\forall 1\leq m\leq4$:
	\begin{align*}
		&\pscal{\xi_0^\Lambda ,  (\cN +1 )^m \xi_0^\Lambda }  = \pscal{ U_{\cZ_0}\xi_0^\Lambda ,  (U_{\cZ_0}(\cN +1 )^m U_{\cZ_0}^*)  U_{\cZ_0}\xi_0^\Lambda } \\
		&\leq C_m (1+\Vert \cZ_0\cZ_0^*-1\Vert_{\rm HS}^2+ \Vert \cZ_0\Vert_{\rm op}^2)^m \pscal{ U_{\cZ_0}\xi_0^\Lambda  ,   (\cN +1 )^m   U_{\cZ_0}\xi_0^\Lambda } \\
		&\leq C_m(1+\Vert \cZ_0\cZ_0^*-1\Vert_{\rm HS}^2+ \Vert \cZ_0\Vert_{\rm op}^2)^m \leq C_m \Lambda^{(1+2\epsilon)m}\,,
	\end{align*}
	where we used \cite[Lemma~4.4]{BPPS22},
$\Vert \cZ_0^\Lambda\Vert_{\cL(L^2\oplus JL^2)}\leq C\Lambda^\epsilon$, $\Vert \cZ_0^\Lambda(\cZ_0^\Lambda)^*-1\Vert_{\mathrm{HS}(L^2\oplus JL^2)}\leq C \Lambda^{1/2+\epsilon}$ and \cref{con:Localized-Tracer}$_{\psi_{\Lambda,0}^{\mathrm{BF}},2n,\epsilon}$ together with $4n\geq\frac{6}{2(1/3-s)} \left( \frac{1}{\alpha} -2 -\frac{s}{3} \right)\geq 8\geq m$ for $\alpha,s\in(0,1/3)$.
\end{proof}

\subsubsection{Tracer Localization}
One of the important ingredients for \cref{thm:FullDynamics-BF-estimate} is the localization of the tracer particle, which is an interesting result on its own. For details about its derivation and interpretation we refer to \cref{sec:Tracer-localization}.

\begin{theorem}[Tracer Localization in Position and Momentum for the Effective Dynamics] \label{cor:tracer-position-estimate-for-BF}
Let $M\in\BN_+$. Assume that the potentials $V$ and $W$ satisfy \cref{Assumption:Initial-datum-and-potential}$_{\max\{M,2\}}$.
\begin{itemize}
	\item[i)](Condensate conditions) For all volumes $\Lambda\geq1$ let $\varphi_t^\Lambda$ be the solution of the Hartree equation \eqref{eq:Hartree}. Assume that its initial data $\varphi_0^\Lambda$  varies on the scale $\Lambda^{1/3}$, that is, 
\cref{con:Initial condition} and \cref{con:Initial condition derivative condensate}$_{4}$. Furthermore, we require that $\varphi_0^\Lambda$ is flat around the origin, namely \cref{con:Condensate-flat-around-origin}$_{2,s}$, $0<s<1/3$, and assume that there exist $\eta\in H^1(\BR^3)$ and constants $\delta,C>0$ such that for all $\Lambda\geq1$
\begin{align}
	\Vert \varphi_0(\Lambda^{1/3}\,.\,) -\eta \Vert_2 \leq  C\Lambda^{-\delta} \,.
\end{align}
	\item[ii)](Tracer localization) For all densities $\rho\geq1$ and volumes $\Lambda\geq1$ let $\psi_{\Lambda,0}^{\mathrm{BF}}\in L^2\big(\BR^3, Q( \rmd\Gamma(1-\Delta))\big)$. Assume that $\psi_{\Lambda,0}^{\mathrm{BF}}$ is a perturbation of a quasi-free state with localized tracer particle, namely that  $\psi_{\Lambda,0}^{\mathrm{BF}}$ satisfies \cref{con:Localized-Tracer}$_{M,\epsilon}$ with power $M$ and  $0<\epsilon\leq 1/4 \min\{ \delta,\, s,\, 3/2(1/3-s),\, 1/6\}$. 
\end{itemize}
Let $\psi_{\Lambda,t}^{\mathrm{BF}}$ be the solution of the effective Bogoliubov-Fröhlich dynamics with initial data $\psi_{\Lambda,0}^{\mathrm{BF}}$ (see \cref{rem:Bog-BF-Ham}).
	Then 
	\begin{itemize}
		\item[a)] $(t\mapsto \psi_{\Lambda,t}^{\mathrm{BF}}) \in \cL_{\rm loc}^\infty\big(\BR, Q((-\Delta_x)^M)\big)\cap\cL_{\rm loc}^\infty\big(\BR, Q( x^{2M})\big)$. 
		\item[b)] For all times $ T\geq0$ there exists a constant $C>0$ such that for all densities $\rho\geq1$ and volumes $ \Lambda\geq1$   
	\begin{align*}
		\sup_{t\in[-T,T]} \pscal{\psi_{\Lambda,t}^{\mathrm{BF}},\big((-\Delta_x)^M + x^{2M} \big) \psi_{\Lambda,t}^{\mathrm{BF}}}  \leq C  \,.
	\end{align*}
	\end{itemize}
\end{theorem}

The proof of \cref{cor:tracer-position-estimate-for-BF} is given in \cref{sec:Tracer-localization}.

\subsection{Notation} \label{sec:Notation}
Our notation is based on \cite{DFPP16,LaPi22,NNS16,LaTr25,LNS15}. Let $\sH,\sK$ be $\BC$--Hilbert spaces and $\sH^*$ the topological dual of $\sH$.
\begin{enumerate}
\item By $C$ we denote a universal constant, which is independent of our scaling parameters $\Lambda$ and $\rho$ and whose value may change from one line to another.
\item We denote $\BN_+:=\BN\setminus \{0\}$ and $\BN_0:=\BN\cup \{0\}$.
\item Let $E$ be a Banach space and $p\in[1,\infty]$. We denote by $\cL^p(\BR^d, E)$ the space of strongly measurable functions $f$ with $\Vert f\Vert_{\cL^p}<\infty$ and by $L^p(\BR^d, E)$ the corresponding $\cL^p$-space modulo functions vanishing almost everywhere. We set $\Vert\,.\,\Vert_p:=\Vert\,.\,\Vert_{L^p}$. By $L^p_s(\BR^{dn}, E)$ we denote the subspace of $L^p(\BR^{dn}, E)$ which is symmetric under exchange of the $n$ variables. We set  $L^p(\BR^{d}):=L^p(\BR^{d}, \BC) $.
\item For a Banach space $E$ we denote the Sobolev spaces by $W^{m,p}(\BR^d,E):=\{ f\in L^p(\BR^d, E) \,|\, \partial^\alpha f \in L^p(\BR^d, E)\ for \ |\alpha|\leq m \}$, $p\in[1,\infty],m\in \BN_0$, by $H^m= W^{m,2}(\BR^d,E)$ and by $H^\infty = \cap_{m\in \BN_0} H^m$. We set $H^m(\BR^d):= H^m(\BR^d,\BC)$.
\item Let $f\in L^2(\BR^d,\sH)$ and let $A$ be an operator on $\sH$. We denote the Fourier transform of $f$ and the Fourier-transformed operator associated with $A$ by
\begin{align*}
	\cF(f)(k) &=\hat{f}(k) =\frac{1}{(2\pi)^{d/2}} \int_{\BR^d} \rme^{-ikx} f(x) dx\,, \quad k\in \BR^d\,,
	\end{align*}
	and
	\begin{align*}
	\cF A\cF^{-1} &= \hat{A}\,.
\end{align*}
\item $J: \sH\to \sH^*:\psi\to \langle \psi,\,.\,\rangle$ denotes the canonical anti-unitary map between a Hilbert space and its dual.
\item  We denote by $\mathscr{H}\bar{\otimes}\mathscr{K}:=\text{lin}\{f\otimes g| f\in\sH,g\in\sK\}$ the algebraic tensor product and by $\mathscr{H}\otimes_s\mathscr{K}$ the symmetric tensor product of Hilbert spaces. 
For $\phi \in\sH^{\otimes k}$ and $\psi \in \sH^{\otimes l}$ we set
\begin{align*}
&\phi_k\otimes_s \psi_l(y_1,\dots,y_{k+l}) \\ 
&:=  \frac{1}{\sqrt{k!l!(k+l)!}} \sum_{\sigma\in P_{k+l}} \phi(y_{\sigma(1)},\dots , y_{\sigma(k)})  \psi(y_{\sigma(k+1)},\dots , y_{\sigma(k+l)})\,.
\end{align*}
\item For $\psi\in D(\cN^{1/2})$ and $f\in  L^2(\BR^3)$ we define the bosonic creation operator by $ a(f)^*\psi= f\otimes_s\psi$ and $a(f)$ as its adjoint. For $f,g\in L^2(\BR^3)$ we denote $A\left(  f \oplus J g \right):= a(f) + a^*(g)$.
\item We denote by $\mathcal{F}(\sH)= \bigoplus_{n=0}^\infty \sH^{\otimes_s n}$ the symmetric Fock space over $\sH$ and by $\mathcal{F}^{\leq N}(\sH)$ its truncation to at most $N$ particles. For $\psi\in \cF(\sH)$ we denote its $n$--th component by $\psi_n\in \sH^{\otimes_s n}$ and define the particle number operator $\cN\psi:= \sum_{n\geq0} n \psi_n$ on a suitable subspace of $\cF(\sH)$.
\item For a self-adjoint operator $A$ on $\sH$ we set $A_j= I\otimes...\otimes I\otimes A \otimes ...\otimes I$, where $A$ acts on the $j$--th space and set
	\begin{align*}
		 \rmd\Gamma(A)\psi ={}& \sum_{n\geq1} \sum_{1\leq j\leq n} A_j \psi_n\,, 
	\end{align*}
	on the domain of self-adjointness.
\item Let $A\geq -\beta$, for $\beta\in\BR$, be an operator on a Hilbert space then $q_A$ denotes the closed symmetric quadratic form associated to $A$ and $Q(A)$ its quadratic form domain (see \cite[Chapter VIII.6]{ReSi80}). We denote $q_A(\psi)= \pscal{\psi, A\psi}$.
\item For $f:\BR^3\to \BC$ we denote $f_x(y):=f(x-y)$.
\item For a map $A:x\mapsto A_x$ on $\BR^3$ whose values are quadratic forms we define for suitable $\psi\in Q(A)\subset L^2(\BR^3,\cF(\sH))$
	\begin{align*}
		\pscal{\psi,A\psi} ={}& \int \pscal{\psi(x),A_x\psi(x)} dx\,, \\
		Q(A) ={}&\Big\{ \psi\in L^2(\BR^3,\cF(\sH)) \Big| \, \big(x\mapsto \pscal{\psi(x),A_x\psi(x)}\big)\in L^1(\BR^3) \Big\} \,.
	\end{align*}	
	\item For $L^2(\BR^d, \mathcal{F}(L^2(\BR^d))$ or similar spaces we often use the notation $L^2(\BR^d_x, \mathcal{F}(L^2(\BR^d_y))$ to clearly separate the different arguments of the corresponding functions.
	\item We call a two-parameter family of operators $U(s,t)$, $s,t\in\BR$, on  $\sH$ a unitary propagator if satisfies $U(r,s)U(s,t)=U(r,t)$, $U(t,t)=I$ and $U(s,t)$ is jointly strongly continuous in $s$ and $t$.
	\item We denote by $\cL(\sH,\sK)$ the space of all bounded operators mapping from $\sH$ to $\sK$, and by $\mathrm{HS}(\sH,\sK)$  the space of all Hilbert-Schmidt operators. The corresponding norms are denoted by $\| \,.\,\|_{\cL(\sH,\sK)}$ and $\| \,.\,\|_{\mathrm{HS}(\sH,\sK)}$, respectively.
	\item We denote by $\ast$ the convolution:
	\[ (f\ast g )(x) = \int f(x-y) g(y) dy  \]
	for  measurable $f,g:\BR^3 \to \BC$ such that $\big(y\mapsto f(x-y) g(y)\big)\in L^1$ for almost all $x\in \BR^3$.
	\item We denote by $\cC$ the complex conjugation operator, $\cC\psi=\psi^*$, and by $R$ the reflection operator, $R\psi(y)= \psi(-y)$.
\end{enumerate}


\section{Bogoliubov Approximation} \label{sec:Bogoliubov-Approximation}
 In this section, we establish the validity of the intermediate Bogoliubov-Fröhlich Hamiltonian $H^{\rm BF}_{\Lambda}(t)+ \rho^{1/2} W\ast |\varphi_t^\Lambda|^2(x)$. To this end, we first isolate those contributions in the excitation Hamiltonian $H_\rho^{\rm ex}$ that constitute error terms (see \cref{sec:Determination-Error-Terms}). We then proceed with the Bogoliubov approximation in \cref{sec:Bogoliubov-Approximation-subsection}, where we rigorously justify the resulting intermediate Bogoliubov–Fröhlich Hamiltonian.

\subsection{Determination of the Error Terms} \label{sec:Determination-Error-Terms}

Starting from the excitation Hamiltonian \eqref{eq:Ex-Ham}, we derive the intermediate Bogoliubov-Fröhlich Hamiltonian $H^{\rm BF}_{\Lambda}(t)+ \rho^{1/2} W\ast |\varphi_t^\Lambda|^2(x)$, which still includes the mean tracer-condensate interaction, treated in \cref{sec:Control-tracer-condensate-mean-field-interaction}. For this, we isolate all terms in the excitation Hamiltonian that are small when the number of excitations is small compared to the total number of particles $N$. These terms are collected into an error term $R_N$. Additionally, we extract the constant $-\mu_t^\Lambda$ from the dynamics. The process of obtaining the effective Hamiltonian is referred to as the Bogoliubov approximation. The analysis in this section is based on results from \cite{LNS15,PPS20, LaPi22}.

Unless stated otherwise, we use a time-dependent orthonormal basis $\{u_n\}_{n\in\mathbb{N}_0}$ of the Hilbert space $L^2(\mathbb{R}^3)$ with $u_0(t):=\varphi_t^\Lambda/\Lambda^{1/2}$. We set $ V_{mnpq}= \pscal{u_m\otimes u_n, V(y-y') u_p\otimes u_q}$ and denote by $\cN_+^\Lambda(t)=\cN-a^*\left(\frac{\varphi_t^\Lambda}{\Lambda^{1/2}}\right) a\left(\frac{\varphi_t^\Lambda}{\Lambda^{1/2}}\right)$ the number operator on the excitation space $ \mathcal{F}\big(\{\varphi_t^\Lambda\}^\perp\big)$. Observe that $\cN_+^\Lambda(t)$ and $\cN$ coincide on the excitation space.

\begin{prop}[Intermediate Bogoliubov-Fröhlich Hamiltonian] \label{prop:HBF-TildeHN}
	  For all volumes $\Lambda\geq1$ let $\varphi_t^\Lambda$ be the solution of the Hartree equation \eqref{eq:Hartree}. We set $\mu_t^\Lambda:=\frac{1}{2} \pscal{\frac{\varphi_t^\Lambda}{\Lambda^{1/2}}, V\ast |\varphi_t^\Lambda|^2 \frac{\varphi_t^\Lambda}{\Lambda^{1/2}}}$ as in \eqref{eq:Mu}. Then
 \begin{align}
 	H_{\rho}^{\rm ex}(t) ={}&  H^{\rm BF}_{\Lambda}(t)+ \rho^{1/2} W\ast |\varphi_t^\Lambda|^2(x)  -\mu_t^\Lambda + R_N(t)  \label{eq:Connection-H-ex-and-H-BF} 
 \end{align}
and 
\begin{align}
	R_{1,N} ={}& - \frac{1}{2} \rmd\Gamma \left( Q_t^\Lambda \left[ V\ast |\varphi_t^\Lambda|^2  +K_1(t) - \mu_t^\Lambda \right] Q_t^\Lambda \right) \frac{\cN_+^\Lambda(t)}{N} \nonumber \\
	&- \frac{(\cN_+^\Lambda(t) +1)\sqrt{N-\cN_+^\Lambda(t)}}{N} a\left(Q_t^\Lambda V\ast |\varphi_t^\Lambda|^2 \frac{\varphi_t^\Lambda}{\Lambda^{1/2}}\right)  \nonumber \\
	&+  \frac{1}{2} \sum_{m,n\geq 1} \rho^{-1}{V}_{mn00}a_m^*a_n^* \left(  \sqrt{(N-\cN_+^\Lambda(t)-1)(N-\cN_+^\Lambda(t))} -N \right) \nonumber \\
	&+ \sum_{m,n,p\geq1} \Lambda{V}_{0mnp} \frac{\sqrt{N-\cN_+^\Lambda(t)}}{N} a_m^* a_n a_p  + \frac{1}{2} \mu_t^\Lambda \frac{\cN_+^\Lambda(t)}{N} + \mathrm{h.c.} \,, \label{eq:Def-R1N} \\
	R_{2,N} ={}& \frac{1}{2N} \sum_{m,n,p,q\geq 1} \Lambda{V}_{mnpq} a_m^* a_n^* a_p a_q  \,, \label{eq:Def-R2N} \\
	R_{3,N} ={}& a^*(Q_t^\Lambda W_x \varphi_t^\Lambda) \frac{1}{\sqrt{N}}\left( \sqrt{N-\cN_+^\Lambda(t)} -\sqrt{N}\right) + \mathrm{h.c.} \,, \label{eq:Def-R3N} \\
	R_{4,N} ={}& - \frac{1}{\sqrt{\rho}}  W\ast |\varphi_t^\Lambda|^2(x) \cN_+^\Lambda(t) + \frac{1}{\sqrt{\rho}} \rmd\Gamma(Q_t^\Lambda W_x Q_t^\Lambda) \,, \label{eq:Def-R4N} \\
	R_N ={}&\sum_{i=1}^4 R_{i,N}  \,, \label{eq:Def-RN}
\end{align} 
on the truncated excitation space $ \mathcal{F}^{\leq N}\big(\{\varphi_t^\Lambda\}^\perp\big)$ and $R_{i,N}=0$ elsewhere. 
\end{prop}
\begin{proof}[Proof of \cref{prop:HBF-TildeHN}]
The error terms in \eqref{eq:Connection-H-ex-and-H-BF} can be extracted by splitting the purely bosonic part of the Hamiltonian from the tracer-dependent part, and then calculating the action of the excitation map. For the purely bosonic Hamiltonian this was done in \cite[Appendix B and Lemma~6]{LNS15}. The proof can be adapted to our setting if one replaces the potential there with $\Lambda\cdot V$ and $u_t$ with $\frac{\varphi_t}{\Lambda^{1/2}}$. Note that $\Vert \frac{\varphi_t}{\Lambda^{1/2}}\Vert_2=1$ is normalized. This gives
\begin{align}
	& U_t^\Lambda  \Big( - \sum\limits_{i=1}^{N} \frac{\Delta_{y_i}}{2} + \frac{1}{\rho}\sum\limits_{1\leq i<j\leq N} V(y_i- y_j) \Big) (U_t^\Lambda )^* + \rmi  \left( \partial_t{ U}_{t}^\Lambda \right) (U_t^\Lambda)^* \nonumber \\
	&= H^{\rm Bog}_\Lambda(t) -\mu_t^\Lambda + R_{1,N} + R_{2,N} \,,
\end{align}
where $H^{\rm Bog}_\Lambda(t)$ is the Bogoliubov Hamiltonian defined in \eqref{eq:HBog-Definition}.
\\ The calculation of the tracer-dependent part of the Hamiltonian can be found in \cref{lem:Tech-UW_xU-splitting} below.
\end{proof}

Now, we follow the method displayed in \cite[Lemma 3.2]{LaPi22} to transform the tracer particle contributions to the microscopic Hamiltonian $H_\rho$ into the excitation space.

\begin{lemma} \label{lem:Tech-UW_xU-splitting}
For all volumes $\Lambda\geq1$ let $\varphi_t^\Lambda$ be the solution of the Hartree equation \eqref{eq:Hartree}.
Then
 \begin{align*}
 	&-\frac{\Delta_x}{2m}+ \frac{1}{\rho^{1/2}}U_t^\Lambda  \sum_{n=1}^{N} (W_x)_{n} (U_t^\Lambda )^*  \\
 	 &=\rho^{1/2} W\ast |\varphi_t^\Lambda|^2(x)  -\frac{\Delta_x}{2m} 
	+ A\left(  Q_t^\Lambda W_x \varphi_t^\Lambda\oplus  J Q_t^\Lambda W_x \varphi_t^\Lambda\right) + R_{3,N} +R_{4,N}\,,
 \end{align*}
 where $R_{3,N}$ and $R_{4,N}$ are defined in \cref{prop:HBF-TildeHN}.
\end{lemma}
\begin{proof}[Proof of \cref{lem:Tech-UW_xU-splitting}] 
 Our analysis is carried out for a fixed $x\in\BR^3$. We write $\sum_{n=1}^{N} (W_x)_{n}$ in its second quantized form to see that 
 \begin{align*}
	\sum_{n=1}^{N} (W_x)_{n} ={}&  \sum_{j,k=1}^\infty (W_x)_{jk} a_j^* a_k + \sum_{j=1}^\infty (W_x)_{0j} a_0^* a_j \\
	&+ \sum_{j=1}^\infty (W_x)_{j0} a_j^* a_0 + (W_x)_{00} a_0^* a_0 \,.
\end{align*}
We can transform creation and annihilation operators with the excitation map using \cite[Proposition~4.2]{LNSS15}. That is, $U_t^\Lambda = \bigoplus\limits_{k=0}^N Q_t^{\otimes k} \frac{\left(\Lambda^{-1/2}\varphi_t^\Lambda\right)^{N-k} }{\sqrt{(N-k)!}}$. With this we find
\begin{align*}
&U_t^\Lambda \sum_{j=1}^\infty (W_x)_{j0} a_j^*a_0 (U_t^\Lambda )^* =\sum_{j=1}^\infty (W_x)_{j0} a_j^* \sqrt{N-\cN_+} \\
&=  \sum_{j=1}^\infty \pscal{u_j,W_x  u_0} a^*(u_j) \sqrt{N-\cN_+^\Lambda} 
= a^*\left(Q_t^\Lambda W_x \frac{\varphi_t^\Lambda}{\Lambda^{1/2}} \right) \sqrt{N-\cN_+^\Lambda} 
\end{align*}
and
\begin{align*}
U_t^\Lambda (W_x)_{00} a_0^* a_0(U_t^\Lambda )^*&= (W_x)_{00} (N-\cN_+^\Lambda)=\frac{N-\cN_+^\Lambda}{\Lambda} W\ast |\varphi_t^\Lambda|^2(x)\,, \\
 U_t^\Lambda \sum_{j,k=1}^\infty (W_x)_{jk} a_j^* a_k (U_t^\Lambda )^* &= \sum_{j,k=1}^\infty (W_x)_{jk} a_j^* a_k = \rmd\Gamma(Q_t^\Lambda W_x Q_t^\Lambda)\,, 
\end{align*}
yielding the claim.
\end{proof}

\subsection{Error Estimation}  \label{sec:Bogoliubov-Approximation-subsection}

We aim to estimate $\Vert  \rme^{-\rmi \int^t_0 \mu_s^\Lambda ds} \psi^{\rm ex}_{\rho,t}- \xi_t^\Lambda \Vert$, where $\psi^{\rm ex}_{\rho,t} = U_t^\Lambda  \rme^{-\rmi H_\rho t}\psi_0^\Lambda $ and $\xi_t^\Lambda$ solves the equation $\rmi  \partial_t \xi_t^\Lambda = \left( H^{\rm BF}_{\Lambda}(t)+ \rho^{1/2} W\ast |\varphi_t^\Lambda|^2(x) \right) \xi_t^\Lambda$.
To this end, we adapt the methods of \cite{PPS20} to the second quantized setting, incorporating the interaction with the tracer particle. The key tool is a bound on the excitation number operator acting on the effective dynamics $\xi_t^\Lambda$, which controls the difference between the two dynamics.
The following lemma shows that the number of excitations grows at most as the volume.

\begin{lemma}[Excitation Number Estimate] \label{lem:Tech-particle-number-estimate} 
 For all volumes $\Lambda\geq1$, let $\varphi_0^\Lambda$ be the condensate satisfying \cref{con:Initial condition}.
 Then for all $n\in \mathbb{N}_0$ and times $ T\geq0$, there exists a constant $C>0$ such that for all densities $\rho\geq1$ and volumes $\Lambda\geq1$ we have
	\begin{align}
		\sup_{t\in [-T,T]}\pscal{\xi_t^\Lambda , I\otimes (\cN +1 )^n \xi_t^\Lambda }  \leq C \pscal{\xi_0^\Lambda , I\otimes (\Lambda+\cN +1 )^n \xi_0^\Lambda }  \,. \label{eq:Excitation-Number-Bound}
	\end{align}
\end{lemma}

\begin{rem}\label{rem:Tech-particle-number-estimate}
	The bound \eqref{eq:Excitation-Number-Bound} indicates that the excitation number grows with the volume $\Lambda$. 
	Indeed, by applying Duhamel we see that on short timescales
	\[\pscal{U^{\rm Bog}_{\Lambda,t}\Omega, \cN_+U^{\rm Bog}_{\Lambda,t}\Omega}\sim \Vert K_2^\Lambda(t)\Vert_{\rm HS}^2\sim \Lambda \Vert V\Vert_2^2 |\eta(0)|^2\,, \]
	where we set $\varphi_0^\Lambda(y)= \eta (\Lambda^{-1/3}y)$ and used \eqref{eq:K2-Def-2} for large volumes.
\end{rem}
The proof of \cref{lem:Tech-particle-number-estimate} can be found in Appendix~\ref{sec:Proof-lem:Tech-particle-number-estimate}. The growth in $\Lambda$ is determined by the estimates in \cref{lem:qqpp-term-estimate}.
The final Bogoliubov approximation is given in the following theorem.

\begin{theorem}[Bogoliubov Approximation]\label{thm:tildePsi-psiBF-time-estimate}
Assume that for all volumes $\Lambda\geq1$ the condensate $\varphi_0^\Lambda\in H^\infty (\BR^3)$ varies on the scale $\Lambda^{1/3}$, namely that it satisfies \cref{con:Initial condition}. For all  volumes $\Lambda\geq1$ and densities $\rho\geq1$ let $\psi_0^\Lambda \in L^2(\BR^3,H^1_{\rm sym}(\BR^{3N}))$ and  $ \xi_0^\Lambda:=U_0^\Lambda\psi_0^\Lambda \in H^1(\BR^3,\cF(L^2))\cap L^2(\BR^3, Q( \rmd\Gamma(1-\Delta))) $, where $U_t^\Lambda $ is the excitation map from \eqref{eq:Excitation-map}.
\\ If $\xi_0^\Lambda\in Q(\cN^4)$ and there exists constants $ C,\kappa\geq0$ such that for all densities $\rho\geq1$ and volumes $\Lambda\geq1$ and $1\leq n\leq 4$
 \begin{align}
 	\pscal{\xi_0^\Lambda , I\otimes (\cN +1 )^n \xi_0^\Lambda }  \leq C \Lambda^{n(1+\kappa)} \,, \label{eq:tildePsi-psiBF-initial-data}
 \end{align}
 then for all times $T\geq0$, there exists a constant $C>0$ such that for all densities $\rho\geq1$, volumes $\Lambda\geq1$ and $t\in[-T,T]$
 \begin{align}
 	 \Vert \rme^{-\rmi \int^t_0 \mu_s^\Lambda ds} U_t^\Lambda  \rme^{-\rmi H_\rho t}\psi_0^\Lambda  - \xi_t^\Lambda \Vert \leq C \left( \frac{\Lambda^{3(1+\kappa)}}{\rho} \right) ^{1/2} \bigg( 1+ \left(\frac{\Lambda^{(1+\kappa)}}{\rho} \right) ^{1/2} \bigg)  \,. \label{eq:Bog-Approximation}
 \end{align}
\end{theorem}

\begin{rem} \label{rem:tildePsi-psiBF-time-estimate}
If we choose $\Lambda = \rho^{\alpha}$ with $\alpha(1+\kappa)\leq 1$ then the right-hand side of \eqref{eq:Bog-Approximation} can be simplified. In this case, we have that $\Lambda^{1+\kappa}/\rho\leq 1$ and thus
	\[ \sup_{t\in[-T,T]} \Vert \rme^{-\rmi \int^t_0 \mu_s^\Lambda ds} U_t^\Lambda  \rme^{-\rmi H_\rho t}\psi_0^\Lambda  - \xi_t^\Lambda \Vert \leq C \left( \frac{\Lambda^{3(1+\kappa)}}{\rho} \right) ^{1/2}\,, \]
	which tends to zero for $\rho\to \infty$ if $\alpha<1/3$ and $\kappa<1/(3\alpha)-1$.
\end{rem}

\begin{proof}[Proof of \cref{thm:tildePsi-psiBF-time-estimate}]

	To shorten our notation we write $\rme^{\rmi \int^t_0 \mu_s^\Lambda ds}\xi_t^\Lambda =:\Phi_t$ and $U_t^\Lambda  \rme^{-\rmi H_\rho t}\psi_0^\Lambda  =:\psi^{\rm ex}_t$. For the proof we use a Grönwall estimate. We start by calculating the time derivative
	\begin{align*}
		\frac{d}{dt}  \Vert  \psi^{\rm ex}_t - \Phi_t \Vert^2  
		 &=2\mathrm{Im} \pscal{\psi^{\rm ex}_t , \left(H_{\rho}^{\rm ex }  -    H^{\rm BF}_{\Lambda}(t) -\rho^{1/2} W\ast |\varphi_t^\Lambda|^2(x)  + \mu_t^\Lambda \right) \Phi_t } \\
		&=2\mathrm{Im} \pscal{\psi^{\rm ex}_t, R_N \Phi_t } \,.
	\end{align*}
	We now use the estimate on the remainder term proven in the appendix (see \cref{lem:Tech-RN-estimate}) to bound the right-hand side by the particle number operator acting on the effective dynamics $\Phi_t$
	\begin{align}
		2\mathrm{Im} \pscal{\psi^{\rm ex}_t, R_N \Phi_t } &\leq  C \rho^{-1/2} \Vert \psi^{\rm ex}_t -\Phi_t \Vert  \Vert (\cN+1)^{\frac{3}{2}} \Phi_t \Vert \nn \\
		&\quad +C \rho^{-1} \Vert \psi^{\rm ex}_t -\Phi_t \Vert \Vert (\cN + 1)^2 \Phi_t \Vert \,. \label{eq:Remainder-Term-Estimate}
	\end{align}
	It follows with our initial condition \eqref{eq:tildePsi-psiBF-initial-data} and the excitation number estimates \cref{lem:Tech-particle-number-estimate} that $\forall T\geq0$ $\exists C>0$ such that $\forall \Lambda\geq1$ and $-T\leq t \leq T$
	\begin{align*}
		\frac{d}{dt}  \Vert \psi^{\rm ex}_t - \Phi_t \Vert^2 &\leq \Vert \psi^{\rm ex}_t -\Phi_t \Vert \Big\{C \rho^{-1/2} \pscal{\Phi_0, (\Lambda + \cN + 1)^3 \Phi_0 }^{1/2} \\
		&\quad +C \rho^{-1}  \pscal{\Phi_0, (\Lambda + \cN + 1)^4 \Phi_0 }^{1/2} \Big\} \\
		\leq{}& \Vert \psi^{\rm ex}_t -\Phi_t \Vert \Big\{ C \rho^{-1/2} \Lambda^{3/2(1+\kappa)} 
		+C \rho^{-1}  \Lambda^{2(1+\kappa)} \Big\} \,.
	\end{align*}
	The claim now follows with Grönwall's lemma.
\end{proof}


\section{Control of the Tracer-Condensate Interaction}\label{sec:Control-tracer-condensate-mean-field-interaction}

We started with the Hamiltonian $H_\rho$ and established the validity of $ H^{\rm BF}_\Lambda + \sqrt{\rho} W\ast |\varphi_t^\Lambda|^2(x)$ for the corresponding excitation dynamics in \cref{sec:Bogoliubov-Approximation}. 
In this section, we show that the mean tracer-condensate interaction $\sqrt{\rho} W\ast |\varphi_t^\Lambda|^2(x)$  is approximately constant, allowing us to remove it from the effective dynamics. As a result, the system is well described by $H^{\rm BF}_\Lambda(t)$.

\subsection{Mean-Field Interaction} \label{sec:Result-Control-Tracer-Condensate}

\paragraph{Overview of the Method.}
Since the mean tracer-condensate interaction term $\sqrt{\rho} W\ast |\varphi_t^\Lambda|^2(x)$ is of $\cO(\sqrt{\rho})$, it could potentially dominate the dynamics of the tracer particle and lead to the tracer leaving the Bose gas during times of order one (for details see \cref{rem:Flat-condensate}). However, our focus is on the interaction between the tracer and the excitations, so we choose a setting where $\sqrt{\rho} W\ast |\varphi_t^\Lambda|^2(x)$  does not outweigh the other interactions. In fact, for the tracer to be able to interact with excitations it has to stay inside the gas cloud.
To achieve this, we show that 
\begin{align}
	\sqrt{\rho} W\ast |\varphi_t^\Lambda|^2(x)\sim\sqrt{\rho} W\ast 1\sim \mathrm{constant} \,, \label{eq:aux-eq-17}
\end{align}
by using that the condensate remains flat around the position of the tracer particle, specifically $|\varphi_t^\Lambda|^2\sim 1$ if $|\varphi_0^\Lambda|^2\sim1$. 
 Note that we only need to control $|\varphi_t^\Lambda|^2$ within the range of interaction potential $W$ around the tracer particle's position.
\\The idea above is realized in two separate steps:
\begin{description}
	\item[1. Flatness of the Condensate.] Show that the condensate remains flat around the origin, ensuring that \eqref{eq:aux-eq-17} holds rigorously in this region.
	\item[2. Tracer Localization.] Show that the tracer is localized around the origin, assuming the condensate remains flat in this region.
\end{description}
Since we do not control the precise position of the tracer, we assume it is initially localized around a fixed point -- taken to be the origin -- and show that the condensate remains approximately flat around this point. Then, we demonstrate that the tracer remains localized around this fixed position over time.

\begin{rem}[Necessity of a Flat Condensate] \label{rem:Flat-condensate}
	Consider a rescaled condensate $\varphi_0^\Lambda(y) = \eta(\Lambda^{-1/3}y)$, which is not necessarily flat near the origin. Then the term $\sqrt{\rho}W\ast |\varphi_t^\Lambda|^2$ contributes to the tracer dynamics. Its influence can be estimated the following way 
	\begin{align*}
	\partial_t^2 \pscal{\psi_t, x \psi_t} =(2m)^{-1} \pscal{\psi_t, \rmi[H,-\rmi\nabla_x]\psi_t}\,.
\end{align*}
The contribution of the additional $\sqrt{\rho}W\ast |\varphi_t^\Lambda|^2(x)$ term in the Hamiltonian $H$ can be estimated by $\pscal{\psi_t, (\nabla\sqrt{\rho}W\ast |\varphi_t^\Lambda|^2) \psi_t} \leq C \sqrt{\rho} \Lambda^{-1/3}$. And thus for $\cO(1)$ times it can lead to a position change of $\cO(\sqrt{\rho} \Lambda^{-1/3})$.
 To ensure the tracer remains inside the condensate of volume $\cO(\Lambda)$, we obtain the constraint $\rho \ll \Lambda^{4/3}$. This conflicts with the condition $\Lambda^{3(1+ \kappa)} \ll \rho$ of the Bogoliubov approximation (\cref{thm:tildePsi-psiBF-time-estimate}).   Therefore, without flatness to suppress tracer energy gain, it is not expected  that the tracer remains inside the condensate on $\cO(1)$ timescales.
\end{rem}

Following the strategy outlined above, we derive the following estimate, proving convergence of the intermediate Bogoliubov-Fröhlich dynamics $\xi_t^\Lambda$ (see \eqref{eq:Intermediate-Bog-Froehlich-Dynamics}) to the Bogoliubov-Fröhlich dynamics  $\psi_{\Lambda,t}^{\mathrm{BF}}$ (\cref{rem:Bog-BF-Ham}) with extracted mean-field interaction of the tracer and condensate.

\begin{theorem}[Control of the Mean Tracer-Condensate Interaction] \label{Part1:MainThm-new}
Let  $\alpha>0$, $\gamma>0$, $1/3>s>0$ and $n,k\in\mathbb{N}_+$ with
\begin{align}
	n\geq \frac{1}{(1/3-s)} \left( \frac{1}{4\alpha} + \gamma - \frac{s}{4} \right)  \,, \quad
	k\geq \frac{(2n-1/2)(1/3 -s)  }{ 1/3 + s }\,. \label{eq:tilde-n-and-k-large-enough-condition}
\end{align}
Assume that the potentials $V$ and $W$ satisfy \cref{Assumption:Initial-datum-and-potential}$_{n}$. 
\begin{itemize}
		\item[i)](Condensate conditions) Assume that for all volumes $\Lambda\geq1$ the condensate $\varphi_0^\Lambda\in H^\infty (\BR^3)$ varies on the scale $\Lambda^{1/3}$. That is, \cref{con:Initial condition} with additional regularity in the derivatives of $\varphi_0^\Lambda$, namely \cref{con:Initial condition derivative condensate}i)$_{k+2n-1}$ and \cref{con:Initial condition derivative condensate}$_{m}$ for $m=\max\{k,2\}+2$. Furthermore, we require that $\varphi_0^\Lambda$ is flat around the origin, namely \cref{con:Condensate-flat-around-origin}$_{2n,s}$, and assume that $\exists \eta\in H^1(\BR^3)$, $\delta,C>0$ such that $\forall \Lambda\geq1$
\begin{align}
	\Vert \varphi_0^\Lambda(\Lambda^{1/3}\,.\,) -\eta \Vert_2 \leq  C\Lambda^{-\delta} \,.
\end{align}
		\item[ii)](Tracer Localization) For all densities $\rho\geq1$ and volumes $\Lambda\geq1$ let $\xi_0^\Lambda=\psi_{\Lambda,0}^{\mathrm{BF}}\in I\otimes Q( \rmd\Gamma(1-\Delta))$. Assume that $\psi_{\Lambda,0}^{\mathrm{BF}}$ is a perturbation of a quasi-free state with localized tracer particle, namely that $\psi_{\Lambda,0}^{\mathrm{BF}}$ satisfies  \cref{con:Localized-Tracer}$_{2n,\epsilon}$ with $0<\epsilon\leq 1/4 \min\{ \delta,\, s,\, 3/2(1/3-s),\, 1/6\}$. 
	\end{itemize}
 Then for all times $ T\geq0$, there exists a constant $C>0$ such that for all volumes $\Lambda\geq1$ and densities $\rho=\Lambda^{1/\alpha}$ we have
	\[ \sup_{t\in [-T,T]}\Vert   \rme^{\rmi t \rho^{1/2}\int W} \xi_t^\Lambda -\psi_{\Lambda,t}^{\mathrm{BF}} \Vert \leq C \Lambda^{-\gamma} \,.\]
\end{theorem}

\begin{proof}[Proof of \cref{Part1:MainThm-new}]
Applying Duhamel's formula we obtain the estimate
	\begin{align}
		&\Vert   \rme^{\rmi t \rho^{1/2}\int W} \xi_t^\Lambda -\psi_{\Lambda,t}^{\mathrm{BF}} \Vert^2 \nn \\
		&= \int_0^t 2 \mathrm{Im} \pscal{   \rme^{\rmi \tau \rho^{1/2}\int W} \xi_\tau^\Lambda  , \sqrt{\rho} W\ast (|\varphi_\tau|^2-1)(x) \psi_{\Lambda,\tau}^{\mathrm{BF}}   } d\tau \nn \\
		&\leq 2\int\limits_0^t \Vert \Theta_\Lambda \sqrt{\rho} W\ast (|\varphi_t|^2-1) \Vert_{\infty}  \Vert \xi_0^\Lambda \Vert \left\Vert \frac{ \psi_{\Lambda,\tau}^{\mathrm{BF}} }{\Theta_\Lambda(x)}  \right\Vert d\tau \,, \label{eq:psi-psi-s-estimate-step-2} 
	\end{align}
	where we have inserted the  function $\Theta_\Lambda (x)= \frac{1}{1+ (\Lambda^{-s} |x|)^{2n}}$, $0<s<1/3$, which localizes on a scale smaller than the $\cO(\Lambda^{1/3})$ scale of the condensate (see Appendix~\ref{sec:LocalizationFunction}).

	The estimation of the localized mean-field interaction $\Vert\tl \sqrt{\rho} W\ast (|\varphi_t^\Lambda|^2-1) \Vert_\infty$ corresponds to step 1 (flatness of the condensate) and can be found in \cref{cor:ThetaDPhi-TPhi}. The estimation of the effective dynamics weighted with moments of the tracer position operator $\left\Vert \frac{ \psi_{\Lambda,\tau}^{\mathrm{BF}} }{\Theta_\Lambda(x)}  \right\Vert$ corresponds to step 2 (tracer localization) and can be found in \cref{cor:tracer-position-estimate-for-BF}. We conclude
	\begin{align}
	\Vert   \rme^{\rmi t \rho^{1/2}\int W} \xi_t^\Lambda -\psi_{\Lambda,t}^{\mathrm{BF}} \Vert^2 \leq C \Big\{ \Lambda^{\frac{1}{2\alpha}-  \left(\frac{1}{6} +\frac{k}{3}+ ks \right)} + \Lambda^{\frac{1}{2\alpha} -\frac{s}{2}-2n(1/3-s)}   \Big\} \,, \label{eq:psi-psi-s-estimate-step-3} 
	\end{align}
	where the exponent on right-hand side of \eqref{eq:psi-psi-s-estimate-step-3} is smaller than $-\gamma$ due to the condition \eqref{eq:tilde-n-and-k-large-enough-condition} on $n$ and $k$.
\end{proof}

To control $ \Vert\tl \sqrt{\rho} W* (|\varphi_t^\Lambda|^2-1)  \Vert_\infty$ in the proof of \cref{Part1:MainThm-new}, we need  precise control over the localized condensate $\tl \varphi_t^\Lambda$ and that its flatness is preserved in time. The desired estimates on the condensate are proven in Appendix~\ref{sec:Localized-condensate}, leading to the following bound.

\begin{prop}[Local Control of the Mean Tracer-Condensate Interaction] \label{cor:ThetaDPhi-TPhi}
Let $\alpha>0,1/3>s>0$, $n,k\in \mathbb{N}_+$, and $\tl(x)=\frac{1}{1+(\Lambda^{-s}x)^{2n}}$ be the localization function. For all volumes $\Lambda\geq1$, let $\varphi_t^\Lambda$ be the solution of the Hartree equation \eqref{eq:Hartree}. Assume that its initial data $\varphi_0^\Lambda$ varies on the scale $\Lambda^{1/3}$. That is, \cref{con:Initial condition} with additional regularity in the derivatives of $\varphi_0^\Lambda$, given by \cref{con:Initial condition derivative condensate}i)$_{k+2n-1}$ and \cref{con:Initial condition derivative condensate}ii)$_{k+2}$. Furthermore, we require that $\varphi_0^\Lambda$ is flat around the origin, namely \cref{con:Condensate-flat-around-origin}$_{2n,s}$.
\\ Then for all times $T\geq0$, there exists a constant $C>0$ such that for all volumes $\Lambda\geq1$ and densities $\rho = \Lambda^{1/\alpha}$ we have 
	\begin{align}
		\sup_{t\in [-T,T]}\Vert \Theta_\Lambda \sqrt{\rho} W\ast (|\varphi_t^\Lambda|^2-1)\Vert_{\infty} \leq C \Big\{ \Lambda^{\frac{1}{2\alpha}-  \left(\frac{1}{6} +\frac{k}{3}+ ks \right)} + \Lambda^{\frac{1}{2\alpha} -\frac{s}{2}-2n(1/3-s)}   \Big\}  \,. \label{eq:EstimateThetaWLambdaLInfty} 
	\end{align}
\end{prop}

\begin{rem} By choosing $k$ and $n$ large enough, we can obtain an arbitrarily good convergence rate in \eqref{eq:EstimateThetaWLambdaLInfty}. 
\end{rem}

\begin{proof}[Proof of \cref{cor:ThetaDPhi-TPhi}] 
Let $T\geq0$ and $-T\leq t \leq T$. 
We start the estimate by splitting $\Vert \Theta_\Lambda \sqrt{\rho} W\ast (|\varphi_t^\Lambda|^2-1) \Vert_{\infty}$ into two parts
\begin{align}
	\Vert \Theta_\Lambda \sqrt{\rho} W\ast (|\varphi_t^\Lambda|^2-1) \Vert_{\infty} 
	&\leq \Vert \Theta_\Lambda \sqrt{\rho}  W\ast (\vert \varphi_0^\Lambda\vert^2 - 1)  \Vert_{\infty} \label{eq:ThetaWLambdaSplitting2Parts1} \\
	&\quad + \Vert \Theta_\Lambda \sqrt{\rho}  W\ast (\vert \varphi_t^\Lambda\vert^2 - \vert \varphi_0^\Lambda\vert^2)  \Vert_{\infty} \,. \label{eq:ThetaWLambdaSplitting2Parts}
\end{align}
 Since we assume that initially the condensate is flat around the origin (see \cref{con:Condensate-flat-around-origin}$_{2n,s}$), the term in \eqref{eq:ThetaWLambdaSplitting2Parts1} can be estimated straightforwardly. 
We use Taylor expansion of $\tl(x-y +y)$ at $x-y$ to move $\tl$ inside the convolution (see \cref{Lemma:CalculationTheta}b) and then estimate the potential $W$ with \cref{Assumption:Initial-datum-and-potential}. This yields, $\forall m\in\mathbb{N}_+$ and $\Lambda\geq1$
	\begin{align}
		&\left\vert \Theta_\Lambda W\ast (\vert \varphi_0^\Lambda\vert^2 - 1) \right\vert (x) \nn \\
		&\leq \sqrt{\rho} C_m \Vert ((\varphi_0^\Lambda)^*+1)\tl (\varphi_0^\Lambda-1)\Vert_\infty + \sqrt{\rho} \Lambda^{-sm} C_m ( \Vert\varphi_0^\Lambda\Vert_\infty^2 + 1) \nn \\ 
		&\leq C_{n} \sqrt{\rho} \Lambda^{-2n \left(\frac{1}{3} -s \right)  } \,, \label{eq:Theta-varhpi0-sup-estimate-1}
	\end{align}
	where we used $\Vert \varphi_0^\Lambda\Vert_\infty\leq C$, \cref{Lemma:CalculationTheta}a) (which requires flatness around the origin of the condensate) 
	and that we can choose $m$ large enough (for fixed $n, s>0$), such that $\Lambda^{-sm} \leq \Lambda^{-2n(\frac{1}{3} -s )}$.

Since $|\varphi_t^\Lambda\vert^2 - \vert \varphi_0^\Lambda\vert^2$ is missing a flatness condition around the origin, the estimate of  \eqref{eq:ThetaWLambdaSplitting2Parts} is considerably more delicate. However, using the estimate of the localized condensate  $\tl \varphi_t^\Lambda$ from Appendix~\ref{sec:Localized-condensate} we can complete the argument.
We again use \cref{Lemma:CalculationTheta}b) to move $\tl$ inside the convolution to conclude
\begin{align}
	\Vert \Theta_\Lambda \sqrt{\rho}  W\ast (\vert \varphi_t^\Lambda\vert^2 - \vert \varphi_0^\Lambda\vert^2)  \Vert_{\infty} 
	 &\leq
	 \sqrt{\rho} C
	 \Vert   \Theta_\Lambda(\vert \varphi_t^\Lambda\vert^2 - \vert \varphi_0^\Lambda\vert^2)  \Vert_{1\wedge 2 \wedge \infty}
	  \nonumber\\
	 &\quad + 
	 C \sqrt{\rho} \Lambda^{-ms} 
	 \Vert   \, \vert \varphi_t^\Lambda\vert^2 - \vert \varphi_0^\Lambda\vert^2 \Vert_{1\wedge 2 \wedge \infty} \nn
\end{align}
Note that the norm $\Vert\,.\,\Vert_{1\wedge 2 \wedge \infty}$ is defined in \cref{def:Young-norms}. We use \cref{cor:ThetaPhi-Phi0}, for $(\beta=0,\tilde{k}=k-1)$, and \cref{cor:phi-propagation} to get that $\forall m\in\BN_+$
\begin{align}
	\Vert \Theta_\Lambda \sqrt{\rho}  W\ast (\vert \varphi_t^\Lambda\vert^2 - \vert \varphi_0^\Lambda\vert^2)  \Vert_{\infty} 
	&\leq
	 \sqrt{\rho} C \Big\{ \Lambda^{-\frac{1}{6} -\frac{k}{3}- ks} + \Lambda^{-\frac{1}{2}s -2n(1/3-s)} \Big\} 
	 \nn \\
	&\quad +
	 C \sqrt{\rho} \Lambda^{-\left(\frac{1}{6}+ms\right)}  \,. \label{eq:PullingThetaInsideConvOfWConvPhiSquared-Phi0Squared}
\end{align}
Since $s>0$, we can choose $m$ in \eqref{eq:PullingThetaInsideConvOfWConvPhiSquared-Phi0Squared} large enough such that $\Lambda^{-\frac{1}{6}}\Lambda^{-ms}\leq \Lambda^{-s/2-2n(1/3-s)}$.
With this choice, \cref{cor:ThetaDPhi-TPhi} follows directly from \eqref{eq:ThetaWLambdaSplitting2Parts1}, \eqref{eq:ThetaWLambdaSplitting2Parts}, \eqref{eq:Theta-varhpi0-sup-estimate-1} and \eqref{eq:PullingThetaInsideConvOfWConvPhiSquared-Phi0Squared}.
\end{proof}

\subsection{Tracer Localization} \label{sec:Tracer-localization}
In this section we explain the idea behind the tracer localization in both position and momentum space in the dynamics generated by $H^{\rm BF}_\Lambda(t)$, given in \cref{cor:tracer-position-estimate-for-BF}. Its proof can be found at the end of this section.

Our goal is to show that the tracer remains confined within a region of $\cO(1)$ over timescales of $\cO(1)$.
For this, we have to control two contributions to its dynamics, one coming from the condensate $\varphi_t^\Lambda$ and one from the excitations.
\begin{description}
	\item[$\bullet$ Condensate Contribution.] This is minor, as the largest contribution that arises from its interaction with the tracer, $\sqrt{\rho} W\ast |\varphi_t^\Lambda|^2(x)$, has already been extracted from the dynamics.
	\item[$\bullet$ Excitation Contribution.] We will show that the tracer gains at most $\mathcal{O}(1)$ energy from its interactions with the excitations. This involves proving two key points:
\begin{description}
\item[1. The Number of Effective Interactions.] The tracer effectively interacts with only $\mathcal{O}(1)$ many excitations. Heuristically, this follows from the finite $\cO(1)$-range  of the tracer-boson interaction potential $W$, together with the assumption that the globally present excitations of essential order  $\Lambda$ (see \cref{rem:Tech-particle-number-estimate}) are roughly evenly distributed in the gas.
\item[2. The Energy Gain per Excitation.] The tracer particle gains at most $\mathcal{O}(1)$ energy from each excitation it interacts with, as controlled by the interaction prefactor in $H^{\rm BF}_\Lambda$. Rescaling the term $\frac{1}{\rho^r}\sum\limits_{i=1}^{N} W(x-y_i)$ in  $H_\rho$ yields an interaction  $\rho^{1/2-r}A(  Q_t^\Lambda W_x\varphi_t^\Lambda \oplus J Q_t^\Lambda W_x\varphi_t^\Lambda)$ in $H^{\rm BF}_\Lambda$. To ensure localization, this energy gain must remain $\mathcal{O}(1)$, requiring  $r \geq 1/2$. For $r < 1/2$, the interaction is too strong and localization may fail.
\end{description}
\end{description}

\paragraph{Number of Effectively Interacting Excitations with the Tracer.}

To establish localization of the tracer, we derive bounds on the moments of its position operator, $\langle \psi^{\mathrm{BF}}_{\Lambda,t}, |x|^{2M} \psi^{\mathrm{BF}}_{\Lambda,t}\rangle$, by using Grönwall-type arguments. For a simpler understanding, consider $\langle\psi^{\mathrm{BF}}_{\Lambda,t}, x_i \psi^{\mathrm{BF}}_{\Lambda,t}\rangle$.
Its time derivative involves the commutator $[x_i,H^{\rm BF}_\Lambda]=2\partial_{x_i}$, requiring bounds on $\partial_{x_i}$. Applying the same strategy to $\partial_{x_i}$, we encounter the commutator $[\partial_{x_i},H^{\rm BF}_\Lambda]=a(\partial_{x_i}Q_t^\Lambda W_x\varphi_t^\Lambda) + a^*(\partial_{x_i}Q_t^\Lambda W_x\varphi_t^\Lambda)\leq C (\cN+1)^{1/2}$. According to \cref{rem:Tech-particle-number-estimate},  $(\cN+1)^{1/2}$ is expected to scale as $\Lambda^{1/2}$, which does not suffice to conclude localization of  $x_i$ to an $\cO(1)$ region.

To improve the above estimate, we conjugate $H^{\rm BF}_\Lambda(t)$ with the propagator of the Bogoliubov Hamiltonian, $U_{\Lambda,t}^{\mathrm{Bog}}:=U^{\mathrm{Bog}}_\Lambda(t,0)=U_{\cV_t^\Lambda}$, which is itself a Bogoliubov transformation (see \eqref{eq:V-Matrix-Representation}).
By the conjugation with $U_{\cV_t^\Lambda}$ we extract the excitations that are effectively non-interacting with the tracer. The transformed Hamiltonian is
\begin{align}
	\widetilde{H}^{\mathrm{BF}}_\Lambda(t) ={}& I\otimes (U^{\rm Bog}_{\Lambda,t})^* H^{\rm BF}_\Lambda(t)I\otimes  U^{\rm Bog}_{\Lambda,t} + \rmi \left(I\otimes\partial_t U^{\rm Bog}_{\Lambda,t}\right) ^* I\otimes U^{\rm Bog}_{\Lambda,t} \nn \\
	 ={}& -\frac{\Delta_x}{2m} +A\left(\big(\cV_t^\Lambda)^{-1} (Q_t^\Lambda W_x\varphi_t^\Lambda \oplus J Q_t^\Lambda W_x\varphi_t^\Lambda\big)\right) \,. \label{eq:Transformed-BF-Hamiltonian}
\end{align}

In the transformed dynamics we get in the Grönwall argument for the excitation number operator  that
\[ A\left(\big(\cV_t^\Lambda)^{-1} (Q_t^\Lambda W_x\varphi_t^\Lambda \oplus J Q_t^\Lambda W_x\varphi_t^\Lambda\big)\right) \leq C \Vert(\cV_t^\Lambda)^{-1}\Vert_{\rm op} (\cN+1)^{1/2} \,. \]
And hence due to $\Vert\cV_t^\Lambda\Vert_{\rm op}\leq C$ (see \eqref{eq:Bog-Map-bounded}), $(\cV_t^\Lambda)^{-1}= S\cV_t^* S$ and Grönwall: 
\begin{align}
\pscal{\widetilde{\psi}_{\Lambda,t}^{\mathrm{BF}} ,(\cN+1) \widetilde{\psi}_{\Lambda,t}^{\mathrm{BF}}} \leq C \,,\quad t\in [-T,T] \label{eq:Ex-Number-Bound-Transformed-Dynamics}
\end{align}
 if we assume $\pscal{\widetilde{\psi}_{\Lambda,0}^{\mathrm{BF}} ,(\cN+1)\widetilde{\psi}_{\Lambda,0}^{\mathrm{BF}} } \leq C$.
 Although we have extracted a global number of excitations growing with the volume $\Lambda$ with the Bogoliubov transformation $U_{\Lambda,t}^{\rm Bog}$ coming from $H^{\rm Bog}_\Lambda(t)$ (see \cref{rem:Tech-particle-number-estimate}), the interaction term with the tracer particle has changed only by $\cO(1)$ due to $\Vert\cV_t^\Lambda\Vert_{\rm op}\leq C$. This indicates that the tracer effectively interacts with only $\cO(1)$ many excitations.
This way $\Vert\cV_t^\Lambda\Vert_{\rm op}$ controls the local quantity of effective interacting excitations with the tracer.

 In the transformed dynamics, the Grönwall estimate becomes:
\begin{align*}
\widetilde{H}^{\mathrm{BF}}_\Lambda: x_i &\overset{-\Delta_x}{\longrightarrow} &\partial_{x_i} &\overset{A((\cV_t^\Lambda)^{-1} ( Q_t^\Lambda W_x\varphi_t^\Lambda \oplus J Q_t^\Lambda W_x\varphi_t^\Lambda))}{\longrightarrow}& (\cN+1)^{1/2} && \quad \mathclap{\underset{\eqref{eq:Ex-Number-Bound-Transformed-Dynamics}}{\overset{\text{Estimate}}{\longrightarrow}}} \quad && \mathcal{O}(1)& \,, 
\end{align*}
ensuring localization.
Since  $I\otimes U_{\Lambda,t}^{\rm Bog}$ commutes with $x \otimes I$, bounds on tracer position in the transformed system $\widetilde{H}^{\mathrm{BF}}_\Lambda$ carry over to the original dynamics.
\\The bound on the excitation number in the transformed dynamics indicates that indeed the excitations are evenly distributed in the gas as claimed in our heuristic argument.

\paragraph{Generalization of the Initial Conditions.}
In addition to the Bogoliubov transformation $U^{\rm Bog}_{\Lambda,t}$, we can apply a time-independent Bogoliubov transformation, $U_{\cZ_0^\Lambda}$.
 This allows us to follow a similar argument as above while generalizing the initial state in \cref{cor:tracer-position-estimate-for-BF} from $\psi^{\mathrm{BF}}_{\Lambda,0}=\widetilde{\psi}^{\mathrm{BF}}_{\Lambda,0}$ to $\psi^{\mathrm{BF}}_{\Lambda,0}=I\otimes U_{\cZ_0^\Lambda}^* \widetilde{\psi}^{\mathrm{BF}}_{\Lambda,0} $ (see \eqref{eq:Transformation-prop-BF-Ham} below), where we only need to assume control of the excitation number on the transformed initial state, as in \cref{con:Localized-Tracer}. 

The resulting transformed Hamiltonian takes the form:
\begin{align}
 	\widetilde{H}^{\mathrm{BF}}_{\Lambda,\cZ_0^\Lambda}(t)&=  U_{\cZ_0^\Lambda}\left(  (U^{\rm Bog}_{\Lambda,t})^* H^{\rm BF}_\Lambda(t)   U^{\rm Bog}_{\Lambda,t} + \rmi \left(\partial_t U^{\rm Bog}_{\Lambda,t}\right) ^* U^{\rm Bog}_{\Lambda,t} \right)  U_{\cZ_0^\Lambda}^* \nn \\
 	&= -\frac{\Delta_x}{2m} +A\left(\cZ_0^\Lambda (\cV_t^\Lambda)^{-1} ( Q_t^\Lambda W_x\varphi_t^\Lambda \oplus J Q_t^\Lambda W_x\varphi_t^\Lambda)\right) \,. \label{eq:Transformed-BF-Hamiltonian-Z}
 \end{align}
The dynamics of the original state $\psi_t$ under $H^{\rm BF}_\Lambda(t)$ are equivalent to the transformed dynamics:
\begin{align}
		\rmi  \partial_t \psi_t ={}& H^{\rm BF}_\Lambda(t) \psi_t \quad  \Leftrightarrow \quad \rmi   \partial_t   U_{\cZ_0^\Lambda}(U^{\rm Bog}_{\Lambda,t})^*\psi_t = \widetilde{H}^{\mathrm{BF}}_{\Lambda,\cZ_0^\Lambda}(t)  U_{\cZ_0^\Lambda}(U^{\rm Bog}_{\Lambda,t})^*\psi_t \,. \label{eq:Transformation-prop-BF-Ham}
	\end{align}


\paragraph{Technical Implementation.}

For a rigorous version of the Grönwall estimate we use \cite[Theorem 8]{LNS15}, which was adapted to our setting in \cite[Theorem~D.1.1]{Spr25}.
Specifically, we apply \cite[Theorem 8]{LNS15} to the time-dependent quadratic form generated by $\widetilde{H}^{\mathrm{BF}}_{\Lambda,\cZ_0^\Lambda}(t)$ with comparison operator $A:=B:=h_{\rm oc}^M$, $M\in\BN_+$. Here, $h_{\rm oc}^M=(-\Delta_x +x^2 +(\cN+1)^2)^M$ is the well studied harmonic oscillator in the tracer position $x$, which controls $x^{2M}, (-\Delta_x)^{M}$ and $(\cN+1)^{2M}$.
The conditions of the theorem are verified in \cref{lem:Conditons-LNS15-Theorem-8} below.

\begin{lemma}  \label{lem:Conditons-LNS15-Theorem-8}
Let $M\in\BN_+$. Assume that the potentials $V$ and $W$ satisfy \cref{Assumption:Initial-datum-and-potential}$_{\max\{M,2\}}$. 
\begin{itemize}
	\item[i)](Condensate condition) For all volumes $\Lambda\geq1$, let $\varphi_t^\Lambda$ be the solution of the Hartree equation \eqref{eq:Hartree}. Assume that its initial data $\varphi_0^\Lambda$  varies on the scale $\Lambda^{1/3}$, that is, 
\cref{con:Initial condition} and \cref{con:Initial condition derivative condensate}$_{4}$. Furthermore, we require that $\varphi_0^\Lambda$ is flat around the origin, namely \cref{con:Condensate-flat-around-origin}$_{2,s}$, $0<s<1/3$, and assume that $\exists \eta\in H^1(\BR^3)$, $\delta$, $C>0$ such that $\forall \Lambda\geq1$
\begin{align}
	\Vert \varphi_0^\Lambda(\Lambda^{1/3}\,.\,) -\eta \Vert_2 \leq  C\Lambda^{-\delta} \,.
\end{align}
 	\item[ii)](Bogoliubov map) For all densities $\rho\geq1$ and volumes $\Lambda\geq1$ let $\cZ_0^\Lambda\in \cL(L^2\oplus JL^2)$ be a unitarily implementable Bogoliubov map such that $\exists C>0$ and $0<\epsilon\leq 1/4 \min\{ \delta,\, s,\, 3/2(1/3-s),\, 1/6\}$ with $\forall
 	\Lambda,\rho\geq1$
 	\begin{gather*}
 		\Vert \widehat{\cZ_0^\Lambda}\cT^{-1} (\tau \oplus J\tau J^*)\Vert_{\cL(L^2\oplus JL^2)} \leq C \,, \quad \Vert \cZ_0^\Lambda \Vert_{\cL(L^2\oplus JL^2)}\leq C\Lambda^{\epsilon} \,.
 	\end{gather*}
\end{itemize}
Then for all $T\geq0$, there exists a constant $C>0$ such that for all densities $\rho\geq1$, volumes $\Lambda\geq1$ and $-T\leq t\leq T$ we have 
\begin{itemize}
   		\item[a)] That 
   				\begin{align}
		C h_{\rm oc}^M \geq \widetilde{H}^{\mathrm{BF}}_{\Lambda,\cZ_0^\Lambda}(t)  \geq - C h_{\rm oc}^M  \label{eq:aux-eq-27}
 	\end{align}
   			 and the operator $h_{\rm oc}^M$ bounds the commutator of $h_{\rm oc}^M$ with $\widetilde{H}^{\mathrm{BF}}_{\Lambda,\cZ_0^\Lambda}(t)$, meaning 
   			\begin{align}
   				   \mp 2 \mathrm{Im} \,\pscal{ \widetilde{H}^{\mathrm{BF}}_{\Lambda,\cZ_0^\Lambda}(t) \psi, h_{\rm oc}^M \psi } \leq C \pscal{\psi, h_{\rm oc}^M\psi}  \,, \quad \forall \psi\in D(h_{\rm oc}^M)\,. \label{eq:LNS15-Thm-commutator-realtion}
   			\end{align} 
   		\end{itemize}
   		\item[b)] The time derivative of $\widetilde{H}^{\mathrm{BF}}_{\Lambda,\cZ_0^\Lambda}(t)$ is bounded by $h_{\rm oc}^M$, meaning $\forall \psi\in Q(h_{\rm oc}^M)$: $\big( t\mapsto \langle \psi, \widetilde{H}^{\mathrm{BF}}_{\Lambda,\cZ_0^\Lambda}(t)  \psi \rangle\big)\in C^1(\BR, \mathbb{R})$ and
   		\begin{align} 
   		\Big\vert \frac{d}{dt} \big\langle \psi, \widetilde{H}^{\mathrm{BF}}_{\Lambda,\cZ_0^\Lambda}(t)  \psi \big\rangle  \Big\vert \leq  C \pscal{\psi, h_{\rm oc}^M\psi}\,. \label{eq:aux-eq-28}
   		\end{align} 
   	
\end{lemma}
The proof of \cref{lem:Conditons-LNS15-Theorem-8} can be found in Appendix~\ref{sec:Proof-lem:Conditons-LNS15-Theorem-8}. 
Note that the bounds obtained in \cref{lem:Conditons-LNS15-Theorem-8} are uniform in $\Lambda$ and $\rho$, which allows us, together with the validity of \cite[Theorem 8]{LNS15}, to prove \cref{cor:tracer-position-estimate-for-BF} and thereby establish tracer localization.

\begin{proof}[Proof of \cref{cor:tracer-position-estimate-for-BF}]
	 Due to \cref{con:Localized-Tracer}$_{\psi_{\Lambda,0}^{\mathrm{BF}},M,\epsilon}$ there exists a unitarily implementable Bogoliubov map  $\cZ_0^\Lambda\in \cL(L^2\oplus JL^2)$ such that we have $\psi_{\Lambda,0}^{\mathrm{BF}}\in I\otimes U_{\cZ_0^\Lambda}^*Q(h_{\rm oc}^M)\cap Q( \rmd\Gamma(1-\Delta))\subset Q(h_{\rm oc})\cap  Q( \rmd\Gamma(1-\Delta))$ with
	\begin{align}
 		\pscal{\psi_{\Lambda,0}^{\mathrm{BF}},U_{\cZ_0^\Lambda}^*h_{\rm oc}^M U_{\cZ_0^\Lambda}\psi_{\Lambda,0}^{\mathrm{BF}}}  \leq C \,. \label{eq:(N+1)-x-Nabla-initial-condition-cor-2}
 \end{align}
 We aim to estimate the dynamics $\psi_{\Lambda,t}^{\mathrm{BF}}$. To this end, we apply two Bogoliubov transformations and define the transformed state by $\widetilde{\psi}_{\Lambda,t}^{\mathrm{BF}}= I\otimes U_{\cZ_0^\Lambda} U^{\rm Bog}_\Lambda(t,0)^*\psi_{\Lambda,t}^{\mathrm{BF}}$ with initial data $\widetilde{\psi}_{\Lambda,0}^{\mathrm{BF}}:=I\otimes U_{\cZ_0^\Lambda} \psi_{\Lambda,0}^{\mathrm{BF}}$, where $U^{\rm Bog}_\Lambda(t,0)$ denotes the propagator of the Bogoliubov Hamiltonian as discussed above. 
	In $\widetilde{\psi}_{\Lambda,t}^{\mathrm{BF}}$, we have extracted the dominant contribution to the excitation number coming from $U^{\mathrm{Bog}}_\Lambda(t,0)$. We can now apply the Grönwall-type estimate \cite[Theorem 8]{LNS15} to the transformed state $\widetilde{\psi}_{\Lambda,t}^{\mathrm{BF}}$ to control the growth of $h_{\mathrm{oc}}^M$ under the dynamics. Specifically, for $T\geq0$ there exists a $C>0$ such that for all densities $\rho\geq1$, volumes $\Lambda\geq1$ and $-T\leq t\leq T$
	\begin{align}
		\big\langle\widetilde{\psi}_{\Lambda,t}^{\mathrm{BF}}, h_{\rm oc}^M \widetilde{\psi}_{\Lambda,t}^{\mathrm{BF}}\big\rangle \leq C \big\langle\widetilde{\psi}_{\Lambda,0}^{\mathrm{BF}}, h_{\rm oc}^M \widetilde{\psi}_{\Lambda,0}^{\mathrm{BF}} \big\rangle \,. \label{eq:Harmonic-Osc-Estimate}
	\end{align}
	The verification of the conditions of \cite[Theorem 8]{LNS15} is carried out in \cref{lem:Conditons-LNS15-Theorem-8}. To transfer the estimate \eqref{eq:Harmonic-Osc-Estimate} back to the original dynamics $\psi_{\Lambda,t}^{\mathrm{BF}}$, we note that the operator $I \otimes U_{\cZ_0^\Lambda} U^{\mathrm{Bog}}_\Lambda(t,0)^*$ commutes with $x^{2M} \otimes I$ and $(-\Delta_x)^M \otimes I$, and that $(-\Delta_x)^M + x^{2M}\leq h_{\rm oc}^M$. Thus, using \eqref{eq:(N+1)-x-Nabla-initial-condition-cor-2}, we obtain
	\begin{align*}
		\pscal{\psi_{\Lambda,t}^{\mathrm{BF}}, \big(x^{2M} + (-\Delta_x)^M\big) \psi_{\Lambda,t}^{\mathrm{BF}}} \leq C \pscal{ \psi_{\Lambda,0}^{\mathrm{BF}}, U_{\cZ_0}^* h_{\rm oc}^M U_{\cZ_0^\Lambda} \psi_{\Lambda,0}^{\mathrm{BF}}} \leq C \,,
	\end{align*}
	which proves the claim.
\end{proof}

\section{Infinite-Volume Approximation} \label{sec:Infinite-Volume-Dynamics}

In this section, we establish the effective description of the dynamics through the infinite-volume Bogoliubov-Fröhlich Hamiltonian $H^{\rm BF}_\infty$. 
To this end, we first motivate the infinite-volume limit, which leads to the general result stated in \cref{thm:Infinite-Volume-Dynamics}.

To derive the infinite-volume limit description of the effective dynamics,
we restrict ourselves to the case where the interaction potentials $V$ and $W$ are independent of $\rho$ and $\Lambda$.

 As in the tracer localization argument of \cref{sec:Tracer-localization}, we have to extract the Bogoliubov Hamiltonian $H^{\mathrm{Bog}}_\Lambda(t)$ from $H^{\rm BF}_\Lambda(t)$, since it generates a divergent number of excitations that grows with the volume (see \cref{rem:Tech-particle-number-estimate}). This is done by conjugating with the propagator $U^{\rm Bog}_{\Lambda,t}=U_{\cV_t^\Lambda}$ of the Bogoliubov Hamiltonian.
 Although in general $U_{\cV_t^\Lambda}$ does not converge as $\Lambda\to\infty$, the corresponding Bogoliubov map $\cV_t^\Lambda$ admits a well-defined limit, introduced later in \eqref{eq:def-V-infty}.
 
Moreover, as in \cref{thm:FullDynamics-BF-estimate}, we consider initial data for the excitation dynamics of the form $\psi_{\Lambda,0}^{\mathrm{BF}}=U_{\cZ_0^\Lambda}^*\psi^{\Lambda}$, with $\pscal{\psi^{\Lambda}, (\cN+1)\psi^{\Lambda}}\leq C$. Here the Bogoliubov transformation $U_{\cZ_0^\Lambda}$ can create order $\Lambda^{1+2\epsilon}$ many excitations. Accordingly, we conjugate the Hamiltonian with $U_{\cZ_0^\Lambda}$ in order to extract these contributions from the initial state to the dynamics.

The resulting Hamiltonian $\widetilde{H}^{\mathrm{BF}}_{\Lambda,\cZ_0^\Lambda}$ relevant in the infinite-volume limit is given in \eqref{eq:Transformed-BF-Hamiltonian-Z}, precisely 
\begin{align}
	\widetilde{H}^{\mathrm{BF}}_{\Lambda, \cZ_0^\Lambda} =  -\frac{\Delta_x}{2m} + A\left( \cZ_0^\Lambda (\cV_t^\Lambda)^{-1} \left( Q_t^\Lambda W_x\varphi_t^\Lambda\oplus JQ_t^\Lambda W_x\varphi_t^\Lambda \right)\right) \,.
\end{align}
In order to define the infinite-volume limit of the Hamiltonian, we introduce the corresponding limits of all quantities appearing in $\widetilde{H}^{\mathrm{BF}}_{\Lambda,\cZ_0^\Lambda}$, while retaining their precise definitions for each theorem in which they are used.

For the initial data we assume that for $\Lambda\to \infty$ the Bogoliubov map $\cZ_0^\Lambda$ admits a limit $\cZ_0^\infty$ (possibly unbounded), which satisfies the symmetry condition $\cJ\cZ_0^\infty\cJ=\cZ_0^\infty$ with $\cJ=\begin{pmatrix}
 	&J^* \\
 	J&
 \end{pmatrix}$.
	 For the initial condensate, scaled to order one, we assume that  
	\begin{align}
	\varphi_{0}^\Lambda(\Lambda^{1/3}\,.\,) \to \eta 
	\end{align}
	in $L^2(\BR^3)$. In this setting, we define the limit of the constant $\mu_t^\Lambda$, appearing in the Hartree equation \eqref{eq:Hartree},  as
	\begin{align}
		\mu^\infty = \frac{1}{2} \int |\eta|^4 \int V\,.
	\end{align}
Further we make use of the flatness of the condensate around the origin in order to neglect the Laplacian in the Hartree equation \eqref{eq:Hartree}, leading us to the limit
\begin{align}
	\varphi_t^\Lambda \to \rme^{\rmi t \left(\int V -\mu^\infty\right)} \,. \label{eq:condensate-infinite-volume}
\end{align}

To define $\widetilde{H}^{\mathrm{BF}}_{\Lambda,\cZ_0^\Lambda}$ in the infinite-volume limit, we construct a candidate $ \cV_t^\infty$ for the limit of $ \rme^{\rmi t \left(\int V -\mu^\infty\right) S}\cV_t^\Lambda$, by taking the limit of its generator. 
Thus, we define $ \cV_t^\infty$ as the solution of the differential equation
\begin{align}
	\rmi \partial_t \cV_t^\infty = \cA^\infty(t) \cV_t^\infty \,, \quad \cV_0^\infty= I \,,	\label{eq:def-V-infty}
\end{align}
with generator
\begin{align}
		&\cA^\infty 
		= \begin{pmatrix}
		-\frac{\Delta}{2} + K_1^\infty & -K_2^\infty \\
		(K_2^\infty)^* & -J\big(-\frac{\Delta}{2} + K_1^\infty \big) J^*
\end{pmatrix} \,, \label{eq:def-A-infty} 	
	\end{align}
	where $K_1^\infty$ is the limit of $K_1^\Lambda(t)$ and $K_2^\infty$ the limit of $\rme^{2\rmi t \left(\int V -\mu^\infty\right)}K_2^\Lambda(t)$, given by
	\begin{align}
 K_1^\infty\psi &= V\ast \psi \,, \\
	K_2^\infty(t) J\psi &= V\ast \psi^*  \,,
 \end{align}
 for $\psi\in L^2(\BR^3)$. 
Equation \eqref{eq:def-V-infty} has a unique global solution, which can be proved similarly to \cite[Lemma 4.8]{BPPS22} or \cite[Lemma D.2.6]{Spr25}. Rigorous convergence results on $\mu_t^\Lambda$, $\varphi_t^\Lambda$ and $K_i^\Lambda(t)$ are collected in the appendix (see \cref{lem:Approx-V-by-V-infty-Conv-Rates-2}). There it is also shown that $\rme^{\rmi t \left(\int V -\mu^\infty\right) S}\cV_t^\Lambda \to \cV_t^\infty$ strongly in $L^2\oplus JL^2$ as $\Lambda\to\infty$ (see \cref{cor:Approx-V-by-V-infty}).

Motivated by the limits established above, we define the infinite-volume Hamiltonian:
\begin{align}
	H^\infty(t) :=  -\frac{\Delta_x}{2m} + A\left( \cZ_0^\infty (\cV_t^\infty)^{-1} ( W_x\oplus JW_x) \right) \,.
\end{align}
Since $\cZ_0^\infty$ and $(\cV_t^\infty)^{-1}$ are of the form $\begin{pmatrix}
	a &J^*bJ \\
	b &JaJ^*
\end{pmatrix}$ we know that $H^\infty(t)$ is symmetric. Moreover, by \cite[Theorem~8]{LNS15}, $H^\infty(t)$ generates a well-defined dynamics, such that the differential equation
	$
	\rmi \partial_t \psi_t^{\infty} ={}H^{\infty}(t) \psi_t^{\infty}$, $\psi_{t=0}^{\infty}= \psi_0^{\infty}
	$ has in a weak sense the unique global solution given by 
	$
	\psi_t^{\infty}={}U^{\infty}(t,0) \psi_0^{\infty}$, $\forall \psi_0^{\infty}\in L^2(\BR^3,\cF(L^2)),
	$
	where $U^{\infty}_t:=U^{\infty}(t,0)$ is the unitary propagator of $H^{\infty}(t)$.
The infinite-volume Hamiltonian $H^\infty(t)$ admits a considerably simpler time evolution than $\widetilde{H}^{\mathrm{BF}}	_{\Lambda,\cZ_0^\Lambda}(t)$. In fact, the condensate evolution governed by the Hartree equation reduces to a constant phase factor, while the time evolution of $\cV_t^\infty$ can be made explicit (see \eqref{eq:Diagonal-V-infty} below).

\paragraph{Diagonalization of $\cV_t^\infty$.} 
The diagonalization of the translation-invariant Bogoliubov map $\cV_t^\infty$ is well known in the literature, see for example \cite{BrDe07}.
To fix notation we repeat the results here.
In Fourier representation $\cA^\infty$ takes the form
	\begin{align}
	&\widehat{\cA^\infty}=(\cF\oplus J\cF J^* )\cA^\infty ( \cF^{-1}\oplus J\cF^{-1}J^*)\nn \\
	 &= \begin{pmatrix}
1 &0 \\
0 & J \cC R 
\end{pmatrix}   \begin{pmatrix}
		c & -b    \\
		b   & -c 
\end{pmatrix} \begin{pmatrix}
1 &0 \\
0 &  R \cC J^*
\end{pmatrix} \,, \label{eq:def-A-infty-fourier}
	\end{align}
	where $c(p)=\frac{p^2}{2}  + (2\pi)^{3/2}\widehat{V}(p)$, $b(p)=(2\pi)^{3/2} \widehat{V}(p)$, $\cC\psi=\psi^*$ and $R\psi(p)=\psi(-p)$. If $\widehat{V}\geq0$, we can explicitly diagonalize $\widehat{\cA^\infty}$: 
\begin{align}
\cT\widehat{\cA^\infty}\cT^{-1} = \begin{pmatrix}
	\omega &0 \\
	0 &-J\omega J^*
\end{pmatrix} \,,
\end{align}
where $\omega(p)=\sqrt{c^2-b^2}=\sqrt{p^4/4 + p^2 \widehat{V}(p) (2\pi)^{3/2}}$ and the Bogoliubov map $\cT$ is defined in \eqref{def:T-infty-volume}.
Since $\widehat{\cV_t^\infty}$ is generated by $\widehat{\cA^\infty}$ (see \eqref{eq:def-V-infty}), we find
\begin{align}
	\widehat{\cV_t^\infty}= \cT^{-1}  \begin{pmatrix}
	\rme^{-\rmi t \omega} & 0 \\
	0 & J \rme^{-\rmi t \omega}J^*
\end{pmatrix}	   \cT  \label{eq:Diagonal-V-infty} \\
	=  \begin{pmatrix}
	L(t) &  M(t)^* \cC R J^*\\
	J\cC R M(t) & J L(t)J^*
\end{pmatrix} \label{eq:-V-infty-Explicit}
\end{align}
with
\begin{align}
	L(t) &= \cos (\omega t ) -\rmi \frac{c}{\omega} \sin(\omega t)\,, \quad
	M(t) = -\rmi \frac{b}{\omega} \sin(\omega t) \,.
\end{align}

\paragraph*{Convergence to the Infinite Volume Dynamics.}

We are now able to prove the convergence of the dynamics generated by the transformed microscopic dynamics to the one generated by $H^\infty(t)$.

\begin{theorem} \label{thm:Infinite-Volume-Dynamics}
For given $\alpha,s \in (0,1/3)$ choose $\Lambda=\rho^{\alpha}\geq1$, $n,k\in \BN_+$ large enough and $\epsilon>0$ small enough (as in \cref{rem:FullDynamics-BF-estimate}). Assume that the potentials $V$ and $W$ satisfy \cref{Assumption:Initial-datum-and-potential}$_{n}$.  
For the initial condensate $\varphi_0^\Lambda$ we impose the conditions of \cref{thm:FullDynamics-BF-estimate}. 

Furthermore, assume that there exists a family of unitarily implementable Bogoliubov maps $\cZ_0^\Lambda$  that satisfies \cref{con:Infinite-volume-Bog-maps} with growth rate $\epsilon$ and limiting operator $\cZ_0^\infty$.

Let $\psi_0^\infty\in L^2(\BR^3,\cF(L^2(\BR^3)))$ and assume that there exist a family of states $\psi_0^\Lambda \in L^2(\BR^3,L^2_{\rm s}(\BR^{3N}))$  with $U_{\cZ_0^\Lambda}U_0^\Lambda\psi_0^\Lambda  \to  \psi_0^\infty $ in $L^2(\BR^3,\cF(L^2(\BR^3))$ as $\rho^\alpha=\Lambda\to \infty$, where  $U_t^\Lambda $ denotes the excitation map. 
Then, for all times $T\geq0$, 
 \begin{align*}
 	\sup_{t\in [-T,T]}\Vert  \rme^{\rmi \nu_t^\Lambda}  U_{\cZ_0^\Lambda} U_{\cV_t^\Lambda}^*U_t^\Lambda  \rme^{-\rmi t H_\rho}\psi_0^\Lambda  -U_t^\infty \psi_0^\infty \Vert_{L^2(\BR^3, \cF(L^2))} \to 0
 \end{align*}
as $\rho^\alpha=\Lambda\to \infty$, where $U_t^\infty$ denotes the unitary propagator of $H^\infty(t)$, $U_{\cV_t^\Lambda}$ the propagator of $H^{\rm Bog}_\Lambda(t)$, $\nu_t^\Lambda=\int^t_0(\rho^{1/2}\int W- \mu_s^\Lambda )ds\in\BR$ and $\mu_t^\Lambda\in\BR$ is given by \eqref{eq:Mu}.
\end{theorem}
\begin{proof}[Proof of \cref{thm:Infinite-Volume-Dynamics}] 
We first restrict ourselves to the case of regular initial data for the infinite-volume dynamics. Let $\tilde\psi_0^\infty\in Q((-\Delta_x+x^2+(\cN+1)^2)^{2n(1+4\alpha(1+2\epsilon))})\cap Q( \rmd\Gamma(-\Delta))$, where we choose $n\in \BN_+$ as in \cref{thm:FullDynamics-BF-estimate}. Let $\epsilon>0$ such that it satisfies \eqref{eq:Condition-epsilon-small-enough} and we have \eqref{eq:Z-Lambda-epsilon-bounds}.
For the initial data of the microscopic dynamics we choose a specific sequence 
\begin{align}
	\tilde\psi_0^\Lambda := (U_0^\Lambda)^* \mathbbm{1}^{\leq N} U_{\cZ_0^\Lambda}^* \Gamma(Q_0^\Lambda) \tilde\psi_0^\infty\,,
\end{align}
where $\mathbbm{1}^{\leq N}$ is the projection to the truncated Fock space with at most $N=\lceil\rho\Lambda\rceil$ particles. 
Note that, due to the conditions imposed on $U_{\cZ_0^\Lambda}^*$, we have $U_{\cZ_0^\Lambda}^* \Gamma(Q_0^\Lambda) \tilde\psi_0^\infty \in Q( \rmd\Gamma(-\Delta+1))\cap L^2(\BR^3,\cF(\{\varphi_0^\Lambda\}^\perp))$. Thus the action of $(U_0^\Lambda)^*$ in the definition of $ \tilde\psi_0^\Lambda $ is well defined, and $\tilde\psi_0^\Lambda \in L^2(\BR^3,H^1_{\rm s}(\BR^{3N}))$.

Applying \cref{thm:FullDynamics-BF-estimate}, we obtain
	\begin{align}
	\left\Vert   U_{\cZ_0^\Lambda} U_{\cV_t^\Lambda}^* \left(\rme^{\rmi \int^t_0(\rho^{1/2}\int W- \mu_s^\Lambda )ds} U_t^\Lambda \rme^{-\rmi t H_\rho}\tilde\psi_0^\Lambda  -\psi_{\Lambda,t}^{\mathrm{BF}} \right)\right\Vert \leq C \rho^{\frac{3(1+2\epsilon)\alpha -1}{2}}\,,  \label{eq:psi-ex-psi-BF-approx}
	\end{align}
	where $\psi_{\Lambda,t}^{\mathrm{BF}}$ is the solution of the effective Bogoliubov-Fröhlich dynamics with initial datum $\psi_{\Lambda,0}^{\mathrm{BF}}=U_0^\Lambda\tilde\psi_0^\Lambda $ (see \cref{rem:Bog-BF-Ham}). The conditions of \cref{thm:FullDynamics-BF-estimate} apply, since $U_0^\Lambda\tilde\psi_0^\Lambda $ satisfies \cref{con:Localized-Tracer}$_{2n}$, especially \eqref{eq:(N+1)-x-Nabla-initial-condition-cor}, as we now show. To obtain the bound \eqref{eq:Z0-bounds-1} on $\cZ_0^\Lambda$ required in \cref{con:Localized-Tracer} we use the convergence \eqref{eq:Z-Lambda-convergence-to-T-infty} together with the uniform boundedness principle.
	The tracer localization \eqref{eq:(N+1)-x-Nabla-initial-condition-cor} follows from combining the bound $\mathbbm{1}^{> N}\leq (\cN/N)^m$ for all $m\in\BN_+$, the estimate  $U_{\cZ_0^\Lambda}(\cN +1 )^m U_{\cZ_0^\Lambda}^*\leq C \Lambda^m(\cN +1 )^m$ (see \cite[Lemma~4.4]{BPPS22}), the strong convergence $\Gamma(Q_0^\Lambda)\to 1$ (see \cref{cor:Strong-convergence-infinite-volume}), and the fact that $\tilde\psi_0^\infty\in Q((-\Delta_x+x^2+(\cN+1)^2)^{2n(1+4\alpha(1+2\epsilon))})$. Moreover, in the same way, one shows that
	\begin{align}
		 U_{\cZ_0^\Lambda}U_0^\Lambda\tilde\psi_0^\Lambda  \to \tilde\psi_0^\infty\,,\quad \rho^\alpha=\Lambda \to \infty\,. \label{eq:Initial-data-limit}
	\end{align}
	 It remains to prove that
	\begin{align}
	\left\Vert U_{\cZ_0^\Lambda} U_{\cV_t^\Lambda}^* \psi_{\Lambda,t}^{\mathrm{BF}} - U^\infty_t \tilde\psi_0^\infty  \right\Vert_{L^2(\BR^3,\cF)} \to 0 \,, \label{eq:psi-BF-psi-infty-approx}
\end{align}
which would yield the claim for initial data $ \tilde\psi_0^\infty $ and $\tilde\psi_0^\Lambda $.
To simplify the notation we set $\phi_{t}^\Lambda= U_{\cZ_0^\Lambda} U_{\cV_t^\Lambda}^* \psi_{\Lambda,t}^{\mathrm{BF}}$. Then
	\begin{align}
		&\pm\partial_t\left\Vert \phi_{t}^\Lambda - \psi_t^\infty  \right\Vert_{L^2(\BR^3,\cF)}^2 = \pm 2\mathrm{Im} \pscal{\phi_{t}^\Lambda - \psi_t^\infty  , \big(\widetilde{H}^{\mathrm{BF}}_{\Lambda,\cZ_0^\Lambda}(t) -H^\infty(t)\big)  \phi_{t}^\Lambda  }  \nn \\
		&\leq  2   \Vert \phi_{t}^\Lambda - \psi_t^\infty  \Vert   \Big(\sup_x   (1+x^2)^{-1/2} \left\Vert  \cZ_0^\Lambda  F^{\Lambda}_x(t) - \cZ_0^\infty  F^\infty_x(t)\right\Vert_{L^2\oplus JL^2} \Big) \nn \\
		&\quad  \times \Vert (1+x^2)^{1/2}(\cN+1)^{1/2}  \phi_{t}^\Lambda \Vert  \,, \label{eq:Groenwall-phi-psi-infty}
	\end{align}
	where we introduced $F^\Lambda_x(t)= (\cV_t^\Lambda)^{-1} ( Q_t^\Lambda W_x\varphi_t^\Lambda \oplus JQ_t^\Lambda W_x\varphi_t^\Lambda) $ and $F^\infty_x(t)=   (\cV_t^\infty)^{-1} ( W_x\oplus JW_x)  $, which we localized with $(1+x^2)^{-1/2}$ around the origin. The term $\Vert  (1+x^2)^{1/2}(\cN+1)^{1/2}  \phi_t^\Lambda \Vert$ is controlled by our tracer localization argument, discussed in \cref{sec:Tracer-localization}. In fact with \eqref{eq:Harmonic-Osc-Estimate}  and \cref{con:Localized-Tracer}$_{\psi_{\Lambda,0}^{\mathrm{BF}},M=2}$, we get
	\begin{align}
		\Vert  (1+x^2)^{1/2}(\cN+1)^{1/2}  \phi_{t}^\Lambda \Vert \leq C \Vert h_{\rm oc}  U_{\cZ_0^\Lambda}\psi_{\Lambda,0}^{\mathrm{BF}} \Vert \leq C \,. \label{eq:Bound-initial-harmonic-osc}
	\end{align}
	We now show that
	\begin{align}
		\sup_{x\in\BR^3} (1+x^2)^{-1/2}\Vert \cZ_0^\Lambda   F^{\Lambda}_x(t) - \cZ_0^\infty  F^\infty_x(t)\Vert_{L^2\oplus JL^2} \to 0 \,. \label{eq:F-Lambda-to-F-infty}
	\end{align}
	Let $x\in\BR^3$ fixed. Then we use that $\Vert \cZ_0^\Lambda\Vert_{\rm op}\leq C \Lambda^\epsilon$ to get
	\begin{align}
		&(1+x^2)^{-1/2}\big\Vert \cZ_0^\Lambda   F^{\Lambda}_x(t) - \cZ_0^\infty  F^\infty_x(t)  \big\Vert \nn \\
		&\leq  C \Lambda^\epsilon   (1+x^2)^{-1/2}\big\Vert   F^\Lambda_x(t)  - F^\infty_x(t)   \big\Vert \label{eq:aux-eq-108} \\
		&\quad+ (1+x^2)^{-1/2}\big\Vert  \big( \cZ_0^\Lambda - \cZ_0^\infty \big) (\cV_t^\infty)^{-1} ( W_x\oplus JW_x)  \big\Vert \,.  \label{eq:aux-eq-109}
	\end{align}
	We conclude the convergence \eqref{eq:aux-eq-108}$\to 0$ as $\Lambda\to \infty$, from \cref{lem:Approx-V-by-V-infty-on-specific-state}. The convergence of \eqref{eq:aux-eq-109} is proven below. 
	We use the diagonalization of $\widehat{\cV_t^\infty}$ in \eqref{eq:Diagonal-V-infty} to obtain
	\begin{align}
		&\eqref{eq:aux-eq-109}=(1+x^2)^{-1/2} \Big\Vert  \big( \widehat{\cZ_0^\Lambda} - \widehat{\cZ_0^\infty}  \big)  \cT^{-1}  \begin{pmatrix}
	\rme^{\rmi t \omega} & 0 \\
	0 & J\rme^{\rmi t \omega}J^*
\end{pmatrix}	   \cT\big( \widehat{W_x}\oplus J\widehat{W_x}\big)  \Big\Vert \nn \\
	&= (1+x^2)^{-1/2}\Big\Vert  \big( \widehat{\cZ_0^\Lambda} - \widehat{\cZ_0^\infty} \big)\cT^{-1} \big( \rme^{\rmi t \omega} \tau \widehat{W_x}\oplus J \cC R e^{-\rmi t \omega} \tau \widehat{W_x} \big)  \Big\Vert \,. \label{eq:aux-eq-110}
	\end{align}
	Now we use $\widehat{W_x}(p)= e^{\rmi px}\widehat{W}(p)$ to get, since  $e^{\rmi px}$ is unitary,
	\begin{align}
		\eqref{eq:aux-eq-110} &\leq (1+x^2)^{-1/2}  \Big\Vert  \left[\widehat{\cZ_0^\Lambda} - \widehat{\cZ_0^\infty}, \rme^{\rmi px}\oplus J\rme^{\rmi px}J^*\right]\cT^{-1} \nn \\
		&\quad \times \big( \rme^{\rmi t \omega} \tau \widehat{W_x}\oplus J \cC R \rme^{-\rmi t \omega} \tau \widehat{W_x} \big)  \Big\Vert \label{eq:aux-eq-111} \\
		&\quad + \Big\Vert  \big( \widehat{\cZ_0^\Lambda} - \widehat{\cZ_0^\infty} \big) \cT^{-1} \big( \rme^{\rmi t \omega} \tau \widehat{W}\oplus J \cC R e^{-\rmi t \omega} \tau \widehat{W} \big)  \Big\Vert \,, \label{eq:aux-eq-112}
	\end{align}
	where \eqref{eq:aux-eq-111} and \eqref{eq:aux-eq-112} converge to zero due to convergence of the commutator of $\cZ_0^\Lambda - \cZ_0^\infty$ with translations and the convergence of $\cZ_0^\Lambda \to \cZ_0^\infty$ (see \cref{con:Infinite-volume-Bog-maps}).
	And thus we have the full convergence \eqref{eq:F-Lambda-to-F-infty}.

	We conclude using Grönwall from \eqref{eq:Groenwall-phi-psi-infty} and \eqref{eq:Bound-initial-harmonic-osc} that 
	\begin{align}
		&\left\Vert \phi_{t}^\Lambda - \psi_t^\infty  \right\Vert \leq \left\Vert \phi_{0}^\Lambda - \tilde\psi_0^\infty  \right\Vert  \nn \\
		&\quad +C  \int_0^t  \sup_{x\in\BR^3} (1+x^2)^{-1/2} \Vert \cZ_0^\Lambda F^{\Lambda}_x(s) -  \cZ_0^\infty F^\infty_x(s)\Vert_{L^2\oplus JL^2}  \,ds  \,. \label{eq:Phi-psi-infty-approx-2}
	\end{align}
	Then \eqref{eq:psi-BF-psi-infty-approx} follows from \eqref{eq:F-Lambda-to-F-infty}, dominated convergence, and $\phi_{0}^\Lambda \to \tilde\psi^\infty_0$ in \eqref{eq:Initial-data-limit}. Indeed, for $s\in[0,t]$ we have the integrable majorant $\Vert F^{\Lambda}_x(s) - F^\infty_x(s)\Vert_{L^2\oplus JL^2}\leq C_t\in L^1((0,t))$ with $C_t$ independent of $s$ and $\Lambda$, which follows directly from \cref{lem:uniform-boundedness-F-Lambda} and the fact that $\widehat{\cZ^\infty_0} \cT^{-1}$ is bounded. 
	
	Thus as a direct consequence of \eqref{eq:psi-ex-psi-BF-approx} and \eqref{eq:psi-BF-psi-infty-approx}, we get
	 \begin{align}
 	\Vert  \rme^{\rmi \int^t_0(\rho^{1/2}\int W- \mu_s^\Lambda )ds}  U_{\cZ_0^\Lambda} U_{\cV_t^\Lambda}^*U_t^\Lambda  \rme^{-\rmi t H_\rho}\tilde\psi_0^\Lambda  -U_t^\infty \tilde\psi_0^\infty \Vert \to 0 \,, \label{eq:Psi-psi-infty-tilde}
 \end{align}
 which proves the claim for the specific initial data $ \tilde\psi_0^\infty $ and $\tilde\psi_0^\Lambda $.
	
		Now, let $\tilde\epsilon>0$ and consider general initial data $\psi_0^\infty$ and $\psi_0^\Lambda $ with $U_{\cZ_0^\Lambda} U_0^\Lambda\psi_0^\Lambda  \to \psi_0^\infty$. By density, there exists a $\tilde\psi_0^\infty\in Q((-\Delta_x+x^2+(\cN+1)^2)^{2n(1+4\alpha(1+2\epsilon))})\cap Q( \rmd\Gamma(-\Delta))$ such that
	$\Vert 	\psi_0^\infty - \tilde\psi_0^\infty\Vert<\tilde\epsilon$. As before, we set   $\tilde\psi_0^\Lambda =(U_0^\Lambda)^* \mathbbm{1}^{\leq N} U_{\cZ_0^\Lambda}^* \Gamma(Q_0^\Lambda) \tilde\psi_0^\infty$. 
	 Then from
	 \begin{align*}
	 	& \Vert \tilde\psi_0^\Lambda  - \psi_0^\Lambda   \Vert \leq \Vert \Gamma(Q_0^\Lambda) \tilde\psi_0^\infty - U_{\cZ_0^\Lambda} U_0^\Lambda \psi_0^\Lambda  \Vert   + \Vert \mathbbm{1}^{>N} U_{\cZ_0^\Lambda}^* \Gamma(Q_0^\Lambda) \tilde\psi_0^\infty \Vert  \\
	 	& \leq \Vert \Gamma(Q_0^\Lambda) \tilde\psi_0^\infty -  \tilde\psi_0^\infty \Vert + \Vert\psi_0^\infty - U_{\cZ_0^\Lambda} U_0^\Lambda \psi_0^\Lambda  \Vert  + \Vert \tilde\psi_0^\infty -\psi_0^\infty \Vert \\
	 	 &\quad + \Vert \mathbbm{1}^{>N} U_{\cZ_0^\Lambda}^* \Gamma(Q_0^\Lambda) \tilde\psi_0^\infty \Vert \,,
	 \end{align*}
	the strong convergence $\Gamma(Q_0^\Lambda)\to 1$,  $U_{\cZ_0^\Lambda} U_0^\Lambda \psi_0^\Lambda  \to \psi_0^\infty$, and $\mathbbm{1}^{>N}\leq \cN/N$, it follows that $ \Vert \tilde\psi_0^\Lambda  - \psi_0^\Lambda  \Vert<2\tilde\epsilon$ for large enough $\Lambda$. Combining this with \eqref{eq:Psi-psi-infty-tilde}, we obtain for large enough $\Lambda$
\begin{align*} 
	\Vert  \rme^{\rmi \int^t_0(\rho^{1/2}\int W- \mu_s )ds}  U_{\cZ_0^\Lambda} U_{\cV_t^\Lambda}^*U_t^\Lambda  \rme^{-\rmi t H_\rho}\psi_0^\Lambda  -U_t^\infty \psi_0^\infty \Vert 	&\leq \epsilon + \Vert \tilde\psi_0^\Lambda  - \psi_0^\Lambda  \Vert \\
	&\quad +\Vert \tilde\psi_0^\infty -\psi_0^\infty \Vert \leq 4\tilde\epsilon  \,,
\end{align*}
which proves the claim.
\end{proof}

An  explicit example of $\cZ_0^\Lambda$ satisfying \cref{con:Infinite-volume-Bog-maps} with  $\cZ_0^\infty=  \widecheck{\cT} $ is constructed in \cref{sec:Construction-of-Z0}, under the additional assumption that the initial condensate $\varphi_0^\Lambda$ is real-valued.

In the case that $\cZ_0^\infty=  \widecheck{\cT} $  we observe, using \eqref{eq:Diagonal-V-infty}, that the infinite-volume Bogoliubov-Fröhlich Hamiltonian defined in \eqref{eq:H-BF-Infty-Volume} is equal to
\begin{align*}
	H_\infty^{\rm BF} &= \rme^{-\rmi t  \rmd\Gamma(\cF^{-1}\omega \cF)}   H^\infty  \rme^{\rmi t  \rmd\Gamma(\cF^{-1}\omega \cF)}  + \left(\partial_t \rme^{-\rmi t  \rmd\Gamma(\cF^{-1}\omega \cF)}   \right) \rme^{\rmi t  \rmd\Gamma(\cF^{-1}\omega \cF)} \,.
\end{align*}
Consequently, \cref{thm:Infty-Volume-Dynamics} follows directly from \cref{thm:Infinite-Volume-Dynamics}.

\paragraph{Acknowledgements.}
This work was funded by the Deutsche Forschungsgemeinschaft (DFG, German Research Foundation) – SFB/CRC TRR 352 – Project-ID 470903074; and Project-ID 258734477 – SFB/CRC 1173,  and by the Agence Nationale de la Recherche (ANR-23-CE40-0025) together with the Deutsche Forschungsgemeinschaft (Project-ID 529797785).
S. Spruck gratefully acknowledges financial support from the German Academic Scholarship Foundation, and the EIPHI Graduate School (contract ANR-17-EURE-0002). J.L. received financial support from the Bougogne-Franche Comté region
through the project SQC.


\appendix
\appendixpage

\section{The Localization Function} \label{sec:LocalizationFunction}
For $0<s<1/3$ and $n\in\BN_0$ we define the localization function
$$\Theta_\Lambda(x)=\frac{1}{1+|\Lambda^{-s} x|^{2n}}\in C^{\infty}(\mathbb{R}^3, \mathbb{R})\,,$$ which is used to localize the condensate around the origin. It satisfies the following characterizing properties: For all $ \beta\in\mathbb{N}_0^3$ there exists a constant $C>0$ such that for all volumes $\Lambda\geq1$
\begin{align}
	|\Theta_\Lambda|\leq C \, , \quad 
 	|\partial^\beta\Theta_\Lambda|\leq C \Lambda^{-s|\beta|}|\Theta_\Lambda| \,. \label{eq:Loc-Fct-Porp-2}
\end{align}

In order for $\tl$ to localize the condensate $\varphi_0$ we have to choose $s<1/3$ such that its scale is smaller than the $\cO(\Lambda^{1/3})$ scale of the condensate.

To facilitate our estimates, we introduce some notation from \cite{DFPP16}, which is particularly useful when applying Young's inequality.
\begin{notation} \label{def:Young-norms}
 For $1\leq
  p_1,\ldots, p_M\leq \infty$, $M\in\mathbb{N}_+$, we define the norms
  \begin{align*}
    \Vert f \Vert_{p_1\wedge \ldots\wedge p_M} & :=
    \inf_{f=f_{p_1}+\ldots+f_{p_M}} \Big(
    \Vert f_{p_1} \Vert_{p_1}+\ldots
    +\Vert f_{p_M} \Vert_{p_M}\Big) \, . 
  \end{align*}
and
\[
    \Vert f \Vert_{p_1, \ldots,
    p_M} := \Vert f \Vert_{p_1}+\ldots+\Vert f \Vert_{p_M} \,.
\]
\end{notation}

All relevant estimates for the localization function, obtained via Taylor expansion, are collected in the following lemma.

\begin{lemma}[Localized Estimates] \label{Lemma:CalculationTheta}
	Let $n\in\BN_0$ and $s>0$. Let $\tl(x)=\frac{1}{1+(\Lambda^{-s}x)^{2n}}$ be the localization function.
	\begin{itemize}
		 \item[a)](Functions flat around the origin) Let $k\in\BN_0$ with $k\leq 2n$. For all volumes $\Lambda\geq1$, let $f:=f_\Lambda\in C^{k}(\BR^3,\BC)$. If we have flatness around the origin of $f$, namely there exists a constant $C>0$ such that for all volumes $\Lambda\geq1$ and $ 0\leq |\beta|\leq k-1$ we have  $|\partial^{\beta}f(0)|\leq \Vert \partial^{\beta}f\Vert_\infty C \Lambda^{-(k-|\beta|)(1/3-s)} $ then 
		 \begin{itemize}
		 		\item[i)]  There exists a constant $C>0$ such that for all volumes $\Lambda\geq1$
		 \begin{align}
		 	\Vert \tl f\Vert_\infty \leq{}& C  \Lambda^{-k(1/3-s)} \Big( \sum_{|\beta|\leq k} \Vert \partial^{\beta}f \Vert_\infty\Lambda^{|\beta|/3}\Big) \,. \label{eq:Theta-flat-around-origin-estimate-infty}
 		 \end{align}
 		 	\item[ii)] If in addition $n\geq1$ and $k\leq 2(n-1)$ then there exists a constant $C>0$ such that for all volumes $\Lambda\geq1$ 
 		 \begin{align}
 		 	\Vert \tl f\Vert_2 \leq{}& C \Lambda^{-k(1/3-s) +3s/2}\Big( \sum_{|\beta|\leq k} \Vert \partial^{\beta}f \Vert_\infty\Lambda^{|\beta|/3}\Big) \,. \label{eq:Theta-flat-around-origin-estimate-L2}
 		 \end{align}
		\end{itemize}		 
 		 
		 \item[b)](Convolution)
		 Then for all orders $ m\in \BN_+$, there exists a constant $ C>0$ such that for all volumes $\Lambda\geq1$ and $ f,W\in L^p(\mathbb{R}^3)$, for $1\leq p\leq \infty$, we have
		 \begin{align}
	\Vert \Theta_\Lambda  W\ast f\Vert_2
	  \leq{}& 
	  C \sum\limits_{0\leq |\beta | \leq m-1} 
	  \Vert \,|y|^{|\beta|} W \Vert_{1,2}  \Vert \tl f \Vert_{1 \wedge 2} 
	  \nn \\
	  &+
	  C  \Lambda^{-sm} 
	  \Vert \, |y|^m W\Vert_{1,2} \Vert f \Vert_{1 \wedge 2} \,, \label{eq:MoveThetaInsideConvolutionL2}
	  \\
	  \Vert \Theta_\Lambda  W\ast f\Vert_\infty
	  \leq{}& 
	  C \sum\limits_{0\leq |\beta | \leq m-1} 
	  \Vert \,|y|^{|\beta|} W \Vert_{1,2,\infty}  \Vert \tl f \Vert_{1 \wedge 2\wedge \infty} 
	  \nn \\
	  &+
	  C  \Lambda^{-sm} 
	  \Vert \, |y|^m W\Vert_{1,2,\infty} \Vert f \Vert_{1 \wedge 2\wedge \infty} \,. \label{eq:MoveThetaInsideConvolutionLInfinity}
\end{align}
	\end{itemize}
\end{lemma}

\begin{rem} 
Part b) of the lemma is required to close Grönwall-type estimates of the form 
$\partial_t\Vert \Theta_\Lambda f_t\Vert \leq \alpha(t) +  \beta(t)\Vert \Theta_\Lambda f_t\Vert$, where the localization function $\tl$ has to be moved inside convolution terms of the form $W\ast f$ in order to complete the argument. In this context, the terms $C\Lambda^{-sm} 
	  \Vert \, |y|^m W\Vert_{1,2} \Vert f \Vert_{1 \wedge 2}$ and $ C\Lambda^{-sm} 
	  \Vert \, |y|^m W\Vert_{1,2,\infty} \Vert f \Vert_{1 \wedge 2\wedge \infty}$ appear as error terms.
	  
	 Note that the flatness around the origin of the initial condensate, \cref{con:Condensate-flat-around-origin}, is chosen in such a way that the condensate satisfies the requirements of \cref{Lemma:CalculationTheta}.
\end{rem}

\begin{proof}[Proof of \cref{Lemma:CalculationTheta}]
\textit{Part a):}
\\ The case $k=0$ is trivial. Now let $k\geq1$. By the Taylor expansion formula  of $f$ up to order $k-1$ around 0 we have $\forall x\in \BR^3$ $\exists \xi:=\xi_x\in [0,1]$ such that
	\begin{align}
		| (\tl f)(x)| ={}&  \tl (x) \bigg\vert   \sum_{0\leq |\beta|\leq k-1} \frac{\partial^{\beta}f(0)}{\beta!} x^{\beta} + \sum_{|\beta|=k} \frac{\partial^{\beta} f(\xi x)}{\beta!} x^{\beta}  \bigg\vert \nn \\
		\leq{}&  C  \sum_{|\beta|\leq k}   \Vert \partial^{\beta}f \Vert_\infty \Lambda^{-(k-|\beta|)(1/3-s)}  \frac{|x|^{|\beta|}}{1+ (\Lambda^{-s}|x|)^{2n}} \,, \label{eq:Theta-f-sup-estimate}
	\end{align} 
	where we used $|\partial^{\beta}f(0)|\leq \Vert \partial^{\beta}f\Vert_\infty \Lambda^{-(k-|\beta|)(1/3-s)} $, $\forall 0\leq |\beta|\leq k-1$, and the definition $\tl(x)= \frac{1}{1+ (\Lambda^{-s}|x|)^{2n}}\geq0$.
	\\Using $|\beta|\leq k\leq 2n$ and distinguishing between the cases $|x|<\Lambda^s$ and $|x|\geq\Lambda^s$ we can show for all $x\in\BR^3$
	\begin{align}
		\frac{|x|^{|\beta|}}{1+ (\Lambda^{-s}|x|)^{2n}} \leq  \Lambda^{s|\beta|} \,. \label{eq:Localized-postion-in-Taylor-1}
	\end{align}
	The estimate \eqref{eq:Theta-flat-around-origin-estimate-infty} follows from \eqref{eq:Theta-f-sup-estimate} and \eqref{eq:Localized-postion-in-Taylor-1}.
	\\Now we want to prove the $\Vert \tl f\Vert_2$ estimate. Let us assume $k\leq 2(n -1)$ and thus $n\geq1$. We use \eqref{eq:Theta-f-sup-estimate} and the substitution $y=\Lambda^{-s} x$ to obtain 
	\begin{align*}
		\Vert \tl f\Vert_2& = \left( \int_{\BR^3} |\tl f(x)|^2 dx \right)^{1/2} \\
		&\leq C\Lambda^{-k(1/3-s)}  \sum_{|\beta|\leq k} \Vert \partial^{\beta}f \Vert_\infty\Lambda^{|\beta|/3}       \left( \int_{\BR^3} \bigg( \frac{|y|^{|\beta|}}{1+ |y|^{2n}} \right)^2 \Lambda^{3s} dy \bigg)^{1/2} \\
		&\leq  C\Lambda^{-k(1/3-s) +3/2s} \sum_{|\beta|\leq k} \Vert \partial^{\beta}f \Vert_\infty\Lambda^{|\beta|/3}\,.
	\end{align*}
\textit{Part b):}
\\We want to change the argument of $\tl$ from $x$ to $x-y$ such that we can move it inside the convolution. To this end, for $x,y\in \BR^3$, we expand $\tl$ around the point $x-y$. By Taylor’s theorem, there exists a $ \theta\in[0,1]$ such that
\begin{align}
	 \vert \Theta_\Lambda(x)  \vert 
	 \leq{}& \sum\limits_{0\leq |\beta | \leq m-1} C_m | \Theta_1\left(\Lambda^{-s}(x-y)\right)| \cdot   \Lambda^{-s |\beta|}\frac{|y|^{|\beta|}}{\beta!} \nonumber\\
	  &+ \sum\limits_{|\beta|=m} C_{m} \cdot \Lambda^{-s|\beta|}\frac{|y|^{|\beta|}}{\beta!} \,,  \label{eq:DiffThetaXThetaX-Y}
\end{align}
where we have used \eqref{eq:Loc-Fct-Porp-2}.
Now with \eqref{eq:DiffThetaXThetaX-Y}
\begin{align*}
	\vert  \Theta_\Lambda (x) W\ast f(x) \vert
	\leq{}& C_m \sum\limits_{0\leq |\beta | \leq m-1} \int dy    |y|^{|\beta|}\cdot |W(y)  \Theta_\Lambda\left(x-y\right) f(x-y)| \\
	  &+ C_m \Lambda^{-sm}  \int dy |y|^{m}\cdot |W(y)  f(x-y)|
\end{align*}
with Young's inequality for $p\geq1$
\begin{align*}
	\Vert \Theta_\Lambda  W\ast f\Vert_p
	  \leq{}& 
	  C_m \sum\limits_{0\leq |\beta | \leq m-1} 
	  \Vert \,|y|^{|\beta|} W\Vert_{1,p}  \Vert \tl f \Vert_{1 \wedge p} 
	  \\
	  &+
	  C_m  \Lambda^{-sm} 
	  \Vert \, |y|^m W\Vert_{1,p} \Vert f \Vert_{1 \wedge p} \,.
\end{align*}
This proves \eqref{eq:MoveThetaInsideConvolutionL2}. The proof of \eqref{eq:MoveThetaInsideConvolutionLInfinity} is analogous.
\end{proof}


\section{Control of the Condensate Dynamics} \label{sec:Control-of-the-condensate}
To control and approximate the dynamics of the condensate $\varphi_{t}^\Lambda$, we define an auxiliary function $\tilde{\varphi}_t^\Lambda$ as in \cite{DFPP16}. 
\begin{definition}[Auxiliary Function] \label{def:Auxiliary-Function}
	For all volumes $\Lambda\geq1$, let $\varphi_0^\Lambda\in H^\infty(\BR^3)$. We call
	\begin{align}
		\tilde{\varphi}_t^\Lambda = \rme^{-\rmi \left( t V \ast |\varphi_0^\Lambda|^2 - \int_0^t \mu_s^\Lambda ds\right) } \varphi_0^\Lambda 
	\end{align}
	the auxiliary function, where $\mu_t^\Lambda$ is defined in \eqref{eq:Mu}.
\end{definition}
\begin{rem}\label{rem:Auxiliary-Function}

The kinetic term in the time evolution of $\varphi_t^\Lambda$ in \eqref{eq:Hartree} can be omitted, as it is subleading compared to the interaction term. Specifically,
\[ \Vert -\Delta\varphi_0^\Lambda\Vert_2\leq C \Lambda^{1/2-2/3} \,, \quad \text{while}\quad \Vert V\ast |\varphi_0^\Lambda|^2\varphi_0^\Lambda\Vert_2\leq C \Lambda^{1/2}\,. \]
Due to the appropriate scaling of $\varphi_0^\Lambda$, these properties also hold for the time evolved state (see \cref{cor:derivative-varphi-L2-norm-estimate} below). This justifies the approximation of $\varphi_t^\Lambda$ through $\tphi_t^\Lambda$. In particular, $\tphi_t^\Lambda$ and its derivatives satisfy analogous estimates as $\varphi_0^\Lambda$ (see e.g. \cref{lem:EstimatesTPhi} below).
\end{rem}

The following preliminaries \eqref{eq:moduls-phi-phi0-L2-bound} and \eqref{eq:phi-L2infty-bound} form the foundation of our control of the condensate. 
They are adapted from bounds in \cite{DFPP16}, and remain valid under slightly weaker conditions than considered there. Nevertheless, their proofs are analogous to those presented in \cite{DFPP16} (see \cite{Spr25} for an adaptation to the setting considered here).

\begin{prop}\label{prop:phi-propagation} 
 For all volumes $\Lambda\geq1$, let $\varphi_t^\Lambda$ be the solution of the Hartree equation \eqref{eq:Hartree}. Assume that its initial data $\varphi_0^\Lambda$ varies on the scale $\Lambda^{1/3}$, that is, \cref{con:Initial condition}. Then for all times $T\geq0$, there exists a constant $C>0$ such that for all volumes $ \Lambda\geq1$ and $-T\leq t \leq T$
  \begin{align}
   \label{eq:moduls-phi-phi0-L2-bound}
    \left\Vert \varphi_{t}^\Lambda -
       \tilde{\varphi}_{t}^\Lambda
    \right\Vert_{2} & \leq C \Lambda^{-\frac{1}{6}}\,,
    \\
    \Vert\varphi_t^\Lambda\Vert_\infty &\leq C\, . \label{eq:phi-L2infty-bound}
  \end{align}
\end{prop}

\begin{cor}\label{cor:phi-propagation} 
Assume the conditions of \cref{prop:phi-propagation}. Then for all times $T\geq0$, there exists a constant $C>0$ such that for all volumes $ \Lambda\geq1$ and $-T\leq t \leq T$
\begin{align}
	\Vert |\varphi_t^\Lambda|^2-|\varphi_0^\Lambda|^2\Vert_{1\wedge2}\leq C \Lambda^{-\frac{1}{6}}\,.
\end{align}
\end{cor}
\begin{proof}[Proof of \cref{cor:phi-propagation}]
The claim follows directly from \cref{prop:phi-propagation}, the definition of the norm $\Vert\,.\,\Vert_{1\wedge2}$ (see \cref{def:Young-norms}) and the identity $|\varphi_t^\Lambda|^2-|\varphi_0^\Lambda|^2 = |\varphi_t^\Lambda -\tphi_t^\Lambda|^2 + 2\mathrm{Re}(\tphi_t^\Lambda)^* (\varphi_t^\Lambda - \tphi_t^\Lambda)$, where we use $|\varphi_0^\Lambda|=|\tphi_t^\Lambda|$. 
\end{proof}

\subsection{Propagation Estimates} 

In this section, we use the auxiliary function $\tphi_t^\Lambda$ to approximate the condensate $\varphi_t^\Lambda$. We begin by estimating $\tphi_t^\Lambda$.
\begin{lemma}[Estimates of $\tphi_t^\Lambda$] \label{lem:EstimatesTPhi}
Let $\beta\in \mathbb{N}_0^3$. For all $\Lambda\geq1$ let $\varphi_0^\Lambda$ be the condensate, which varies on the scale $\Lambda^{1/3}$, meaning it satisfies \cref{con:Initial condition derivative condensate}$_{|\beta|}$. Then there exists a constant $ C>0$ such that for all volumes $ \Lambda\geq1$ and times $t\in\BR$
\begin{align}
		\Vert  \partial^\beta\tphi_t^\Lambda  \Vert_{\infty} \leq{}& C \Lambda^{- \frac{|\beta|}{3}}  \,, \label{eq:DTPhiInfty} \\
		\Vert  \partial^\beta\tphi_t^\Lambda  \Vert_{2}\leq{}& C \Lambda^{- \frac{|\beta|}{3} +\frac{1}{2}} \,. \label{eq:DTPhiL2}
	\end{align}
\end{lemma}
\begin{proof}[Proof of \cref{lem:EstimatesTPhi}]
The claim follows directly from the definition $\tilde{\varphi}_t^\Lambda = \rme^{-\rmi \left( t V \ast |\varphi_0^\Lambda|^2 - \int_0^t \mu_s^\Lambda ds\right) } \varphi_0^\Lambda$, our assumptions, and the Leibniz rule.
\end{proof}

The following \cref{lem:FirstApproxThetaDPhi-TPhi} extends \eqref{eq:moduls-phi-phi0-L2-bound} to derivatives of $\varphi_t^\Lambda -\tphi_t^\Lambda$, enabling us to transfer estimates from $\tphi_t^\Lambda$ to the condensate $\varphi_t^\Lambda$ and serves as an essential ingredient for the proof of \cref{prop:ThetaDPhi-TPhi} in the following section.

\begin{prop}[Approximation of $\varphi_t^\Lambda$ by $\tphi_t^\Lambda$]  \label{lem:FirstApproxThetaDPhi-TPhi}
 Let $\beta\in \mathbb{N}_0^3$.  For all volumes $\Lambda\geq1$, let $\varphi_t^\Lambda$ be the solution of the Hartree equation \eqref{eq:Hartree}. Assume that its initial data $\varphi_0^\Lambda$ varies on the scale $\Lambda^{1/3}$, that is, \cref{con:Initial condition} and \cref{con:Initial condition derivative condensate}$_{|\beta|+2}$.
 \\ Then for all times $T\geq0$, there exists a constant $C>0$ such that for all volumes $ \Lambda\geq1$ and $-T\leq t \leq T$
	\begin{align}
		\Vert \partial^\beta( \varphi_t^\Lambda -\tphi_t^\Lambda)  \Vert_{2} \leq C \Lambda^{-1/6 -|\beta|/3} \,.
	\end{align}
\end{prop}

\begin{proof}[Proof of \cref{lem:FirstApproxThetaDPhi-TPhi}]
Let $T\geq0$ and $-T\leq t\leq T$.
We prove \cref{lem:FirstApproxThetaDPhi-TPhi} by induction on $|\beta|$, denoted as $n$.
\paragraph{Base Case: $|\beta| = 0$.}
For $|\beta| = 0$, the claim is the statement of \cref{prop:phi-propagation}, noting that we assume \cref{con:Initial condition}.
\paragraph{Induction Step: $|\beta| = n+1$.} 
 Now let $n\in\mathbb{N}_0$ be fixed. Assume that the assumptions of \cref{lem:FirstApproxThetaDPhi-TPhi} are fulfilled for n+1, namely \cref{con:Initial condition} and \cref{con:Initial condition derivative condensate}$_{(n+1)+2}$. Now, consider a multi-index  $\beta\in \mathbb{N}_0^3$ with $|\beta|= n+1$. 
\\As the induction hypothesis, suppose that \cref{lem:FirstApproxThetaDPhi-TPhi} holds for all multi-indices $|\tilde{\beta}|\leq n$. Thus
		$$\Vert \partial^{\tilde{\beta}}( \varphi_t^\Lambda-\tphi_t^\Lambda)  \Vert_{2} \leq C \Lambda^{-1/6 -|\tilde{\beta}|/3}\,.$$
We use a Grönwall estimate to get the desired bound:
 \begin{align}
 	&\partial_t \Vert \partial^{\beta} (\varphi_t^\Lambda -\tphi_t^\Lambda) \Vert_2^2 \nonumber \\
 	={}& 2\mathrm{Re} \Big\langle	\partial^{\beta} (\varphi_t^\Lambda -\tphi_t^\Lambda) \, , \,	 
 	\partial^{\beta} \Big( (-\rmi ) \Big[ -\frac{1}{2}\Delta + V\ast |\varphi_t^\Lambda|^2-\mu_t^\Lambda\Big] (\varphi_t^\Lambda\pm\tphi_t^\Lambda)
 	 \nn \\
 	& \qquad\qquad\qquad\qquad \ \ + \rmi \left(V\ast |\varphi_0^\Lambda|^2-\mu_t^\Lambda\right) \tphi_t^\Lambda \Big)  \Big\rangle
 	 \nonumber\\
 	={}& 2\mathrm{Re} \left\langle	\partial^{\beta} (\varphi_t^\Lambda -\tphi_t^\Lambda) \, , \,	 (-\rmi )\left[  \partial^{\beta} \,, \, V\ast |\varphi_t^\Lambda|^2-\mu_t^\Lambda \right] (\varphi_t^\Lambda -\tphi_t^\Lambda) \right\rangle \label{eq:ScalarProdKommV,Phi} \\
 	&+ 2\mathrm{Re} \Big\langle	\partial^{\beta} (\varphi_t^\Lambda -\tphi_t^\Lambda) \, , \,	 (-\rmi )\Big[ -\frac{1}{2}\Delta +  V\ast |\varphi_t^\Lambda|^2-\mu_t^\Lambda \Big] \partial^{\beta}(\varphi_t^\Lambda-\tphi_t^\Lambda) \Big\rangle \label{eq:ScalarProdHartreeDPhi-TPhi}\\
 	&+ 2\mathrm{Re} \left\langle	\partial^{\beta} (\varphi_t^\Lambda -\tphi_t^\Lambda) \, , \,	(-\rmi ) \partial^{\beta} V\ast \left(|\varphi_t^\Lambda|^2-|\varphi_0^\Lambda|^2\right) \tphi_t^\Lambda \right\rangle \label{eq:ScalarProdDVConvPhi-TPhi}\\
 	&+ 2\mathrm{Re} \Big\langle	\partial^{\beta} (\varphi_t^\Lambda-\tphi_t^\Lambda) \, , \,	(-\rmi ) \partial^{\beta} \Big( -\frac{1}{2}\Delta \Big) \tphi_t^\Lambda \Big\rangle \,,		\label{eq:ScalarProdDLaplace}
 \end{align}
 where we used $\mu_t^\Lambda\in \mathbb{R}$.
 The term \eqref{eq:ScalarProdHartreeDPhi-TPhi} vanishes, because we only consider its real part and $\mu_t^\Lambda\in \mathbb{R}$. In the following we estimate the remaining terms \eqref{eq:ScalarProdKommV,Phi}, \eqref{eq:ScalarProdDVConvPhi-TPhi} and \eqref{eq:ScalarProdDLaplace}.
 \\Given \cref{con:Initial condition derivative condensate}$_{n+3=|\beta|+2}$, the estimate for \eqref{eq:ScalarProdDLaplace} follows directly from \cref{lem:EstimatesTPhi}. 
 For \eqref{eq:ScalarProdDVConvPhi-TPhi} and \eqref{eq:ScalarProdKommV,Phi}, we again use \cref{lem:EstimatesTPhi} together with the induction hypothesis, which ensures that
 \begin{align}
		\Vert \,|V|\ast  \partial^{\beta} |\varphi_t^\Lambda|^2 \, \Vert_\infty \leq{}& C \Lambda^{-|\beta|/3} + C \Vert \partial^{\beta}  (\varphi_t^\Lambda -\tphi_t^\Lambda ) \Vert _2  \,, \label{eq:VConvDPhi-1ForDPhi-TPhi}\\
		\Vert \,|V|\ast \left\vert \partial^{\beta} (|\varphi_t^\Lambda|^2-|\tphi_t^\Lambda|^2) \right\vert\, \Vert_2 \leq{}& C \Lambda^{-1/6-|\beta|/3} + C \Vert \partial^{\beta}  (\varphi_t^\Lambda -\tphi_t^\Lambda ) \Vert _2\,, \label{eq:VConvDPhi-TPhiForDPhi-TPhi} \\
		 \big\Vert   -\frac{1}{2}\Delta \partial^{\beta} \tphi_t^\Lambda \big\Vert_2 \leq{}&  C\Lambda^{\frac{1}{2} -\frac{|\beta|+2}{3}}	\leq  C\Lambda^{-\frac{1}{6} -\frac{|\beta|}{3}} \,. \label{eq:ScalarProdDLaplaceEstimate} 
	\end{align}
To summarize
\begin{align*}
	\partial_t \Vert \partial^{\beta} (\varphi_t^\Lambda-\tphi_t^\Lambda) \Vert_2^2 \leq{}& \eqref{eq:ScalarProdKommV,Phi} + \eqref{eq:ScalarProdDVConvPhi-TPhi} +\eqref{eq:ScalarProdDLaplace} \\
	\leq{}& \Vert\partial^{\beta} (\varphi_t^\Lambda-\tphi_t^\Lambda) \Vert_2 C \left( \Lambda^{-\frac{1}{6} -\frac{|\beta|}{3}} + \Vert\partial^{\beta} (\varphi_t^\Lambda-\tphi_t^\Lambda) \Vert_2 \right) \,,
\end{align*}
with Grönwall and $\tphi_0=\varphi_0$ we get the claim $\Vert \partial^{\beta} (\varphi_t^\Lambda-\tphi_t^\Lambda) \Vert_2\leq C \Lambda^{-\frac{1}{6} -\frac{\beta}{3}}$ for $|\beta|=n+1$, thereby completing the induction step.
\end{proof}

Combining the estimates of \cref{lem:EstimatesTPhi} and \cref{lem:FirstApproxThetaDPhi-TPhi} we gain control of the condensate $\varphi_t^\Lambda$ as stated in the following corollary.

\begin{cor}[Estimates of the Condensate $\varphi_t^\Lambda$] \label{cor:derivative-varphi-L2-norm-estimate}
	 Let $\beta\in \mathbb{N}_0^3$.  For all volumes $\Lambda\geq1$, let $\varphi_t^\Lambda$ be the solution of the Hartree equation \eqref{eq:Hartree}. Assume that its initial data $\varphi_0^\Lambda$ varies on the scale $\Lambda^{1/3}$, that is, \cref{con:Initial condition} and \cref{con:Initial condition derivative condensate}$_{|\beta|+2}$. \\ Then for all times $T\geq0$, there exists a constant $C>0$ such that for all volumes $ \Lambda\geq1$ and $-T\leq t \leq T$
	 \begin{align}
	 	\Vert \partial^\beta \varphi_t^\Lambda \Vert_2 \leq{}& C \Lambda^{1/2-|\beta|/3}\,, \\
	 	\Vert \partial^\beta \varphi_t^\Lambda \Vert_{2\wedge\infty} \leq{}& C \Lambda^{-|\beta|/3} \,,
	 \end{align}
	 where the norm $\Vert\,.\,\Vert_{2\wedge\infty}$ is defined in \cref{def:Young-norms}.
\end{cor}

\begin{proof}[Proof of \cref{cor:derivative-varphi-L2-norm-estimate}]
	\cref{cor:derivative-varphi-L2-norm-estimate} follows directly from \cref{lem:EstimatesTPhi} and \cref{lem:FirstApproxThetaDPhi-TPhi}, using $\varphi_t^\Lambda= \tphi_t^\Lambda + (\varphi_t^\Lambda -\tphi_t^\Lambda)$.
\end{proof}

\subsection{Local Stability} \label{sec:Localized-condensate}

The goal of this section is to show that the localized condensate $\tl|\varphi_t^\Lambda|^2$ can be approximated arbitrarily well by the initial condensate $\tl |\varphi_0^\Lambda|^2$, thereby improving \cref{cor:phi-propagation}. Note that the localization function $\tl$ varies on the smaller scale $\cO(\Lambda^{s})$, with $0<s<1/3$, localizing the condensate varying on the scale $\cO(\Lambda^{1/3})$.

\begin{prop}[Local Stability of the Condensate]\label{cor:ThetaPhi-Phi0}
Let $n\in \mathbb{N}_+$, $k\in \mathbb{N}_0$ and $0<s<1/3$. Set the localization function to be $\tl(x)=\frac{1}{1+(\Lambda^{-s}x)^{2n}}$. For all volumes $\Lambda\geq1$, let $\varphi_t$ be the solution of the Hartree equation \eqref{eq:Hartree}. Assume that its initial data $\varphi_0$  varies on the scale $\Lambda^{1/3}$, that is, 
\cref{con:Initial condition}, \cref{con:Initial condition derivative condensate}$_{k+3}$ and \cref{con:Initial condition derivative condensate}i)$_{(k+2)+2(n-1)}$. Furthermore, we require that $\varphi_0$ is flat around the origin, namely \cref{con:Condensate-flat-around-origin}$_{2(n-1),s}$.
\\Then for all $T\geq0$, there exists a constant $C>0$ such that for all volumes $ \Lambda\geq1$ and $-T\leq t \leq T$
	\begin{align}
		\Vert \Theta_\Lambda (|\varphi_t^\Lambda|^2-|\varphi_0^\Lambda|^2)  \Vert_{1\wedge 2} \leq C \Big\{ \Lambda^{-\frac{1}{6} -\frac{k+1}{3}- (k+1)s} + \Lambda^{-\frac{1}{2}s -2n(1/3-s)} \Big\} \,. 
	\end{align}
\end{prop}

\begin{proof}[Proof of \cref{cor:ThetaPhi-Phi0}]
The claim follows directly from \cref{prop:ThetaDPhi-TPhi} below, together with the definition of the norm $\Vert\,.\,\Vert_{1\wedge2}$ (see \cref{def:Young-norms}) and the identity $|\varphi_t^\Lambda|^2-|\varphi_0^\Lambda|^2 = |\varphi_t^\Lambda -\tphi_t^\Lambda|^2 + 2\mathrm{Re}(\tphi_t^\Lambda)^* (\varphi_t^\Lambda - \tphi_t^\Lambda)$ with $|\tphi_t^\Lambda|=|\varphi_0^\Lambda|$.  
\end{proof}
Although flatness of the initial condensate is not assumed in this section, the above approximation provides a mechanism to propagate flatness, if present in the initial data, to the evolved $\varphi_t^\Lambda$, as for instance in \cref{cor:ThetaDPhi-TPhi}.

The estimate in \cref{cor:ThetaPhi-Phi0} is based on the following approximation of $\tl\varphi_t^\Lambda$ by $\tl\tphi_t^\Lambda$, which also needs control over their derivatives.
\begin{lemma} \label{prop:ThetaDPhi-TPhi}
Assume the conditions of \cref{cor:ThetaPhi-Phi0}. Then for all $\beta\in\mathbb{N}_0^3$ with $ |\beta|\leq k$, and $T\geq0$, there exists a constant $C>0$ such that for all volumes $ \Lambda\geq1$ and $-T\leq t \leq T$
	\begin{align}
		\Vert \Theta_\Lambda  \partial^\beta( \varphi_t^\Lambda-\tphi_t^\Lambda)  \Vert_{2} \leq C \Big\{ \Lambda^{-\frac{1}{6} -\frac{k+1}{3}- (k+1-|\beta|)s} + \Lambda^{-\frac{1}{2}s -\frac{|\beta|}{3}-2n(1/3-s)} \Big\} \,. \label{eq:EstimateThetaDPhi-TPhi}
	\end{align}
\end{lemma}

\begin{proof}[Proof of \cref{prop:ThetaDPhi-TPhi}]
Let $T \geq 0$ and $-T \leq t \leq T$. 
We prove \cref{prop:ThetaDPhi-TPhi} by induction on $k$. 
\paragraph{Base Case: $k = 0$.}
For $k = 0$ also $\beta=0$ and the proof is analogous to the argument in the induction step below. Note that the use of the induction hypothesis is not necessary as no terms $\Vert \tl\partial^{\gamma}( \varphi_t^\Lambda -\tphi_t^\Lambda)\Vert_2$ with $|\gamma|\leq |\beta|-1$ appear. 
\paragraph{Induction Step: $k-1$ to $k$.} 

Let $k \in \mathbb{N}_+$ be fixed. 
In order to estimate derivatives with greater order than $k-1$, for which we cannot use the induction hypothesis, we need an external estimate provided in \cref{lem:FirstApproxThetaDPhi-TPhi}. 
We use a Grönwall argument to prove the induction step. We have
	\begin{align}
		&\partial_t \Vert \tl \partial^\beta (\varphi_t^\Lambda -\tphi_t^\Lambda)\Vert_2^2 \nonumber \\
		={}& 2\mathrm{Re}\Big\langle \tl\partial^{\beta}( \varphi_t^\Lambda -\tphi_t^\Lambda)\, ,\,  \tl  \partial^{\beta}\Big\{ (-\rmi ) \Big( -\frac{1}{2}\Delta + V\ast |\varphi_t^\Lambda|^2-\mu_t^\Lambda \Big)( \varphi_t^\Lambda\pm\tphi_t^\Lambda) \nn \\
		& \qquad\qquad\qquad\qquad\qquad + \rmi \left(V\ast |\varphi_0^\Lambda|^2- \mu_t^\Lambda\right)\tphi_t^\Lambda	\Big\}    \Big\rangle 	\nn \\
		 ={}& +2\mathrm{Re}\left\langle \tl\partial^{\beta}( \varphi_t^\Lambda -\tphi_t^\Lambda)\, ,\, (-\rmi )\tl\left[   \partial^{\beta} \, ,\,  V\ast |\varphi_t^\Lambda|^2- \mu_t^\Lambda \right]( \varphi_t^\Lambda -\tphi_t^\Lambda)   \right\rangle  \label{eq:ScalarProdThetaKommuDVConvPhi} \\
		 &+ 2\mathrm{Re}\left\langle \tl\partial^{\beta}( \varphi_t^\Lambda -\tphi_t^\Lambda)\, ,\, (-\rmi )\tl \partial^{\beta} V\ast\left( |\varphi_t^\Lambda|^2-|\varphi_0^\Lambda|^2 \right)\tphi_t^\Lambda  \right\rangle  \label{eq:ScalarProdThetaDVConvPhi-Phi0}	\\
		 &+ 2\mathrm{Re}\Big\langle \tl\partial^{\beta}( \varphi_t^\Lambda-\tphi_t^\Lambda)\, ,\, (-\rmi )\tl \partial^{\beta}\Big( -\frac{1}{2}\Delta  \Big)\tphi_t^\Lambda  \Big\rangle  \label{eq:ScalarProdThetaDLaplace}	\\
		& -2 \mathrm{Re}\left\langle \tl \partial^{\beta}( \varphi_t^\Lambda -\tphi_t^\Lambda)\, ,\, (-\rmi )(\nabla\tl) \nabla\partial^{\beta}( \varphi_t^\Lambda -\tphi_t^\Lambda)	 \right\rangle  \label{eq:ScalarProdNablaThetaNablaD}	\,.
	\end{align} 
We will now estimate $\partial_t \Vert \tl \partial^\beta (\varphi_t^\Lambda -\tphi_t^\Lambda)\Vert_2^2$ by estimating the individual terms \eqref{eq:ScalarProdThetaKommuDVConvPhi}, \eqref{eq:ScalarProdThetaDVConvPhi-Phi0}, \eqref{eq:ScalarProdThetaDLaplace} and \eqref{eq:ScalarProdNablaThetaNablaD}. To shorten the notation we set $\kappa_{|\beta|,k}(\Lambda) :=\Lambda^{-\frac{1}{6} -\frac{k+1}{3}- (k+1-|\beta|)s}  
		+ \Lambda^{-\frac{1}{2}s -\frac{|\beta|}{3}-2n(1/3-s)}$.
	\\\textit{To \eqref{eq:ScalarProdThetaKommuDVConvPhi}:}
	\\Consider the case $|\beta|\geq 1$. Note that \eqref{eq:ScalarProdThetaKommuDVConvPhi}$_{\beta=0}$=0.
	 We use \cref{cor:derivative-varphi-L2-norm-estimate} and the induction hypothesis \cref{prop:ThetaDPhi-TPhi}$_{k-1}$, $|\gamma|\leq |\beta|-1\leq k-1$ to conclude for all $0\leq |\beta|\leq k$:
	\begin{align}
		\eqref{eq:ScalarProdThetaKommuDVConvPhi}_\beta &\leq \Vert \tl\partial^{\beta}( \varphi_t^\Lambda-\tphi_t^\Lambda)\Vert_2  \nn \\
		&\quad\times\sum_{\gamma=0,\, |\gamma|\leq |\beta|-1}^{\beta} C \Vert V\ast \partial^{\beta-\gamma} |\varphi_t^\Lambda|^2 \Vert_\infty \Vert \tl \partial^{\gamma} (\varphi_t^\Lambda -\tphi_t^\Lambda) \Vert_2 \nn \\
		 &\leq \Vert \tl\partial^{\beta}( \varphi_t^\Lambda -\tphi_t^\Lambda)\Vert_2 C \kappa_{|\beta|,k}(\Lambda) \,, \label{eq:ScalarProdThetaKommuDVConvPhiEstimate}
	\end{align}
		where we used $|\gamma|s\leq (|\beta|-1)s$ and $\Lambda\geq 1$.
	\\\textit{To \eqref{eq:ScalarProdThetaDVConvPhi-Phi0}:}
	\\First note that for $ 0\leq |\beta| \leq k$  we have that
	 \begin{align}
	 	&\Vert  \tl V \ast   \partial^{\beta} \left( |\varphi_t^\Lambda|^2 - |\varphi_0^\Lambda|^2 \right)  \Vert_2
	 	\leq{}  C \kappa_{|\beta|,k}(\Lambda)
		+ C \Vert \tl\partial^{\beta} (\varphi_t^\Lambda - \tphi_t^\Lambda) \Vert_2 \,, \label{eq:auxiliary-estimate-1}
	\end{align}
	which follows from the induction hypothesis \cref{prop:ThetaDPhi-TPhi}$_{k-1}$, the external estimate \cref{lem:FirstApproxThetaDPhi-TPhi}, $\tl\leq C$ and \cref{Lemma:CalculationTheta}b) with the use of the identity $|\varphi_t^\Lambda|^2-|\varphi_0^\Lambda|^2 = |\varphi_t^\Lambda -\tphi_t^\Lambda|^2 + 2\mathrm{Re}(\tphi_t^\Lambda)^* (\varphi_t^\Lambda - \tphi_t^\Lambda)$ and \cref{lem:EstimatesTPhi}.
	\\We use \cref{lem:EstimatesTPhi}, $|\beta|\leq k$, and \eqref{eq:auxiliary-estimate-1}, $|\gamma|\leq k\leq k+2$, to conclude
	\begin{align}
		&\Vert \tl \partial^{\beta} V\ast\left( |\varphi_t^\Lambda|^2-|\varphi_0^\Lambda|^2 \right)\tphi_t^\Lambda \Vert_2 
		\nn \\
		\leq{}&
		 \sum_{\gamma=0}^\beta C \Lambda^{-\frac{|\beta-\gamma|}{3}} \Vert \tl V\ast \partial^{\gamma} \left( |\varphi_t^\Lambda|^2-|\varphi_0^\Lambda|^2 \right)  \Vert_2 
		 \nn \\
		\leq{}&
		\sum_{\gamma=0}^\beta C \Lambda^{-\frac{|\beta-\gamma|}{3}} \Big\{ \kappa_{|\gamma|,k}(\Lambda)  +  \Vert \tl\partial^{\gamma} (\varphi_t^\Lambda - \tphi_t^\Lambda) \Vert_2 \Big\} \label{eq:aux-eq-7}
	\end{align}
	Then we use the induction hypothesis \cref{prop:ThetaDPhi-TPhi}$_{k-1}$ for $|\gamma|\leq |\beta|-1\leq k-1$ to conclude
		\begin{align*}
		\eqref{eq:aux-eq-7}\leq{}&
		 C \kappa_{|\beta|,k}(\Lambda)  + C \Vert \tl\partial^{\beta} (\varphi_t^\Lambda - \tphi_t^\Lambda) \Vert_2   \,.
	\end{align*}
	Therefore for all $0\leq |\beta|\leq k$
	\begin{align}
		\eqref{eq:ScalarProdThetaDVConvPhi-Phi0}_\beta \leq{}& \Vert \tl\partial^{\beta}( \varphi_t^\Lambda-\tphi_t^\Lambda)\Vert_2 C  \kappa_{|\beta|,k}(\Lambda) 
		 + C \Vert \tl\partial^{\beta} (\varphi_t^\Lambda - \tphi_t^\Lambda) \Vert_2^2   \,. \label{eq:ScalarProdThetaDVConvPhi-Phi0Estimate}
	\end{align}
	\textit{To \eqref{eq:ScalarProdThetaDLaplace}:}
	\\We have to estimate $\Vert \tl \partial^\beta \Delta \tphi_t^\Lambda\Vert_2$. In the first step we use  \cref{Lemma:CalculationTheta}a) on $\varphi_0^\Lambda-1$ which is flat around the origin to get that
\begin{align}
	\Vert \tl \partial^{\beta} \Delta(\varphi_0^\Lambda-1)\Vert_2 \leq C \Lambda^{-2n \left(\frac{1}{3} -s \right) -\frac{|\beta|+2}{3} +\frac{3s}{2}}\,. 
\end{align}
Then it follows from $\tilde{\varphi}_t^\Lambda = \rme^{-\rmi \left( t V \ast |\varphi_0^\Lambda|^2 - \int_0^t \mu_s^\Lambda ds\right) } \varphi_0^\Lambda $ and the Leibniz rule that
	\begin{align}
		\Vert \tl \partial^\beta  \Delta\tphi_t^\Lambda  \Vert_{2}\leq{}& C \Lambda^{-2(n-1)\left(\frac{1}{3}-s\right) - \frac{|\beta|+2}{3} +\frac{3}{2}s} \,.\label{eq:ThetaDTPhiL2}
	\end{align}
	Thus	 $\forall 0\leq |\beta|\leq k$ we have the estimate\footnote{This estimate is precisely the reason why we required \cref{con:Condensate-flat-around-origin}$_{2(n-1),s}$, \cref{con:Initial condition derivative condensate}i)$_{(k+2)+2(n-1)}$, and \cref{con:Initial condition derivative condensate}ii)$_{k+2}$.}
	\begin{align}
		\eqref{eq:ScalarProdThetaDLaplace}\leq{}& \Vert \tl\partial^{\beta}( \varphi_t^\Lambda -\tphi_t^\Lambda)\Vert_2 C \Lambda^{-\frac{|\beta|+2}{3} + \frac{3}{2}s} \Lambda^{-2(n-1)\left(\frac{1}{3}-s\right)} \nonumber \\
		 ={}& \Vert \tl\partial^{\beta}( \varphi_t^\Lambda -\tphi_t^\Lambda)\Vert_2 C \Lambda^{- \frac{1}{2}s-\frac{|\beta|}{3} } \Lambda^{-2n\left(\frac{1}{3}-s\right)} \,. \label{eq:ScalarProdThetaDLaplaceEstimate}
	\end{align}
	\textit{To \eqref{eq:ScalarProdNablaThetaNablaD}:}
	\\The desired estimates for \eqref{eq:ScalarProdNablaThetaNablaD}$_\beta$ build on one another: the bounds for higher derivatives propagate to the lower ones. This also explains why it is necessary to control higher derivatives in order to improve the estimate in the derivative-free case, $\Vert \tl( \varphi_t^\Lambda-\tphi_t^\Lambda)\Vert_2$, which is the main objective of this lemma and used to prove \cref{cor:ThetaPhi-Phi0}. Therefore we start with the estimation of \eqref{eq:ScalarProdNablaThetaNablaD}$_{|\beta|=k}$. Here, we use \cref{lem:FirstApproxThetaDPhi-TPhi} and $\tl\leq C$ to get
	\begin{align}
		\eqref{eq:ScalarProdNablaThetaNablaD}_{|\beta|=k} \leq{}& C \Vert \tl\partial^{\beta}( \varphi_t^\Lambda-\tphi_t^\Lambda)\Vert_2 \Vert (\nabla\tl) \nabla\partial^{\beta}( \varphi_t^\Lambda-\tphi_t^\Lambda)\Vert_2 \nonumber \\
		\leq{}& C \Vert \tl\partial^{\beta}( \varphi_t^\Lambda -\tphi_t^\Lambda) \Vert_2 \Lambda^{-s} \Lambda^{-\frac{1}{6} - \frac{k+1}{3}}\,, \label{eq:ScalarProdNablaThetaNablaDEstimateBetaEqK}
	\end{align}
	where in the first step we used \eqref{eq:Loc-Fct-Porp-2} for $\tl$ and in the second \cref{lem:FirstApproxThetaDPhi-TPhi} for $|\beta|+1=k+1$.
	\\For $|\beta|=k$ the estimates obtained above imply
	\begin{align*}
		\partial_t \Vert \tl \partial^\beta (\varphi_t^\Lambda -\tphi_t^\Lambda)\Vert_2^2 \leq{} \left( \eqref{eq:ScalarProdNablaThetaNablaD} + \eqref{eq:ScalarProdThetaDLaplace} + \eqref{eq:ScalarProdThetaDVConvPhi-Phi0} + \eqref{eq:ScalarProdThetaKommuDVConvPhi} \right)_{|\beta|=k}& \\
		\leq{} \left(\eqref{eq:ScalarProdNablaThetaNablaDEstimateBetaEqK} + \eqref{eq:ScalarProdThetaDLaplaceEstimate} + \eqref{eq:ScalarProdThetaDVConvPhi-Phi0Estimate} + \eqref{eq:ScalarProdThetaKommuDVConvPhiEstimate} \right)_{|\beta|=k}& \\
		\leq{} \Vert \tl\partial^{\beta}( \varphi_t^\Lambda-\tphi_t^\Lambda)\Vert_2 C  \kappa_{|\beta|,k}(\Lambda) + C \Vert \tl\partial^{\beta} (\varphi_t^\Lambda - \tphi_t^\Lambda) \Vert_2^2& \,,
	\end{align*}
	with Grönwall and $\tphi_0^\Lambda=\varphi_0^\Lambda$ we conclude for $|\beta|=k$
	\begin{align}
		\Vert \tl \partial^\beta (\varphi_t^\Lambda -\tphi_t^\Lambda)\Vert_2 \leq C  \kappa_{k,k}(\Lambda) \,.
	\end{align}
	We have proven \eqref{eq:EstimateThetaDPhi-TPhi} for $|\beta|=k$. With this we have improved our estimate for $\Vert \tl \partial^\beta (\varphi_t^\Lambda -\tphi_t^\Lambda)\Vert_2$ from \eqref{eq:EstimateThetaDPhi-TPhi}$_{k-1}$, i.e. \cref{prop:ThetaDPhi-TPhi}$_{k-1}$, to \eqref{eq:EstimateThetaDPhi-TPhi}$_{k}$. This result will now be used to improve the estimates for $0\leq \beta\leq k-1$. As we will show below, each estimate for $0\leq |\beta|\leq k-1$ from our induction hypothesis \cref{prop:ThetaDPhi-TPhi}$_{k-1}$ is improved by the one corresponding to the next higher order in $|\beta|$.
	\\\\Now let $1\leq m\leq k$ and we assume that \eqref{eq:EstimateThetaDPhi-TPhi} holds for all $m\leq |\beta|\leq k$. In this case, we will show that \eqref{eq:EstimateThetaDPhi-TPhi} also holds for $|\beta| = m-1$, which proves the lemma.
	\\Due to the estimates \eqref{eq:ScalarProdThetaKommuDVConvPhiEstimate}, \eqref{eq:ScalarProdThetaDVConvPhi-Phi0Estimate} and \eqref{eq:ScalarProdThetaDLaplaceEstimate}, it only remains to bound \eqref{eq:ScalarProdNablaThetaNablaD}$_{|\beta|=m-1}$. 
	With the help of \eqref{eq:EstimateThetaDPhi-TPhi} which holds for $|\beta|+1=m$
	\begin{align}
		\eqref{eq:ScalarProdNablaThetaNablaD}_{|\beta|=m-1} &\leq{} C \Vert \tl\partial^{\beta}( \varphi_t^\Lambda -\tphi_t^\Lambda) \Vert_2 \Lambda^{-s}\Vert \tl \nabla\partial^{\beta}( \varphi_t^\Lambda -\tphi_t^\Lambda)\Vert_2 \nonumber \\
		&\leq{}  \Vert \tl\partial^{\beta}( \varphi_t^\Lambda -\tphi_t^\Lambda)\Vert_2  C \Lambda^{-s} \kappa_{|\beta|+1,k}(\Lambda) \nonumber \\
	 	&\leq{}  \Vert \tl\partial^{\beta}( \varphi_t^\Lambda -\tphi_t^\Lambda)\Vert_2  C  \kappa_{|\beta|,k}(\Lambda) \,. \label{eq:ScalarProdNablaThetaNablaDEstimateBetaEqM-1}
	\end{align}
	With \eqref{eq:ScalarProdNablaThetaNablaDEstimateBetaEqM-1}, \eqref{eq:ScalarProdThetaDLaplaceEstimate},  \eqref{eq:ScalarProdThetaDVConvPhi-Phi0Estimate} and \eqref{eq:ScalarProdThetaKommuDVConvPhiEstimate} all for $|\beta|=m-1$ we conclude
	\begin{align*}
		\partial_t \Vert \tl \partial^\beta (\varphi_t^\Lambda -\tphi_t^\Lambda)\Vert_2^2 
		&\leq \Vert \tl\partial^{\beta}( \varphi_t^\Lambda-\tphi_t^\Lambda)\Vert_2 C \kappa_{|\beta|,k}(\Lambda) \\
		&\quad + C \Vert \tl\partial^{\beta} (\varphi_t^\Lambda - \tphi_t^\Lambda) \Vert_2^2 \,.
	\end{align*}
	From here it follows with Grönwall that for $|\beta|=m-1$
	\begin{align*}
		\Vert \tl \partial^\beta (\varphi_t^\Lambda -\tphi_t^\Lambda)\Vert_2 \leq C \kappa_{|\beta|,k}(\Lambda) \,.
	\end{align*}
	As mentioned above, this proves the lemma.
\end{proof}


\section{Supplementary Proofs}

\subsection{Remainder Term Estimate for Theorem~\ref{thm:tildePsi-psiBF-time-estimate}}

We begin with \cref{lem:qqpp-term-estimate}, which provides one of the key estimates underlying the remainder term estimate in \cref{lem:Tech-RN-estimate}, as well as the excitation number estimate in \cref{lem:Tech-particle-number-estimate}, and consequently \cref{thm:tildePsi-psiBF-time-estimate}.  In particular, \cref{lem:qqpp-term-estimate} yields a result similar to \cite[Lemma 3.5, Part $\gamma^{c}_N$]{PPS20}.

\begin{lemma} \label{lem:qqpp-term-estimate}
For all volumes $\Lambda\geq1$, let $\varphi_t^\Lambda$ be the condensate satisfying $\Vert\varphi_0^\Lambda\Vert_2=\Lambda^{1/2}$.
Let $\tilde{\Phi},\Phi\in L^2(\BR^3,D(\cN))\subset L^2(\BR^3, \cF(L^2))$.
Then
\begin{align}
	\Big\vert \Big\langle \tilde{\Phi} , \sum_{j,k\geq1} \Lambda{V}_{jk00} a_j^* a_k^* \Phi \Big\rangle \Big\vert
	&\leq \Vert \varphi_t^\Lambda \Vert _{\infty}^2 \Vert V \Vert_1 \Vert \cN^{1/2} \tilde{\Phi} \Vert \Vert \cN^{1/2} \Phi \Vert  \nn \\
	&\quad + \Lambda^{1/2} \Vert \varphi_t^\Lambda \Vert _{\infty} \Vert V \Vert_2 \Vert \cN^{1/2} \tilde{\Phi} \Vert \Vert \Phi \Vert \,. \label{eq:qqpp-term-estimate-for-HBF}
\end{align}
\end{lemma}

\begin{proof}[Proof of \cref{lem:qqpp-term-estimate}]
For the proof we follow the ideas of \cite[Lemma 3.5, Part $\gamma^{c}_N$]{PPS20}.
We compute, using the definition of $\Lambda{V}_{jk00}$ and Cauchy-Schwarz,
\begin{align}
	&\Big\vert \Big\langle \tilde{\Phi} , \sum_{j,k\geq1} \Lambda{V}_{jk00} a_j^* a_k^* \Phi \Big\rangle \Big\vert  \nn \\
	&= \Lambda \Big\vert\sum_{j,k\geq1} \Big\langle \tilde{\Phi} , \int dy_1dy_2 V(y_1-y_2) u_0(y_1) u_0(y_2) u_j^*(y_1) u_k^*(y_2) a_j^* a_k^* \Phi  \Big\rangle \Big\vert \nn \\
	&\leq \Lambda \int dy_1 \Vert \sum_{j\geq1} u_j(y_1) a_j \tilde{\Phi} \Vert \cdot | u_0(y_1)| \label{eq:Sum-ujaj-term}\\
	&\qquad  \times \Vert \sum_{k\geq1} \int dy_2 V(y_1-y_2)  u_0(y_2) u_k^*(y_2) a_k^* \Phi \Vert \,. \label{eq:Vu0u0ukakDagger-term}
\end{align}
Both terms \eqref{eq:Sum-ujaj-term} and \eqref{eq:Vu0u0ukakDagger-term} have one  $u_0=u_0^\Lambda =\varphi_t^\Lambda/\Lambda^{1/2}$ term. We would like to bound both using $\Vert u_0^\Lambda \Vert_\infty \leq C\Lambda^{-1/2}$ to cancel the pre-factor $\Lambda$ in \eqref{eq:Sum-ujaj-term}, but this is not possible for some of the terms appearing below.
\\We start with the estimate of \eqref{eq:Vu0u0ukakDagger-term}. We want to estimate the $a_k^*$ operator by $\cN^{1/2}$. Since this cannot be done directly for $a_k^*$, we first replace it by $a_k$ using the CCR.
Using the notation $V_{y_1}(y_2):= V(y_1-y_2)$, we find
\begin{align}
	\eqref{eq:Vu0u0ukakDagger-term}^2 &= \Vert a^*(Q_t^\Lambda V_{y_1} u_0^\Lambda) \Phi \Vert^2  \nn \\
	 &\leq \pscal{ \Phi , a^*(Q_t^\Lambda V_{y_1} u_0^\Lambda) a(Q_t^\Lambda V_{y_1} u_0^\Lambda)  \Phi	} +\Vert Q_t^\Lambda V_{y_1} u_0^\Lambda\Vert_2^2 \Vert\Phi\Vert^2 \nn \\
	&\leq  \Vert a(Q_t^\Lambda V_{y_1} u_0^\Lambda) \Phi	\Vert^2 + \Vert V\Vert_2^2 \Vert u_0^\Lambda\Vert^2_{\infty} \Vert \Phi \Vert^2 \,. \label{eq:Vu0-L2}
\end{align}
The change in \eqref{eq:Vu0u0ukakDagger-term} from $a_k^*$ to $a_k$ is thus at the expense of the additional term  in \eqref{eq:Vu0-L2}, coming from the CCR. We will see below that this additional term will give us the largest order in $\Lambda$. 
For its estimate it is important that the integral over $y_2$ is inside the norm in \eqref{eq:Vu0u0ukakDagger-term}.
To estimate the first term in \eqref{eq:Vu0-L2} together with \eqref{eq:Sum-ujaj-term} we pull $V$ again outside of the annihilation operator and split all terms symmetrically. Here we can estimate both $u_0$ in the $\Vert\,.\,\Vert_\infty$-Norm, since we can regularise the integrals with $|V|^{1/2}(y_1-y_2)$ 
\begin{align}
	&\eqref{eq:Sum-ujaj-term}\cdot \Vert a(Q_t^\Lambda V_{y_1} u_0^\Lambda) \Phi	\Vert 
	\leq \Lambda \int dy_1dy_2 |V|(y_1-y_2) \lvert u_0^\Lambda(y_2)\rvert \lvert u_0^\Lambda(y_1)\rvert \nn \\
	&\qquad \qquad \qquad\qquad \qquad\qquad\times \lVert \sum_{j\geq1} u_j(y_1) a_j \tilde{\Phi} \rVert    \lVert \sum_{k\geq1}   u_k(y_2)  a_k \Phi	\rVert \nn \\ 
	&\leq \Vert \varphi_t^\Lambda \Vert_\infty^2  \bigg( \int dy_1dy_2 \Vert   |V|^{1/2}(y_1-y_2) \sum_{j\geq1} u_j(y_1) a_j \tilde{\Phi} \Vert^2 \bigg)^{1/2}   \nn \\
	&\qquad \times \bigg( \int dy_1dy_2 \Vert  |V|^{1/2}(y_1-y_2) \sum_{k\geq1}   u_k(y_2)  a_k \Phi	\Vert \bigg)^{1/2} \nn \\ 
	&\leq \Vert \varphi_t^\Lambda \Vert_\infty^2 \Vert V\Vert_1 \Vert \cN^{1/2} \tilde{\Phi} \Vert   \Vert \cN^{1/2} \Phi \Vert \,, \label{eq:qqpp-term-estimate-VL1-term-estimate}
\end{align}
 in the last step we have used
\begin{align}
	\int dy_2 \Vert   \sum_{k\geq1}   u_k(y_2)  a_k \Phi	\Vert^2
	&\leq \Big\langle\Phi, \sum_{k\geq1}a_k^* a_k \Phi \Big\rangle \leq  \Vert \cN^{1/2} \Phi \Vert^2 \,. \label{eq:Int-sum-uk-ak-estimate}
\end{align}

Now we estimate the second term in \eqref{eq:Vu0-L2} in combination with \eqref{eq:Sum-ujaj-term}.
Unlike above, here the potential term $\Vert V\Vert_2$ no longer depends on $y_1$, and we  cannot estimate $u_0^\Lambda(y_1)$ by $\Vert u_0^\Lambda\Vert_\infty$, as we need it to use Cauchy-Schwarz in the integral over $y_1$:
\begin{align}
	&\eqref{eq:Sum-ujaj-term}\cdot  \Vert V\Vert_2 \Vert u_0^\Lambda\Vert_{\infty} \Vert \Phi \Vert\leq \Lambda \Vert u_0^\Lambda\Vert_2   \bigg( \int dy_1 \Vert \sum_{j\geq1} u_j(y_1) a_j \tilde{\Phi} \Vert^2  \bigg)^{1/2}  \Vert u_0^\Lambda \Vert_\infty   \Vert V \Vert_2   \Vert \Phi \Vert  \nn \\
	&\leq  \Lambda^{1/2} \Vert \varphi_t^\Lambda \Vert_\infty \Vert V \Vert_2   \Vert \Phi \Vert   \Vert \cN^{1/2} \tilde{\Phi} \Vert \,, \label{eq:qqpp-term-estimate-VL2-term-estimate}
\end{align}
where we have used \eqref{eq:Int-sum-uk-ak-estimate}.
 We finally conclude
\begin{align*}
	&\Big\vert \Big\langle \tilde{\Phi} , \sum_{j,k\geq1} \Lambda{V}_{jk00} a_j^* a_k^* \Phi \Big\rangle\Big\vert  \leq \eqref{eq:Sum-ujaj-term}\cdot \eqref{eq:Vu0-L2}^{1/2}
	\leq \eqref{eq:qqpp-term-estimate-VL1-term-estimate} + \eqref{eq:qqpp-term-estimate-VL2-term-estimate}  \\
	&\leq \Vert \varphi_t^\Lambda \Vert_\infty^2 \Vert V\Vert_1 \Vert \cN^{1/2} \tilde{\Phi} \Vert   \Vert \cN^{1/2} \Phi \Vert 
	+ \Lambda^{1/2} \Vert \varphi_t^\Lambda \Vert_\infty \Vert V \Vert_2     \Vert \cN^{1/2} \tilde{\Phi} \Vert \Vert \Phi \Vert \,,
\end{align*}
which proves the claim.
\end{proof}

We proceed with the estimate of the remainder term $R_N$ as given in \eqref{eq:Remainder-Term-Estimate}.

\begin{lemma}[Remainder Term Estimate] \label{lem:Tech-RN-estimate}
For all volumes $\Lambda\geq1$, let $\varphi_t^\Lambda$ be the solution of the Hartree equation \eqref{eq:Hartree} satisfying \cref{con:Initial condition}. Then for all $ T\geq0$, there exists a constant $C>0$ such that for all $ \Lambda,\rho\geq1$ and $-T\leq t\leq T$ we have
	\begin{align}
		2\mathrm{Im} \pscal{\psi^{\rm ex}_t, R_N \Phi_t } &\leq  C \rho^{-\frac{1}{2}} \Vert \psi^{\rm ex}_t -\Phi_t \Vert  \Vert (\cN+1)^{\frac{3}{2}} \Phi_t \Vert \nn \\
		&\quad +C \rho^{-1} \Vert \psi^{\rm ex}_t -\Phi_t \Vert \Vert (\cN + 1)^2 \Phi_t \Vert \,.
	\end{align} 
\end{lemma}

\begin{proof}[Proof of \cref{lem:Tech-RN-estimate}]  
 Setting $\psi^{\rm ex}_t - \Phi_t =: \tilde{\Phi}_t$,  we write
\begin{align*}
	2\mathrm{Im} \pscal{\psi^{\rm ex}_t,  R_N \Phi_t } = 2\mathrm{Im} \pscal{\psi^{\rm ex}_t - \Phi_t ,  R_N \Phi_t } = 2\mathrm{Im} \pscal{\tilde{\Phi}_t ,  R_N \Phi_t } \,.
\end{align*}
We estimate this expression term by term, considering each $R_{i,N}$ separately.
\\\underline{\textit{To $R_{4,N}$:}}
\\We estimate
\begin{align}
	& 2 \mathrm{Im} \Big\langle\tilde{\Phi}_t , - \frac{1}{\sqrt{\rho}} \cN W\ast \frac{|\varphi_t^\Lambda|^2}{\Lambda}(x) \Phi_t \Big\rangle 
		\leq  \frac{2}{\sqrt{N\Lambda}} \Vert \varphi_t^\Lambda\Vert_\infty^2 \Vert W \Vert_1 \Vert \tilde{\Phi}_t \Vert \Vert \cN \Phi_t \Vert \,. \label{eq:R4N-estimate-part-1}
\end{align}
Using that $\Vert \rmd\Gamma(A)\psi\Vert\leq \Vert A\Vert_{\rm op} \Vert \cN\psi\Vert$ we get for a fixed $x$
\begin{align}
	\frac{2}{\sqrt{\rho}}\mathrm{Im} \pscal{\tilde{\Phi}_t ,  \rmd\Gamma(Q_t^\Lambda W_x Q_t^\Lambda) \Phi_t }	= \frac{2}{\sqrt{\rho}} \Vert W\Vert_\infty \Vert \tilde{\Phi}_t \Vert \Vert \cN  \Phi_t \Vert \,. \label{eq:R4N-estimate-part-2-step-4}
\end{align}
We get the $R_{4,N}$ estimate from \eqref{eq:R4N-estimate-part-1} and \eqref{eq:R4N-estimate-part-2-step-4}:
\begin{align}
	2\mathrm{Im} \pscal{\tilde{\Phi}_t ,  R_{4,N} \Phi_t } = 2\sqrt{\frac{\Lambda}{N}} \left( \Vert\varphi_t^\Lambda\Vert_{\infty}^2 +1\right)  (\Vert W \Vert_1 + \Vert W\Vert_\infty) \Vert \tilde{\Phi}_t \Vert \Vert \cN \Phi_t \Vert \,. \label{eq:R4N-estimate}
\end{align}
\\\underline{\textit{To $R_{3,N}$:}}
\\ We insert $\cN^{-1/2} \cN^{1/2}$ to obtain
\begin{align}
	&2\mathrm{Im} \bigg\langle\tilde{\Phi}_t,  \bigg( \frac{\sqrt{N-\cN}}{\sqrt{N}} -1 \bigg) a (Q_t^\Lambda W_x \varphi_t^\Lambda) \Phi_t \bigg\rangle \nn \\
	&\leq 2 \Big\Vert \bigg( \frac{\sqrt{N-\cN}}{\sqrt{N}} -1 \bigg) (\cN+1)^{-1/2}\tilde{\Phi}_t \Big\Vert \Vert a (Q_t^\Lambda W_x \varphi_t^\Lambda) \cN^{1/2}\Phi_t \Vert \nn \\
	&\leq 2 \Vert W \Vert_2 \Vert \varphi_t^\Lambda \Vert_\infty N^{-1/2}\Vert \tilde{\Phi}_t \Vert  \Vert \cN \Phi_t \Vert\,, \label{eq:R3N-a-estimate} 
\end{align}
where  we have used that $\Vert a (Q_t^\Lambda W_x \varphi_t^\Lambda) \Phi_t \Vert\leq \sup_x\Vert Q_t^\Lambda W_x \varphi_t^\Lambda \Vert_2 \Vert \cN \Phi_t \Vert$ and $\Vert (\sqrt{N-\cN} -\sqrt{N})\psi \Vert^2\leq \pscal{\psi , \cN  \psi }$.
For the Hermitian conjugate term in $R_{3,N}$ we find an estimate similar to \eqref{eq:R3N-a-estimate} 
leading to
\begin{align}
	\pscal{\tilde{\Phi}_t,  R_{3,N} \Phi_t} \leq 4N^{-1/2} \Vert \varphi_t^\Lambda\Vert_\infty \Vert W \Vert_2   \Vert \tilde{\Phi}_t\Vert  \Vert (\cN + 1) \Phi_t \Vert \,. \label{eq:R3N}
\end{align}

\underline{\textit{To $R_{2,N}$:}}
\begin{align}
	&2\mathrm{Im} \Big\langle\tilde{\Phi}_t, R_{2,N}  \Phi_t \Big\rangle = 2\mathrm{Im}\Big\langle \tilde{\Phi}_t, \frac{1}{2N} \sum_{mnpq\geq 1} \Lambda{V}_{mnpq} a_m^*a_n^* a_pa_q \cN^{-1}\cN \Phi_t \Big\rangle \nn \\
	&\leq  \frac{\Lambda}{N} \bigg(  \int dy_1 dy_2 \Vert \sum_{mn\geq 1} V(y_1-y_2) u_m(y_1)u_n(y_2) a_ma_n \cN^{-1}\tilde{\Phi}_t \Vert^2 \bigg)^{1/2} \nn \\
	&\qquad \times \bigg( \int dy_1 dy_2 \Vert \sum_{pq\geq 1} u_p(y_1)u_q(y_2) a_pa_q \cN\Phi_t \Vert^2 \bigg)^{1/2} \nn \\
	&\leq   \frac{\Lambda}{N}  \Vert V \Vert_\infty \Vert  \tilde{\Phi}_t \Vert  \Vert \cN^2 \Phi_t \Vert \,, \label{eq:R2N-estimate}
\end{align}
where in the last step we have used $\pscal{u_p,u_m}=\delta_{p,m}$ to get
\begin{align}
	&  \int dy_1 dy_2 \Vert \sum_{pq\geq 1} u_p(y_1)u_q(y_2) a_pa_q \psi \Vert^2  
	=  \sum_{pq\geq 1} \pscal{\psi , a_p^* a_q^* a_pa_q \psi}
\nn \\
	={}&  \sum_{p\geq 1} \pscal{ a_q\psi , \cN a_q \psi} = \pscal{\psi , (\cN -1) \cN \psi} \,.
\end{align}
\underline{\textit{To $R_{1,N}$:}}
\\The remainder $R_{1,N}$ consists of several terms,
 see \eqref{eq:Def-R1N}. In the following we will estimate each of these terms individually.
\\\underline{To the $\mu_t^\Lambda$ terms in $R_{1,N}$:}
\begin{align}
	2\mathrm{Im} \pscal{\tilde{\Phi}_t, \left[ \frac{1}{2}\mu_t^\Lambda \frac{\cN}{N} + \frac{1}{2} \mu_t^\Lambda\frac{\cN^2}{N} + \mathrm{h.c.} \right] \Phi_t } 
	&\leq \frac{2}{N}\Vert \varphi_t^\Lambda \Vert_\infty^2 \Vert V\Vert_1  \Vert \tilde{\Phi}_t \Vert \Vert \cN^2 \Phi_t \Vert \,, \label{eq:R1N-mu-estimate}
\end{align}
where we have used $\BR\ni \mu_t^\Lambda = \frac{1}{2}\pscal{ \frac{\varphi_t^\Lambda}{\Lambda^{1/2}}, V\ast |\varphi_t^\Lambda|^2  \frac{\varphi_t^\Lambda}{\Lambda^{1/2}}} \leq \frac{1}{2} \Vert \varphi_t^\Lambda \Vert_\infty^2 \Vert V\Vert_1$.
\\\underline{To the $a(.),a^*(.)$ Terms in $R_{1,N}$:}
\begin{align}
	&2\mathrm{Im} \pscal{  \tilde{\Phi}_t ,  \left[ - \frac{(\cN +1)\sqrt{N-\cN}}{N} a \left( Q_t^\Lambda V\ast |\varphi_t^\Lambda|^2 \frac{\varphi_t^\Lambda}{\Lambda^{1/2}} \right) + \mathrm{h.c.} \right]  \Phi_t }  \nonumber \\
	&\leq \frac{4}{\sqrt{N}} 	\Vert V\Vert_1 \Vert \varphi_t^\Lambda \Vert_\infty^2 \Vert \tilde{\Phi}_t\Vert  \Vert (\cN +1)^{3/2} \Phi_t \Vert \,. \label{eq:R1N-a-a-STAR-estimate}
\end{align}
\\\underline{To the $a_j^*a_k$ Terms in $R_{1,N}$:}
\\We get with $\Vert \rmd\Gamma(A)\psi\Vert\leq \Vert A\Vert_{\rm op} \Vert \cN\psi\Vert$ that
\begin{align}
	&2\mathrm{Im} \Big\langle \tilde{\Phi}_t, - \frac{1}{2}  \rmd\Gamma (Q_t^\Lambda V\ast |\varphi_t^\Lambda|^2 Q_t^\Lambda)  \frac{\cN}{N}  \Phi_t \Big\rangle 
	 \leq \frac{\Vert \varphi_t^\Lambda \Vert^2_\infty}{N} \Vert V\Vert_1 \Vert  \tilde{\Phi}_t \Vert  \Vert \cN^{2} \Phi_t\Vert \,. \label{eq:R1N-qpqp-estimate}
\end{align}
The $K_1(t)$ term in $R_{1,N}$ is estimated in the following
\begin{align}
	& 2\mathrm{Im} \Big\langle\tilde{\Phi}_t, - \frac{1}{2} \frac{\cN}{N}  \rmd\Gamma(Q_t^\Lambda K_1^\Lambda(t)Q_t^\Lambda)  \Phi_t \Big\rangle \nn \\
	&= 2\mathrm{Im} \Big\langle\tilde{\Phi}_t, - \frac{1}{2}  \sum_{mn\geq 1} \Lambda{V}_{0mn0} a_m^*a_n \frac{\cN}{N}  \Phi_t \Big\rangle \nn \\
	&\leq \frac{\Vert \varphi_t^\Lambda \Vert_\infty^2}{N}   \nn \\
	  &\times\int dy_1dy_2 \Big\vert  \sum_{m,n\geq1}\Big\langle   |V|^{\frac{1}{2}}(y_1 - y_2) u_m(y_1) a_m \tilde{\Phi}_t ,   |V|^{\frac{1}{2}}(y_1 - y_2) u_n(y_2) a_n \cN  \Phi_t  \Big\rangle \Big\vert \nn \\
	 &\leq \frac{\Vert \varphi_t^\Lambda \Vert^2_\infty}{N} \Vert V\Vert_1 \Vert \tilde{\Phi}_t \Vert  \Vert (\cN + 1)^{2} \Phi_t\Vert \,. \label{eq:R1N-qpqp-h.c.-estimate}
\end{align}
Putting both estimates \eqref{eq:R1N-qpqp-estimate} and \eqref{eq:R1N-qpqp-h.c.-estimate} together we conclude
\begin{align}
	&2\mathrm{Im} \bigg\langle\tilde{\Phi}_t, \Big[- \frac{1}{2}  \sum_{mn\geq 1} \Lambda{V}_{m0n0} a_m^*a_n  \frac{\cN}{N}  - \frac{1}{2}  \sum_{mn\geq 1} \Lambda{V}_{0mn0} a_m^*a_n \frac{\cN}{N}  +\mathrm{h.c.} \Big]  \Phi_t \bigg\rangle \nn \\
	&\leq 4 \frac{\Vert \varphi_t^\Lambda \Vert_\infty^2}{N} \Vert V \Vert_1 \Vert \tilde{\Phi}_t \Vert \Vert (\cN +1 )^2 \Phi_t \Vert \,, \label{R1N-a-STAR-a-estimate}
\end{align}
where we used that the h.c. terms coincide with the original terms due to $V(y_1-y_2)=V(y_2-y_1)$.
\\\underline{To the $a_m^*a_n^*,a_m a_n$ terms in $R_{1,N}$:}
\\The estimates given here rely heavily on \cref{lem:qqpp-term-estimate}. By using \cref{lem:qqpp-term-estimate},  $0\leq -\sqrt{(N-\cN) (N-\cN-1)}+N \leq -(N-\cN-1) +N = \cN +1$ and inserting $(\cN +2)^{-1/2} (\cN +2)^{1/2} $ before $\Phi_t$ we get
\begin{align}
	&2\mathrm{Im} \bigg\langle \tilde{\Phi}_t , \frac{1}{2} \sum_{mn\geq1} \Lambda{V}_{mn00} a_m^*a_n^* \bigg( \frac{\sqrt{(N-\cN) (N-\cN-1)}}{N} -1 \bigg) \Phi_t \bigg\rangle \nn \\
	&\leq  \frac{1}{N} C \Vert \varphi_t^\Lambda \Vert _{\infty}^2 \Vert V \Vert_1 \Vert  \tilde{\Phi}_t \Vert \Vert (\cN +1)^2 \Phi_t \Vert  \nn \\
	&\quad +  \frac{\Lambda^{1/2}}{N} C \Vert \varphi_t^\Lambda \Vert _{\infty} \Vert V \Vert_2 \Vert  \tilde{\Phi}_t \Vert \Vert  (\cN +1)^{3/2} \Phi_t \Vert \,. \label{R1N-a-STAR-a-STAR-estimate}
\end{align}
The conjugate term can be estimated with the same argument as above with $\tilde{\Phi}_t$ and $\Phi_t$ interchanged and inserting $(\cN +1)^{-3/2} (\cN +1)^{3/2} $ before $\Phi_t$:
\begin{align}
	&2\mathrm{Im} \bigg\langle \tilde{\Phi}_t , \bigg( \frac{\sqrt{(N-\cN) (N-\cN-1)}}{N} -1 \bigg)\frac{1}{2} \sum_{mn\geq1} \Lambda{V}_{00mn} a_ma_n \Phi_t \bigg\rangle \nn \\
	&\leq \frac{1}{N} \Vert \varphi_t^\Lambda \Vert _{\infty}^2 \Vert V \Vert_1 \Vert   \tilde{\Phi}_t \Vert \Vert (\cN +1)^{2} \Phi_t \Vert  \nn \\
	&\quad + \frac{\Lambda^{1/2}}{N} \Vert \varphi_t^\Lambda \Vert _{\infty} \Vert V \Vert_2 \Vert  (\cN +3)^{-1/2}  \tilde{\Phi}_t \Vert \Vert (\cN +1)^{2} \Phi_t \Vert  \nn \\
	&\leq C \frac{\Lambda^{1/2}}{N} (\Vert \varphi_t^\Lambda \Vert _{\infty}+ \Vert \varphi_t^\Lambda \Vert _{\infty}^2) ( \Vert V \Vert_1 + \Vert V \Vert_2 ) \Vert   \tilde{\Phi}_t \Vert \Vert (\cN +1)^{2} \Phi_t \Vert\,.  \label{R1N-a-a-estimate}
\end{align}
Note that \eqref{R1N-a-a-estimate} has larger prefactor than \eqref{R1N-a-STAR-a-STAR-estimate}, namely $\Lambda^{1/2}/N$ instead of $1/N$ for the $(\cN+1)^2$ term. However, this is not important, as $R_{2,N}$ already introduces a larger prefactor $\Lambda/N$ for a $(\cN+1)^2$ term, due to \eqref{eq:R2N-estimate}.
\\ \underline{To the $a_m^*a_na_p, a_m^*a_n^*a_p$ terms in $R_{1,N}$:}
\begin{align}
	2 \mathrm{Im} \bigg\langle \tilde{\Phi}_t , \sum_{mnp\geq1} \Lambda{V}_{0mnp} \frac{\sqrt{N- \cN}}{N} a_m^* a_na_p \Phi_t \bigg\rangle 
	\leq 2 \Lambda^{1/2}\Vert \varphi_t \Vert_\infty&  \nn \\
	\times  \bigg(  \int dy_1dy_2 \Vert \sum_{m\geq1} V(y_1-y_2) u_m(y_2) a_m \frac{\sqrt{N- \cN}}{N} \tilde{\Phi}_t \Vert^2 \bigg)^{1/2}& \label{R1N-a-STAR-a-a-1-term} \\
	\times \bigg(  \int dy_1dy_2 \Vert \sum_{np\geq1} u_n(y_1) u_p(y_2) a_na_p  \Phi_t \Vert^2 \bigg)^{1/2}& \,. \label{R1N-a-STAR-a-a-2-term}
\end{align}
Now we estimate both terms \eqref{R1N-a-STAR-a-a-1-term} and \eqref{R1N-a-STAR-a-a-2-term} separately.
\begin{align}
	\eqref{R1N-a-STAR-a-a-1-term}^2 
	&= \int dy_1dy_2 \sum_{mn\geq1} V^2(y_1-y_2)u_n^*(y_2) u_m(y_2)  \Big\langle\tilde{\Phi}_t, a_n^* a_m \frac{N- \cN}{N^2} \tilde{\Phi}_t \Big\rangle \nn \\
	&= \Vert V\Vert_2^2 \sum_{m\geq1} \Big\langle\tilde{\Phi}_t, a_m^* a_m \frac{N- \cN}{N^2} \tilde{\Phi}_t \Big\rangle \leq \Vert V\Vert_2^2 \Big\langle\tilde{\Phi}_t, \frac{\cN}{N} \tilde{\Phi}_t \Big\rangle \,, \label{R1N-a-STAR-a-a-1-term-estimate}
\end{align}
where we have used that $(N-\cN)/N\leq1$. Now we estimate the second term 
\begin{align}
	\eqref{R1N-a-STAR-a-a-2-term}^2 
	={}&  \int dy_1dy_2 \sum_{mpnq\geq1} u_n(y_1) u_p(y_2) u_m^*(y_1) u_q^*(y_2)  \pscal{ \Phi_t , a_m^*a_q^* a_n a_p  \Phi_t } \nn \\
	={}&  \sum_{np\geq1} \pscal{ \Phi_t , a_n^*a_p^* a_n a_p  \Phi_t } \leq \pscal{ \Phi_t ,\cN^2 \Phi_t } \,. \label{R1N-a-STAR-a-a-2-term-estimate}
\end{align}
We conclude
\begin{align}
	&2 \mathrm{Im} \bigg\langle \tilde{\Phi}_t , \sum_{mnp\geq1} \Lambda{V}_{0mnp} \frac{\sqrt{N- \cN}}{N} a_m^* a_na_p \Phi_t \bigg\rangle 
	\nn \\
	&\leq 2 \Lambda^{1/2} \Vert \varphi_t^\Lambda \Vert_\infty   \eqref{R1N-a-STAR-a-a-1-term-estimate}^{1/2} \cdot \eqref{R1N-a-STAR-a-a-2-term-estimate}^{1/2} \nn \\
	&\leq 2 \Lambda^{1/2} \Vert \varphi_t^\Lambda \Vert_\infty  \frac{1}{N^{1/2}}\Vert V\Vert_2 \Vert\cN^{1/2} \tilde{\Phi}_t  \Vert  \Vert \cN\Phi_t \Vert \,. \label{R1N-a-STAR-a-a-estimate-N+-split}
\end{align}
Analogously one finds the same estimate for the conjugate term.
Now by inserting $(\cN +1)^{-1/2} (\cN +1)^{1/2} $ before $\Phi_t$ at the start of our estimate in \eqref{R1N-a-STAR-a-a-estimate-N+-split} and $(\cN +1)^{-1} (\cN +1)$ for the  conjugate term, we conclude
\begin{align}
		&2 \mathrm{Im} \bigg\langle \tilde{\Phi}_t , \Big(\sum_{mnp\geq1} \Lambda{V}_{0mnp} \frac{\sqrt{N- \cN}}{N} a_m^* a_na_p + \mathrm{h.c.} \Big)\Phi_t \bigg\rangle \nn \\
		&\leq 2 \frac{\Lambda^{1/2}}{N^{1/2}} \Vert \varphi_t^\Lambda \Vert_\infty  \Vert V\Vert_2 \Vert \tilde{\Phi}_t  \Vert  \Vert (\cN +1)^{3/2}\Phi_t \Vert \,. \label{R1N-a-STAR-a-a-estimate}
\end{align}
\underline{$R_{1,N}$ estimate conclusion:}
\\We now collect all estimates for the terms in $R_{1,N}$, i.e. \eqref{eq:R1N-mu-estimate}, \eqref{eq:R1N-a-a-STAR-estimate}, \eqref{R1N-a-STAR-a-estimate}, \eqref{R1N-a-STAR-a-STAR-estimate}, \eqref{R1N-a-a-estimate}, \eqref{R1N-a-STAR-a-a-estimate} and
group contributions by the order of $\cN +1$. In doing so, we use the bounds $1/N\leq \Lambda^{1/2}/N \leq \Lambda^{1/2}/N^{1/2}$, valid for $\Lambda,N\geq1$. This leads to
\begin{align}
	&2\mathrm{Im} \pscal{ \tilde{\Phi}_t, R_{1,N} \Phi_t }  \nn \\
	&\leq C\left(\frac{\Lambda}{N}\right)^{1/2} \left( \Vert \varphi_t^\Lambda\Vert_\infty +\Vert \varphi_t^\Lambda\Vert_\infty ^2 \right) \left( \Vert V\Vert_1 + \Vert V\Vert_2 \right)   \Vert \tilde{\Phi}_t  \Vert  \Vert (\cN +1)^{3/2}\Phi_t \Vert  \nn \\
	&+ C \frac{\Lambda^{1/2}}{N} \left( \Vert \varphi_t^\Lambda\Vert_\infty +\Vert \varphi_t^\Lambda\Vert_\infty ^2 \right) \left( \Vert V\Vert_1 + \Vert V\Vert_2 \right)   \Vert \tilde{\Phi}_t  \Vert  \Vert (\cN +1)^2 \Phi_t \Vert \,, \label{eq:R1N-estimate}
\end{align}
which concludes our estimate for $R_{1,N}$.

Collecting the different $R_{i,N}$ estimates done above, i.e. \eqref{eq:R4N-estimate}, \eqref{eq:R3N}, \eqref{eq:R2N-estimate}, \eqref{eq:R1N-estimate} and using $\Lambda\geq1$ to simplify, yields the $R_N$ estimate after applying $\Vert \varphi_t^\Lambda\Vert_\infty \leq C$ (see \cref{prop:phi-propagation}).
\end{proof}


\subsection{Proof of Lemma~\ref{lem:Tech-particle-number-estimate}} \label{sec:Proof-lem:Tech-particle-number-estimate}

The proof of \cref{lem:Tech-particle-number-estimate}, estimating the excitation number, is given below.

\begin{proof}[Proof of \cref{lem:Tech-particle-number-estimate}]
We use a Grönwall estimate. To shorten the notation we write $\psi_{\Lambda,t}^{\mathrm{BF}}=:\psi_t$.  We consider 
\begin{align}
	&\BR\ni \partial_t \pscal{ \psi_t \,,\, (\cN+1)^n \psi_t } =\mathrm{Re} \pscal{ \psi_t \,,\, -\rmi \left[(\cN+1)^n , H^{\rm BF}_\Lambda\right] \psi_t } \nn\\
	 &= \mathrm{Im}  \Big\langle \psi_t \,,\, \Big[(\cN+1)^n , \frac{1}{2} \sum_{mn} \left( (K_2^\Lambda(t)J)_{mn} a_m^*a_n^* +{\rm h.c.}  \right)\Big] \psi_t \Big\rangle \label{eq:N-a*a*-Term} \\
	&\quad +  \mathrm{Im}  \pscal{ \psi_t \,,\, \left[(\cN+1)^n , a^*(Q_t^\Lambda W_x\varphi_t^\Lambda) + a(Q_t^\Lambda W_x\varphi_t^\Lambda)   \right] \psi_t } \label{eq:N-single-a-a*-Term} \,. 
\end{align}
We now estimate both terms \eqref{eq:N-a*a*-Term} and \eqref{eq:N-single-a-a*-Term} separately. 
For \eqref{eq:N-single-a-a*-Term} we use $(\cN+b)^{m}-a^{m} \leq (1+b)^{m} \cN^{m-1}$,  $\forall m\in\mathbb{N}_0$ and $b\geq0$, to obtain
\begin{align}
&\eqref{eq:N-single-a-a*-Term}
	 = 2\mathrm{Im}  \pscal{ \psi_t \,,\,   a^*(Q_t^\Lambda W_x\varphi_t^\Lambda)\{ (\cN+2)^n -(\cN+1)^n \} \psi_t} \nn \\
	 &= 2\mathrm{Im}  \pscal{ \{ (\cN+1)^n -\cN^n \}^{1/2} \psi_t \,,\,   a^*(Q_t^\Lambda W_x\varphi_t^\Lambda)\{ (\cN+2)^n -(\cN+1)^n \}^{1/2} \psi_t} \nn \\
	&\leq C  \Vert W \Vert_2 \Vert (\cN +1)^{\frac{n-1}{2}} \psi_t \Vert^2     +  C \Vert W \Vert_2 \Vert \varphi_t^\Lambda\Vert_\infty^2 \Vert  (\cN +1)^{\frac{n}{2}} \psi_t \Vert^2 \,.  \label{eq:N-single-a-a*-Term-estimate}
\end{align}
The estimate of \eqref{eq:N-a*a*-Term} follows from \cref{lem:qqpp-term-estimate} and the identity\\ $\sum_{mn\geq0} (K_2^\Lambda(t)J)_{mn} a_m^*a_n^*=\sum_{mn\geq1}\Lambda{V}_{mn00} a_m^*a_n^*,$ yielding 
\begin{align}
\eqref{eq:N-a*a*-Term}\leq{}& C( \Vert V\Vert_1 +\Vert V\Vert_2) \Vert \varphi_t^\Lambda\Vert_\infty^2  \Vert  (\cN+1)^{\frac{n}{2}} \psi_t \Vert^2   \nn\\
& + C\Lambda \Vert V\Vert_2 \Vert (\cN+1)^{\frac{n-1}{2}} \psi_t \Vert^2 \,. \label{eq:N-a*a*-Term-estimate}
\end{align}
This finally leads us to
\begin{align}
	& \partial_t \pscal{ \psi_t \,,\, (\cN+1)^n \psi_t } \leq \eqref{eq:N-a*a*-Term} +  \eqref{eq:N-single-a-a*-Term} 	\leq \eqref{eq:N-single-a-a*-Term-estimate} + \eqref{eq:N-a*a*-Term-estimate} \nn \\
	&\leq C \Vert \varphi_t^\Lambda \Vert_\infty^2  ( \Vert V\Vert_1 + \Vert V\Vert_2 + \Vert W\Vert_2 ) \Vert  (\cN +1)^{\frac{n}{2}} \psi_t \Vert^2 \nn\\
	&\quad +  C( \Lambda \Vert V\Vert_2 + \Vert W \Vert_2 ) \Vert (\cN+1)^{\frac{n-1}{2}} \psi_t \Vert^2 \,. \label{eq:Time-derivative-N-estimate}
\end{align}
From \eqref{eq:Time-derivative-N-estimate} and \eqref{eq:phi-L2infty-bound} we are now able to prove by induction on $n$, using Grönwall's inequality, that
\begin{align}
	\pscal{ \psi_t \,,\, (\cN+1)^n \psi_t } \leq C \pscal{ \psi_0 \,,\, \big(\Lambda + (\cN+1)\big)^{n} \psi_0 } \,,\quad \forall n\in\mathbb{N}_0 \,. \label{eq:N+1-estimate-with-induction}
\end{align}
\end{proof}

\subsection{Proof of Lemma~\ref{lem:Conditons-LNS15-Theorem-8}} \label{sec:Proof-lem:Conditons-LNS15-Theorem-8}

\begin{proof}[Proof of \cref{lem:Conditons-LNS15-Theorem-8}] 
Let $U_{\Lambda,t}^{\rm Bog}=U_{\cV_t^\Lambda}$ be the propagator of the Bogoliubov dynamics. We write 
\[\cZ_0^\Lambda=:\begin{pmatrix}
	c^\Lambda & J^*b^\Lambda J^* \\
	b^\Lambda & Jc^\Lambda J^* 
\end{pmatrix}\,, \quad \cV_t^\Lambda=:\begin{pmatrix}
	U_t^\Lambda & J^*V_t^\Lambda J^* \\
	V_t^\Lambda & JU_t^\Lambda J^* 
\end{pmatrix} \,. \]
Set $ f^\Lambda:=\big(t\mapsto (x\mapsto (c^\Lambda +J^*b^\Lambda )((U_t^\Lambda)^*- (V_t^\Lambda)^*J) Q_t^\Lambda W_x \varphi_t^\Lambda)\big)$ and note that $(\cV_t^\Lambda)^{-1}=S(\cV_t^\Lambda)^*S$. 
In \cref{lem:uniform-boundedness-F-Lambda} we prove the following regularity and bounds for $f^\Lambda$:
\begin{itemize}
	\item[a)] For almost all $x\in\mathbb{R}^d$ we have $(t\mapsto f_t^\Lambda(x,\,.\,))\in C^1(\mathbb{R}_t,L^2(\mathbb{R}^d_y))$.
	\item[b)] For all times $T\geq0$ there exists a constant $ C>0$ such that for all volumes $\Lambda\geq1$, $x\in\BR^3$ and $-T\leq t\leq T$
  		\begin{align}
  			   \sum_{|\beta|\leq M} (1+x^2)^{-1/4}\Vert  \partial^\beta_x f_{t,x}^\Lambda\Vert_{L^2(\BR^3)} +  (1+x^2)^{-1/4}\Vert \partial_t f_{t,x}^\Lambda \Vert_{L^2(\BR^3)} \leq C\,. \label{eq:f-estimate}
  		\end{align}
\end{itemize}

 The bound on the Hamiltonian by $\pm h_{\rm oc}$ in \eqref{eq:aux-eq-27} follows directly from $0\leq -\Delta_x\leq h_{\rm oc}$, $\pm A(f_{t,x}^\Lambda \oplus J f_{t,x}^\Lambda)\leq 2 (\sup_x (1+x^2)^{-1/4}\Vert  f_{t,x}^\Lambda \Vert_2)  (1+x^2)^{1/4}(\cN+1)^{1/2}$, $(1+x^2)^{1/4}(\cN+1)^{1/2}\leq h_{\rm oc}$ and the estimate \eqref{eq:f-estimate} on $f$.
	\\ The regularity of $\big(t\mapsto \big\langle \psi, \widetilde{H}^{\mathrm{BF}}_{\Lambda. \cZ_0}(t)  \psi \big\rangle \big)\in C^1(\BR,\BR)$ follows directly from the regularity of $f^\Lambda$. The bound in \eqref{eq:aux-eq-28} can immediately be seen by calculating the derivative $\frac{d}{dt} \big\langle \psi, \widetilde{H}^{\mathrm{BF}}_{\Lambda. \cZ_0}(t)  \psi \big\rangle $ explicitly. 
	\\Thus, it remains to verify the commutator estimate for $\widetilde{H}^{\mathrm{BF}}_{\Lambda,\cZ_0^\Lambda}(t)$ and $h_{\rm oc}^M$ as in \eqref{eq:LNS15-Thm-commutator-realtion}.  To this end, let $\psi \in \bigcup_{L\geq0} \bigoplus_{n=0}^L \sS(\BR^d,\BC) \bar{\otimes} L^2(\BR^d)^{\bar{\otimes}_s n}=:D$ and $f\in C_b^\infty(\BR^d,L^2)$. We prove that there exists a constant $C_{M}$ only dependent on $M$ such that
	\begin{align}
	\left\vert\pscal{\psi,[-\Delta_x,h_{\rm oc}^M] \psi}\right\vert \leq C_{M} \pscal{\psi,h_{\rm oc}^M \psi}\,, \label{eq:aux-eq-18} \\
	\left\vert  \pscal{ A(f_{t,x} \oplus J f_{t,x} ) \psi ,  h_{\rm oc}^M \psi   } -  \pscal{ h_{\rm oc}^M \psi , A(f_{t,x}\oplus J f_{t,x} )  \psi} \right\vert 
		\nn \\
		\leq C_{M} \Big(\sup_{x\in \BR^3} \sum_{|\beta|\leq M} (1+x^2)^{-1/4}\Vert \partial^\beta_x f_{t,x}\Vert_2\Big) \pscal{\psi ,h_{\rm oc}^M  \psi}\,. \label{eq:aux-eq-26}
	\end{align}
	These estimates can be extended to all $\psi\in D(h_{\rm oc}^M)$ and all $f$ satisfying $\sup_{x\in \BR^3} \sum_{|\beta|\leq M} (1+x^2)^{-1/4}\Vert \partial^\beta_x f_{t,x}\Vert_2<\infty$, and hence, in particular, to $f^\Lambda$, by a standard density argument. This will prove \eqref{eq:LNS15-Thm-commutator-realtion} and therefore \cref{lem:Conditons-LNS15-Theorem-8}.
It remains to prove \eqref{eq:aux-eq-18} and \eqref{eq:aux-eq-26}.
	
	We start with the proof of \eqref{eq:aux-eq-26}. 
We begin with the case of $M=:2m$ even. This allows us to symmetrically split $h_{\rm oc}^{2m}=:h^{2m}$ between the arguments of the scalar product. Using the recursive definition of the iterated commutator, namely $\mathrm{ad}_{h}^{(0)}(a^{\#}(f_x)):= a^{\#}(f_x)$ and $\mathrm{ad}_{h}^{(k)}(a^{\#}(f_x)) \psi:= \big[\mathrm{ad}_{h}^{(k-1)}(a^{\#}(f_x)),h \big] \psi$ for $k\in\BN_+$,
we commute the $h$ operators with $a^{\#}(f_x)$ step by step, applying this definition $m$-times. This yields an expansion of the commutator terms in \eqref{eq:aux-eq-26}:
\begin{align}
	& \Big\vert \pscal{ (a^{\#}(f_x))^*\psi ,  h^{2m} \psi   } - \pscal{ h^{2m}\psi ,  a^{\#}(f_x) \psi } \Big\vert \nn \\
	&= \Big\vert \sum_{k=0}^{m-1} {m \choose k}
	\Big\{ \big\langle h^k \psi ,  \mathrm{ad}_{h}^{(m-k)}(a^{\#}(f_x))  h^m \psi   \big\rangle \nn \\
	&\quad - \big\langle h^m \psi , (-1)^{m-k} \mathrm{ad}_{h}^{(m-k)}(a^{\#}(f_x))  h^k \psi   \big\rangle \Big\} \Big\vert\, , \label{eq:aux-eq-65}
\end{align}
where the $k=m$ terms cancel, as they have identical numbers of $h$ operators on both sides of the scalar product.
To proceed, we need the following representation for the iterated commutator (for $m-k\geq1$):
\begin{align}
	\mathrm{ad}_{h}^{(m-k)}(a^{\#}(f_x)) \psi = \sum_{\substack{|\beta|+|\gamma|+|\delta|+ l\leq 2(m-k) \\  |\gamma|+|\delta|+ l\leq 2(m-k)-1}} C_{\beta,\gamma,\delta,l}^{(m-k)} a^{\#}(D_x^\beta f_x)\cN^l x^\gamma D_x^\delta \psi\,. \label{eq:aux-eq-66}
\end{align}
This representation is not valid for $m-k=0$, as $\mathrm{ad}_{h}^{(m-k)}(a^{\#}(f_x))=a^{\#}(f_x)$. Thus, it is essential that the sum in \eqref{eq:aux-eq-65} only runs over $k\leq m-1$.

A key feature of \eqref{eq:aux-eq-66} is that the combined order of $\cN^l x^{\gamma}D_x^{\delta}$ is strictly less than the order of monomials in $h^{m-k}$, allowing us to ultimately close the estimate.
\eqref{eq:aux-eq-66} can be verified by induction and readily checked for $m-k=1$ using the definition of the harmonic oscillator $h$.

We now combine the representation \eqref{eq:aux-eq-66} with the inequality
\begin{align}
		(-iD_x)^\delta x^{\gamma} \cN^{2l}  x^{\gamma}(-iD_x)^\delta \leq C h^{|\delta|+|\gamma|+l} \,, \label{eq:aux-eq-67}
\end{align} 
which follows from the definition of the harmonic oscillator $h=-\Delta_x +x^2 +(\cN+1)^2$ by an induction argument on $(|\gamma|,|\delta|,l)$. 
Applying this, we estimate \eqref{eq:aux-eq-65}: 
\begin{align}
	&\Vert \mathrm{ad}_{h}^{(m-k)}(a^{\#}(f_x)) h^k \psi\Vert \nn \\
	& \leq  \sum_{\substack{|\beta|+|\gamma|+|\delta|+ l\leq 2(m-k) \\  |\gamma|+|\delta| +l \leq 2(m-k)-1}} C_{m,k} \big(\sup_x(1+x^2)^{-1/4}\Vert D_x^\beta f_x \Vert_{L^2}\big) \nn \\
	&\quad \times \Vert (1+x^2)^{1/4} (\cN+1)^{1/2+l} x^\gamma D_x^\delta h^k \psi\Vert \nn \\
	& \overset{\clap{\eqref{eq:aux-eq-67}}}{\leq}  \sum_{\substack{|\beta|+|\gamma|+|\delta|+ l\leq 2(m-k) \\  |\gamma|+|\delta|+l\leq 2(m-k)-1}} C_{m,k} 
	\big(\sup_{x\in \BR^3}(1+x^2)^{-1/4}\Vert \partial^\beta_x f_{t,x}\Vert_2\big)
	  C \Vert h^{\frac{|\delta|+|\gamma|+l +1}{2}}   h^k \psi\Vert \nn \\
	&\leq C_{m,k} \Big(\sup_{x\in \BR^3} \sum_{|\beta|\leq M} (1+x^2)^{-1/4}\Vert \partial^\beta_x f_{t,x}\Vert_2\Big) \Vert h^m\psi\Vert \,. \label{eq:aux-eq-71}
\end{align}
We conclude from \eqref{eq:aux-eq-65} and \eqref{eq:aux-eq-71}
\begin{align*}
	&\big\vert \pscal{(a^{\#}(f_x))^* \psi ,  h^{2m} \psi   } - \pscal{ h^{2m}\psi ,  a^{\#}(f_x) \psi   } \big\vert \\
	&\leq C_m  \Big(\sup_{x\in \BR^3} \sum_{|\beta|\leq M}(1+x^2)^{-1/4}\Vert \partial^\beta_x f_{t,x}\Vert_2\Big) \Vert h^m\psi\Vert^2 \,.
\end{align*}
 This proves the bound \eqref{eq:aux-eq-26} in the case of even $M=2m$ for all $f\in C_b^\infty$.

Now we consider the case $M=2m+1$ odd. 
To avoid commuting $a^{\#}(f_x)$ with $h^{1/2}$, whose commutator is not directly accessible, we rewrite $h$ in the following form:
\begin{align}
h= (-\nabla_x +x) (\nabla_x+x) + (\cN+1)(\cN+1)+3 \,. \label{eq:aux-eq-68}
\end{align}
 To keep the argument simple, we restrict ourselves to the case of $m=0$.
We use the decomposition \eqref{eq:aux-eq-68} to get similarly to \eqref{eq:aux-eq-65} that
\begin{align}
	& \Big\vert \pscal{(a^{\#}(f_x))^* \psi ,  h \psi   } - \pscal{ h\psi ,  a^{\#}(f_x) \psi} \Big\vert \nn\\
	&= \Big\vert \sum_{B\in\{ (\nabla_x +x) ,(\cN+1)\}} 
	  \Big(	
	  \mathrm{Im}\pscal{ \psi ,  \Big[a^{\#}(f_x),B^*\big] B  \psi } \nn \\
	  & \qquad\qquad\qquad\qquad\qquad - \mathrm{Im} \pscal{B \psi ,  \big[a^{\#}(f_x),B^*\big] \psi   } 
	\Big)\Big\vert \,. \label{eq:aux-eq-75}
\end{align}
The commutators in \eqref{eq:aux-eq-75} can be estimated  directly, proving \eqref{eq:aux-eq-26} for $M=1$. 
The result then extends to all odd $M=2m+1$ by analogy with the even case.

 With the use of the following identity,
\begin{align}
	\mathrm{ad}_h^{(k)} (-\Delta_x)\psi &= \sum_{|\gamma|+|\delta|\leq2} C_{\gamma,\delta}^{(k)} x^\gamma D_x^{\delta} \psi \,, \quad \forall k\in\BN_0 \,,
\end{align} 
 instead of the representation of the iterated commutator in \eqref{eq:aux-eq-66} the proof of the estimate \eqref{eq:aux-eq-18} follows in a similar fashion to the proof of the \eqref{eq:aux-eq-26} estimate shown above. 
	\end{proof}


\subsection{Proofs of the Infinite-Volume Approximation}
The following lemma establishes the convergence, as $\Lambda\to \infty$,  of the terms that appear in $\widetilde{H}^{\mathrm{BF}}_{\Lambda,\cZ_0}(t)$ and in the generator of $\cV_t^\Lambda$.

\begin{lemma}\label{lem:Approx-V-by-V-infty-Conv-Rates-2} 
	 For all volumes $\Lambda\geq1$, let $\varphi_t^\Lambda$ be the solution of the Hartree equation \eqref{eq:Hartree}. Assume that its initial data $\varphi_0^\Lambda$  varies on the scale $\Lambda^{1/3}$, that is, 
\cref{con:Initial condition}, \cref{con:Initial condition derivative condensate}i)$_{2}$. Furthermore, we require that $\varphi_0^\Lambda$ is flat around the origin, namely \cref{con:Condensate-flat-around-origin}$_{2,s}$, $0<s<1/3$, and assume that $\exists \eta\in H^1(\BR^3)$, $\delta,C>0$ such that $\forall \Lambda\geq1$
\begin{align}
	\Vert \varphi_0^\Lambda(\Lambda^{1/3}\,.\,) -\eta \Vert_2 \leq  C\Lambda^{-\delta} \,. \label{eq:ApproxV-by-V-infty-0}
\end{align}
We set
\begin{align}
	\gamma=\min\!\left\{\delta,\, s,\, \tfrac{3}{2}\bigl(\tfrac13-s\bigr),\, \tfrac16 \right\} \,. \label{eq:def-gamma}
\end{align}
Then we have for all $T\geq0$ that there exists $C>0$ such that for all $\Lambda\geq1$ and $f\in L^2(\BR^3)$ with $y^2 f\in L^2$ 
 \begin{align}
 	\Vert (Q_t^\Lambda -1) f \Vert_2 &\leq C\Lambda^{-3/2(1/3-s)} \Vert (1+y^2) f\Vert_2 \,,  \label{eq:ApproxV-by-V-infty-2}\\
 	|\mu_t^\Lambda -\mu^{\infty} | &\leq C ( \Lambda^{-\delta} + \Lambda^{-1/6}) \,, \label{eq:ApproxV-by-V-infty-3} \\
 	\left\Vert \left( \varphi_t^\Lambda - \rme^{-\rmi t \left(\int V -\mu^\infty\right)} \right) f \right\Vert_2 &\leq  C \Lambda^{-\gamma}  \Vert (1+y^2) f\Vert_2 \nn \\
 	&\quad  + C\Lambda^{-1/6} \Vert f\Vert_\infty \,, \quad \text{if } f\in L^\infty   \label{eq:ApproxV-by-V-infty-4}
  \end{align}
  and 
  \begin{align}
  	\left\Vert \left(K_1^\Lambda(t) -K_1^\infty \right) f \right\Vert_2 &\leq C  \Lambda^{-\gamma}    \Vert (1+y^2) f\Vert_2  \label{eq:ApproxV-by-V-infty-6} \\
  	\left\Vert \left(K_2^\Lambda(t)  -K_2^\infty \rme^{-2\rmi t\left(\int V -\mu^\infty\right) }\right) J f \right\Vert_2 &\leq C  \Lambda^{-\gamma}    \Vert (1+y^2) f\Vert_2 \,. \label{eq:ApproxV-by-V-infty-6.1}
  \end{align}
  If in addition \cref{con:Initial condition derivative condensate}$_{4}$ is satisfied then
  \begin{align}
  	&\Big\Vert \Big( \left( \rmi \partial_t \varphi_t^\Lambda\right) - \Big(\int V -\mu^\infty\Big) \rme^{-\rmi t\left(\int V -\mu^\infty\right)} \Big) f \Big\Vert_2 \leq  C \Lambda^{-1/6} \Vert f\Vert_\infty \nn \\
  	&\qquad\qquad\qquad\qquad\qquad\qquad\quad +  C \Lambda^{-\gamma}    \Vert (1+y^2) f\Vert_2 \,, \quad \text{if } f\in L^\infty, \label{eq:ApproxV-by-V-infty-7} 
  	\end{align}
  	and 
  	\begin{gather}
  	\left\Vert \rmi \partial_t Q_t^\Lambda f\right\Vert_2 \leq C \left( \Lambda^{-1/6} + \Lambda^{-1/2(1-3s)} \right) \Vert (1+y^2) f\Vert_2 \,. \label{eq:ApproxV-by-V-infty-8}
  \end{gather}

\end{lemma}
\begin{proof}
To prove the lemma, we localize the condensate around the origin, where it is approximately flat: $|\varphi_t^\Lambda|\sim 1$ (see \cref{cor:phi-propagation}  and \cref{con:Condensate-flat-around-origin}). For this purpose, we use the localization function with $n=1$ and $0<s<1/3$: $\tl(y)= 1 /(1+ (\Lambda^{-s} y)^2)$ as defined in \cref{sec:LocalizationFunction}.
 \\Let $T\geq 0$ and $-T\leq t\leq T$. 
  We begin with the estimate of \eqref{eq:ApproxV-by-V-infty-2} 
 \begin{align*}
 	&\Vert (Q_t^\Lambda-1) f \Vert_2 = \frac{1}{\Lambda} \Vert \varphi_t^\Lambda\Vert_2 |\pscal{\tl \varphi_t^\Lambda, \tl^{-1} f}| \\
 	& \leq \frac{\Vert \varphi_t^\Lambda\Vert_\infty}{\Lambda^{1/2}}  \Vert\tl\Vert_2 \Vert (1+y^2) f\Vert_2 \leq C \Lambda^{-3/2(1/3-s)} \Vert (1+y^2) f\Vert_2  \,.
 \end{align*}

Next, we prove $|\mu_t^\Lambda-\mu^\infty |\leq C(\Lambda^{-\delta}+ \Lambda^{-1/6})$. We have
 \begin{align*}
 	\mu_0^\Lambda&= \frac{1}{2\Lambda} \int |\varphi_0^\Lambda|^2 (x) V(y) |\varphi_0^\Lambda|^2(x-y) dxdy \\
 		&=\frac{1}{2} \pscal{|\varphi_0^\Lambda|^2 (\Lambda^{1/3}\,.\,), \Lambda V(\Lambda^{1/3}\,.\,)\ast |\varphi_0^\Lambda|^2(\Lambda^{1/3}\,.\,)}\,,
 \end{align*}
 after the rescaling $x\to \Lambda^{1/3}x$ and $y\to \Lambda^{1/3}y$.
 Using assumption \eqref{eq:ApproxV-by-V-infty-0}, we get
  \begin{align}
 	\Big\vert \mu_0^\Lambda-  \frac{1}{2}\pscal{ |\eta|^2, \Lambda V(\Lambda^{1/3}\,.\,)\ast|\eta|^2} \Big\vert \leq \Vert V\Vert_{2,\infty} \Lambda^{-\delta}\,.
 \end{align}
  The convolution with $\Lambda V(\Lambda^{1/3}\,.\,)$ acts as an approximation of the identity, yielding 
 \begin{align}
 	\Big\vert \pscal{ |\eta|^2, \Lambda V(\Lambda^{1/3}\,.\,)\ast|\eta|^2} -\frac{1}{2} \Big\langle |\eta|^2 ,  |\eta|^2\int V \Big\rangle \Big\vert \leq C\Lambda^{-1/3} \Vert |y|V\Vert_1 \Vert \eta\Vert_{H^1}^4\,.
 \end{align}
 By \cref{cor:phi-propagation}, we also have $|\mu_t^\Lambda-\mu_0^\Lambda|\leq C \Vert V\Vert_{1,2}\Lambda^{-1/6}$, which shows \eqref{eq:ApproxV-by-V-infty-3}.

 We now prove \eqref{eq:ApproxV-by-V-infty-4}. For this we first estimate $\Vert V\ast(|\varphi_t^\Lambda|^2-1) f\Vert_2\leq \Vert \tl V\ast(|\varphi_t^\Lambda|^2-1)\Vert_\infty \Vert 1/\tl f\Vert_2$, and following the proof of \cref{cor:ThetaDPhi-TPhi}, we obtain
 \begin{align}
 	\Vert V\ast(|\varphi_t^\Lambda|^2-1) f\Vert_2 &\leq C \big(\Lambda^{-s} + \Lambda^{-2(1/3-s)}+ \Lambda^{-1/6}\big)  \Vert (1+y^2) f\Vert_2 \,. \label{eq:ApproxV-by-V-infty-1} 
 \end{align}
 We proceed by showing the claim for $\tphi_t^\Lambda$ and subsequently extend it to $\varphi_t^\Lambda$. Using \cref{Lemma:CalculationTheta}a), we have
 \begin{align}
 	&\left\Vert \left( \tphi_t^\Lambda - \rme^{-\rmi \left( t V\ast |\varphi_0^\Lambda|^2  - \int_0^t\mu_s^\Lambda ds \right)}\right) f\right\Vert_2 = \Vert (\varphi_0^\Lambda -1 ) f\Vert_2  \nn \\ 
 	&\leq \Vert \tl (\varphi_0^\Lambda -1 )\Vert_\infty \Vert 1/\tl f\Vert_2
 	 \leq C \Lambda^{-2(1/3-s)} \Vert (1+y^2) f\Vert_2 \,. \label{eq:aux-eq-82}
 \end{align}
 Next, using $|\rme^{\rmi x}-1|\leq |x|$, \eqref{eq:ApproxV-by-V-infty-3} and \eqref{eq:ApproxV-by-V-infty-1}, we obtain
 \begin{align}
 	\left\Vert \left(   \rme^{-\rmi \left( t V\ast |\varphi_0^\Lambda|^2  - \int_0^t\mu_s^\Lambda ds \right)} - \rme^{-\rmi t\left( \int V -\mu^\infty\right)} \right) f\right\Vert_2  
	\leq C  \Lambda^{-\gamma} \Vert (1+y^2) f\Vert_2  \,. \label{eq:aux-eq-83}
 \end{align}

 To estimate $\varphi_t^\Lambda$ by the simplified phase, we insert $\pm\tphi_t^\Lambda$ and  apply \eqref{eq:aux-eq-82}, \eqref{eq:aux-eq-83}, and \cref{prop:phi-propagation}, yielding
 \begin{align*}
 	&\left\Vert \left( \varphi_t^\Lambda - \rme^{-\rmi t\left( \int V -\mu^\infty\right)} \right) f \right\Vert_{2} \leq \Vert \varphi_t^\Lambda - \tphi_t^\Lambda \Vert_2  \Vert f\Vert_\infty + \Vert (\varphi_0^\Lambda-1) f\Vert_2  \\
 	&\quad + \left\Vert \left(   \rme^{-\rmi \left( t V\ast |\varphi_0^\Lambda|^2  - \int_0^t\mu_s^\Lambda ds \right)} - \rme^{-\rmi t\left( \int V -\mu^\infty\right)} \right) f\right\Vert_2 \\
 	&\leq C \Lambda^{-1/6} \Vert f\Vert_\infty + C \big(\Lambda^{-2(1/3-s)} + \Lambda^{-\gamma}\big) \Vert (1+y^2) f\Vert_2  \,,
 \end{align*}
 which proves \eqref{eq:ApproxV-by-V-infty-4}. In the same way, estimating $\Vert (\varphi_t^\Lambda - \tphi_t^\Lambda) f \Vert_{1\wedge 2} \leq \Vert \varphi_t^\Lambda - \tphi_t^\Lambda \Vert_2  \Vert f\Vert_2 $, yields
 \begin{align}
 	\left\Vert \left( \varphi_t^\Lambda - \rme^{-\rmi t \left(\int V -\mu^\infty\right)} \right) f \right\Vert_{1\wedge2} &\leq  C\Lambda^{-\gamma} \Vert (1+y^2) f\Vert_2 \,. \label{eq:ApproxV-by-V-infty-5} 
 \end{align}
 This estimate is used to prove \eqref{eq:ApproxV-by-V-infty-6} and \eqref{eq:ApproxV-by-V-infty-6.1}.
 
Now, we prove \eqref{eq:ApproxV-by-V-infty-6}. Using \eqref{eq:ApproxV-by-V-infty-2}, $K_1^\Lambda(t)=Q_t^\Lambda \tilde{K}_1^\Lambda(t) Q_t^\Lambda$, as well as $\Vert \tilde{K}_1^\Lambda(t)\Vert_{\rm op}\leq C$, we have
 \begin{align}
	&\Vert (K_1^\Lambda(t)- K_1^\infty ) f\Vert_2   \nn \\
	&\leq \Vert Q_t^\Lambda \tilde{K}_1^\Lambda \left(Q_t^\Lambda -1 \right) f \Vert_2 + \Vert Q_t^\Lambda \big(\tilde{K}_1^\Lambda -K_1^\infty \big) f \Vert_2 + \Vert \left( Q_t^\Lambda - 1\right) K_1^\infty f \Vert_2 \nn \\
&\leq  C\Lambda^{-1/2(1-3s)} \Vert (1+y^2) f\Vert_2 + \Vert (\tilde{K}_1^\Lambda -K_1^\infty) f\Vert_2 \,.
 \end{align}
 The last term is estimated using \eqref{eq:ApproxV-by-V-infty-4} and \eqref{eq:ApproxV-by-V-infty-5}
\begin{align}
	&\Vert (\tilde{K}_1^\Lambda -K_1^\infty) f\Vert_2 = \Vert \varphi_t^\Lambda V\ast((\varphi_t^\Lambda)^*f)- V\ast  f\Vert_2 \nn \\
	&\leq \left\Vert \varphi_t^\Lambda V\ast \left(\rme^{\rmi t\left( \int V -\mu^\infty\right)}   f\right) - V\ast f \right\Vert_2 + \left\Vert \varphi_t^\Lambda V \ast \left( (\varphi_t^\Lambda)^* f - \rme^{\rmi t\left( \int V -\mu^\infty\right)}  f \right) \right\Vert_2 \nn \\
	&\leq C\Lambda^{-\gamma} \big(\Vert (1+y^2) V\ast f\Vert_2  + \Vert (1+y^2) f \Vert_2\big) + C\Lambda^{-1/6} \Vert V\ast f\Vert_\infty  \nn \\
	&\leq C \Lambda^{-\gamma}  \Vert (1+y^2)  f\Vert_2 \,, \label{eq:aux-eq-113}
\end{align}
where we used $\Vert (1+y^2) V\ast f\Vert_2= C\Vert (1-\Delta) \widehat{V} \widehat{f}\Vert_2\leq C\Vert (1+y^2)V\Vert_2 \Vert (1+y^2)f\Vert_2 $.
We conclude $\Vert (K_1^\Lambda(t)- K_1^\infty ) f\Vert_2\leq \eqref{eq:aux-eq-113}$. Analogously, one shows \eqref{eq:ApproxV-by-V-infty-6.1}.

From here on, we additionally assume \cref{con:Initial condition derivative condensate}$_{4}$. We start by proving \eqref{eq:ApproxV-by-V-infty-7}. This estimate follows from the Hartree equation, $\rmi\partial_t\varphi_t^\Lambda =  h_t\varphi_t^\Lambda$, where $h_t=-\frac{\Delta}{2} + V\ast |\varphi_t^\Lambda|^2 -\mu_t^\Lambda$, together with the preceding estimates and \cref{cor:derivative-varphi-L2-norm-estimate}$_{|\beta|=2}$. 

Finally, the bound of  \eqref{eq:ApproxV-by-V-infty-8} is obtained analogously, using $\rmi\partial_tQ_t^\Lambda f=  \Lambda^{-1} h_t\varphi_t^\Lambda \pscal{\varphi_t^\Lambda, f} -  \Lambda^{-1}\varphi_t^\Lambda \pscal{- h_t\varphi_t^\Lambda, f}$ and applying the same type of estimates as in \eqref{eq:ApproxV-by-V-infty-7}.
\end{proof}

By a density argument, together with the uniform boundedness of the operators appearing in \cref{lem:Approx-V-by-V-infty-Conv-Rates-2}, the convergence extends immediately to the whole $L^2$-space.

\begin{cor} \label{cor:Strong-convergence-infinite-volume}
Under the conditions of \cref{lem:Approx-V-by-V-infty-Conv-Rates-2} we have for $f\in L^2$ that, as $\Lambda\to \infty$,
 \begin{align}
 	\Vert (Q_t^\Lambda-1)f\Vert_2+ |\mu_t^\Lambda -\mu^\infty |   + \left\Vert \left( \varphi_t^\Lambda - \rme^{-\rmi t \left(\int V -\mu^\infty\right)} \right) f \right\Vert_2 &\to 0 \,, \label{eq:condensate-infinite-volume-1}  \\
 	\Vert (K_1^\Lambda(t) -K_1^\infty) f \Vert_2 + \Vert (K_2^\Lambda(t)  -K_2^\infty \rme^{-2\rmi t\left(\int V -\mu^\infty\right) }) J f \Vert_2&\to 0   \,. 
 \end{align} 
\end{cor}
Consequently, in the infinite-volume limit, all functions in $L^2(\BR^3)\setminus\{0\}$  are not in the condensate and thus represent excitations. This is because the condensate becomes a constant phase (see \eqref{eq:condensate-infinite-volume-1}) and therefore no longer belongs to $L^2(\BR^3)$. This can also be seen from the strong convergence of the projection $Q_t^\Lambda$ to the identity. 

To prove the convergence of $\cV_t^\Lambda$ to $\cV_t^\infty$ we need the following lemma.

\begin{lemma}\label{lem:V-inf-Groenwall-Harmonic-Osc}
Let $\cV^\infty_t$ be from \eqref{eq:def-A-infty}. 
Set $h:= -\Delta + y^2 +1$ then $\cV_t^\infty \big(D(h)\oplus JD(h)\big) = D(h)\oplus JD(h)$, and for all $T\geq0$ there exists a constant $C>0$ such that for all $F\in D(h)\oplus JD(h)$ 
\begin{align}
	\sup_{t\in [-T,T]} \big\Vert \big(h\oplus Jh J^* \big) \cV_t^\infty F \big\Vert_{L^2\oplus JL^2} \leq C \big\Vert \big( h\oplus Jh J^*\big) F\big \Vert_{L^2\oplus JL^2}\,. \label{eq:Harmonic-oc-Bog-infty-bound}
	\end{align}
\end{lemma}
\begin{proof}[Proof of \cref{lem:V-inf-Groenwall-Harmonic-Osc}]
	Let $\psi\in \cS(\BR^3)$. Since $K_1^\infty\psi =V\ast \psi$ and $K_2^\infty J \psi = V\ast \psi^* \rme^{-2\rmi t \left( \int V -\mu^\infty\right)}$ one can verify via Fourier transform that
	\begin{align}
	\Vert [y^2,K_1^\infty]\psi \Vert_2 &\leq C \Vert (y^2 +1)V\Vert_1 \Vert (y^2+1)^{1/2} \psi \Vert_2 \,, \label{eq:aux-eq-88} \\
	\Vert [y^2,K_2^\infty J]\psi \Vert_2 &\leq C \Vert (y^2 +1)V\Vert_1 \Vert (y^2+1)^{1/2} \psi \Vert_2 \,, \label{eq:aux-eq-89} \\
		\Vert [y^2,-\Delta]\psi \Vert_2 &\leq C \Vert (1 + y \nabla) \psi\Vert_2 \overset{\eqref{eq:aux-eq-67}}{\leq} C \Vert h \psi \Vert_2 \,, \label{eq:aux-eq-90} \\
		\Vert h K_2^\infty \psi \Vert_2 &\leq C \Vert (1+y^2)V\Vert_1 \Vert h\psi\Vert_2 \,. \label{eq:aux-eq-91}
	\end{align}
Since $\cA^\infty$ is translation-invariant it commutes with the Laplacian. Therefore, we have
	\begin{align*}
		\left\Vert \left[ h\oplus Jh J^*, \cA^\infty \right] \cV_t^\infty F \right\Vert &= \left\Vert \left[ y^2\oplus Jy^2J^* , \cA^\infty \right] \cV_t^\infty F \right\Vert \\
		&\leq C  \left\Vert ( h\oplus Jh J^*) \cV_t^\infty F \right\Vert \,.
	\end{align*}
We use $(\cA^\infty)^*=S\cA^\infty S$ to conclude
	\begin{align}
		&\pm\partial_t \Vert (h\oplus Jh J^*) \cV_t^\infty F \Vert^2 			=\mp\mathrm{Im}\Big\{ \pscal{ (h\oplus Jh J^*)^2 \cV_t^\infty F, \cA^\infty \cV_t^\infty F }\nn \\
		&\qquad\qquad\qquad\qquad\qquad\qquad\quad\quad\ - \pscal{\cA^\infty \cV_t^\infty F , (h\oplus Jh J^*)^2 \cV_t^\infty F } \Big\} \nn\\
		&= \mp\mathrm{Im}\Big\{ \pscal{  \cV_t^\infty F, \big[(h\oplus Jh J^*)^2,\cA^\infty\big] \cV_t^\infty F } \label{eq:aux-eq-86} \\
		&\qquad\quad\quad + \pscal{ \cV_t^\infty F , \big(\cA^\infty- S\cA^\infty S\big)(h\oplus Jh J^*)^2 \cV_t^\infty F } \Big\} \,. \label{eq:aux-eq-87}
	\end{align}
	First, we estimate \eqref{eq:aux-eq-86} using $(\cA^\infty)^*=S \cA^\infty S$:
	\begin{align*}
		\eqref{eq:aux-eq-86} &\leq \Vert  h\oplus Jh J^* \cdot \cV_t^\infty F \Vert \left\Vert \left[  h\oplus Jh J^* , \cA^\infty \right] \cV_t^\infty F \right\Vert \\
		&\quad +  \left\Vert S\left[  h\oplus Jh J^* , \cA^\infty \right]S \cV_t^\infty F \right\Vert \Vert (h\oplus Jh J^*) \cV_t^\infty F \Vert \\
		&\leq C  \left\Vert ( h\oplus Jh J^*) \cV_t^\infty F \right\Vert^2 \,.
	\end{align*}
 The second term \eqref{eq:aux-eq-87} is controlled in the same way using \eqref{eq:aux-eq-91}
	\begin{align*} 
		\eqref{eq:aux-eq-87}&\leq \left\Vert (h\oplus Jh J^*) \begin{pmatrix}
		 	0 & -2K_2^\infty \\
		 	2(K_2^\infty)^* & 0 
		\end{pmatrix}  \cV_t^\infty  F\right\Vert \Vert (h\oplus Jh J^*) \cV_t^\infty  F\Vert \nn \\
		&\leq \Vert (h\oplus Jh J^*) \cV_t^\infty  F\Vert^2 \,.
	\end{align*}
	The claim follows from Grönwall's Lemma and the estimates above. Moreover, since $(\cV_t^\infty)^{-1}=(\cV_{-t}^\infty)$ (see \eqref{eq:-V-infty-Explicit}), the bound \eqref{eq:Harmonic-oc-Bog-infty-bound} also holds for the inverse dynamics. Consequently, $\cV_t^\infty \big( D(h)\oplus JD(h)\big) = D(h)\oplus JD(h)$.	
\end{proof}

\cref{lem:Approx-V-by-V-infty-Conv-Rates-2} and \cref{lem:V-inf-Groenwall-Harmonic-Osc} allow us to conclude the convergence of $\rme^{\rmi t \left( \int V -\mu^\infty\right) S}\cV_t^\Lambda$ to $\cV_t^\infty$.

\begin{lemma} \label{lem:Approx-V-by-V-infty-Conv-Rates}
Assume the conditions of \cref{lem:Approx-V-by-V-infty-Conv-Rates-2}  and set $h=-\Delta+y^2+1$. Then for all $T\geq0$, there exists a constant $C>0$ such that for all $\Lambda\geq1$, $F\in D(h)\oplus JD(h)$ and $-T\leq t\leq T$
\begin{align}
	&\left\Vert \left(\rme^{\rmi t \left( \int V -\mu^\infty\right) S}\cV_t^\Lambda -  \cV_t^\infty\right) F\right\Vert + \left\Vert \left( (\cV_t^\Lambda)^{-1} \rme^{-\rmi t \left( \int V -\mu^\infty\right) S} - (\cV_t^\infty )^{-1} \right) F\right\Vert \nn \\
	 &\leq  C \Lambda^{-\gamma} \big\Vert \big(h\oplus J hJ^*\big) F \big\Vert\,. \label{eq:V-to-V-infty-Conv-Rate} 
\end{align}
\end{lemma}
\begin{proof}[Proof of \cref{lem:Approx-V-by-V-infty-Conv-Rates}]
Recall that the Bogoliubov map $\cV_t^\Lambda$ satisfies the evolution equation $\rmi \partial_t \cV_t^\Lambda = \cA^\Lambda(t) \cV_t^\Lambda$,  $\cV_0^\Lambda= I$ (see \eqref{eq:V-Matrix-Representation}).
We remark that
		\begin{align}
			\Vert \cV_t^\Lambda \Vert_{\cL(L^2,L^2)}\leq C \,. \label{eq:Bog-Map-bounded}
		\end{align}
	 To prove this, write
	\begin{align}
		\cV_t^\Lambda=:\begin{pmatrix}
	U_t^\Lambda & J^*V_t^\Lambda J^* \\
	V_t^\Lambda & JU_t^\Lambda J^* 
\end{pmatrix} \,.
	\end{align}
Then (similar to \cite[Lemma~4.9]{BPPS22}) it follows from a Grönwall argument with the use of Duhamel and the time evolution of $\cV_t$ in \eqref{eq:V-Matrix-Representation} that
\begin{align*}
	\Vert U_t^\Lambda\Vert_{\text{op}}+ \Vert V_t^\Lambda\Vert_{\text{op}} \leq 1 +\int_0^t \Vert K_2^\Lambda(\tau)\Vert_{\rm op} \big( \Vert U_\tau^\Lambda\Vert_{\rm op}+ \Vert V_\tau^\Lambda\Vert_{\rm op} \big) d\tau \,.
\end{align*}
Since $\Vert K_2^\Lambda(\tau)\Vert_{\rm op}\leq C$ it follows immediately that $\Vert \cV_t^\Lambda \Vert_{\rm op}\leq C$. Indeed, the bound on $K_2$ follows directly from its definition \eqref{eq:K2-Def-2}:
\begin{align}
		 \Vert K_2^\Lambda(t) J \psi \Vert_2 &=\Vert Q_t^\Lambda\tilde{K}_2^\Lambda(t) J Q_t^\Lambda\psi \Vert_2 \leq \Vert \varphi_t^\Lambda V\ast (\varphi_t (Q_t^\Lambda\psi)^*)\Vert_2 \nn \\
		 &\leq \Vert \varphi_t^\Lambda\Vert_\infty \Vert V\Vert_2 \Vert \varphi_t^\Lambda\Vert_\infty \Vert Q_t^\Lambda\psi\Vert_2 \leq C \Vert\psi\Vert_2\,. \label{eq:K2-bound}
	\end{align}
Next, let $F\in D(h)\oplus JD(h)$ and set $\nu=\left(\int V -\mu^\infty\right)$. 
Then
	\begin{align}
		&\pm\partial_t\Vert (1 - (\cV_t^\Lambda)^{-1} \rme^{-\rmi t \nu S} \cV_t^\infty)F \Vert^2  \nn \\
		& =  \pm 2\mathrm{Im} \big\langle (1 - (\cV_t^\Lambda)^{-1} \rme^{-\rmi t \nu S} \cV_t^\infty)F  , (\cV_t^\Lambda)^{-1} (\rme^{\rmi t \nu S}\cA^\Lambda(t)\rme^{-\rmi t \nu S}  \nn \\
		&\qquad\qquad\qquad\qquad\qquad\qquad\qquad\qquad\qquad\quad -\nu S - \cA^\infty) \cV_t^\infty F  \big\rangle   \nn \\
		&\leq \Vert  (1 - (\cV_t^\Lambda)^{-1} \rme^{-\rmi t \nu S} \cV_t^\infty)F\Vert^2 \label{eq:aux-eq-104}  \\
		&\quad + C \Vert (\rme^{\rmi t \nu S}\cA^\Lambda(t)\rme^{-\rmi t \nu S} -\nu S - \cA^\infty)  \cV_t^\infty F\Vert^2 \,. \label{eq:aux-eq-103}
	\end{align}
	The second term \eqref{eq:aux-eq-103} can be estimated using \cref{lem:Approx-V-by-V-infty-Conv-Rates-2} together with the definitions \eqref{eq:A(t)-Def} and \eqref{eq:def-A-infty} of $\cA^\Lambda(s)$ and $\cA^\infty$:
	\begin{align}
		\eqref{eq:aux-eq-103}  \leq  C  \Lambda^{-2\gamma}   \Vert \big((1+y^2)\oplus J(1+y^2)J\big)\cV_t^\infty F\Vert^2 \,. \label{eq:aux-eq-143}
	\end{align}
	
	Since $(1+y^2)\leq h$ and the uniform boundedness of $\cV_t^\infty $ with respect to the graph norm of $h\oplus JhJ^* $ (see \cref{lem:V-inf-Groenwall-Harmonic-Osc}), it follows that
	\begin{align}
		\eqref{eq:aux-eq-143} \leq C \Lambda^{-2\gamma}  \Vert \big(h\oplus JhJ^* \big)F\Vert^2 \,.\label{eq:aux-eq-106}
\end{align}		
 Thus by Grönwall's Lemma using the estimates \eqref{eq:aux-eq-104} and \eqref{eq:aux-eq-106}, we have for all $F\in D(h)\oplus JD(h)$
 \begin{align}
 	\Vert (1 - (\cV_t^\Lambda)^{-1} \rme^{-\rmi t \nu S} \cV_t^\infty)F \Vert\leq  C \Lambda^{-\gamma}  \Vert \big(h\oplus JhJ^* \big)F\Vert \,. \label{eq:aux-eq-107}
 \end{align}
 The claim then follows from the uniform boundedness of $\cV_t^\infty $ with respect to $h\oplus JhJ^* $, and the fact that $(\cV_t^\infty)^{-1}=\cV_{-t}^\infty$.
\end{proof}

From \cref{lem:Approx-V-by-V-infty-Conv-Rates} we conclude strong convergence of $\cV_t^\Lambda$ on the whole $L^2(\BR^3)\oplus JL^2(\BR^3)$.

\begin{cor}\label{cor:Approx-V-by-V-infty} 
Assume the conditions of \cref{lem:Approx-V-by-V-infty-Conv-Rates-2}.
  Then we have in the limit $\Lambda\to\infty$ that
\begin{align}
	\rme^{\rmi t \left( \int V -\mu^\infty\right) S}\cV_t^\Lambda &\to \cV_t^\infty  \,, \label{eq:V-to-V-infty} \\
	(\cV_t^\Lambda)^{-1} \rme^{-\rmi t \left( \int V -\mu^\infty\right) S} &\to (\cV_t^\infty )^{-1} \label{eq:V-to-V-infty-2}
\end{align}
strongly as operators on $L^2(\BR^3)\oplus JL^2(\BR^3)$.

\end{cor}

The following lemma describes the convergence of $\cV_t^\Lambda$ when applied to the state appearing in the creation and annihilation operator of the finite-volume Bogoliubov–Fröhlich Hamiltonian (see \eqref{eq:HBF-Ham}).

\begin{lemma} \label{lem:Approx-V-by-V-infty-on-specific-state} 
Assume the conditions of \cref{lem:Approx-V-by-V-infty-Conv-Rates-2} and let $\gamma$ be given by \eqref{eq:def-gamma}.  
\begin{itemize}
	\item[i)]  Let $\beta\in \BN_0^3$ and assume \cref{Assumption:Initial-datum-and-potential}$_{|\beta|}$. Then we have that there exists a constant $C>0$ such that for all volumes $\Lambda\geq1$
	\begin{align}
		&\sup_{x\in \BR^3}(1+x^2)^{-1/4} \Big\Vert \Big(  (\cV_t^\Lambda)^{-1} \left( Q_t^\Lambda (\partial^\beta W_x) \varphi_t^\Lambda \oplus J Q_t^\Lambda (\partial^\beta W_x) \varphi_t^\Lambda\right) \nn \\
		&\qquad\qquad\qquad\qquad - (\cV_t^\infty)^{-1} \left(\partial^\beta W_x \oplus J \partial^\beta W_x\right)  \Big) \Big\Vert_{L^2\oplus JL^2}  \nn \\
		&\leq C\Lambda^{-\gamma/4} \Big(\Vert (1+y^2 -\Delta)^{1/4} \partial^\beta W\Vert_2  + \Vert \partial^\beta W\Vert_\infty\Big) \,.  \label{eq:V-to-V-infty-convergence-rate-on-x-dependent-state} 
		\end{align}
\end{itemize}
 \item[ii)] Assume \cref{Assumption:Initial-datum-and-potential}$_{2}$ and \cref{con:Initial condition derivative condensate}$_{4}$. Then we have that there exists a constant $C>0$ such that for all volumes $\Lambda\geq1$
	\begin{align}
		&\sup_{x\in \BR^3}  (1+x^2)^{-1/4}\Big\Vert \ \partial_t\Big(  (\cV_t^\Lambda)^{-1} \left( Q_t^\Lambda  W_x \varphi_t^\Lambda \oplus J Q_t^\Lambda  W_x \varphi_t^\Lambda\right) \nn \\
		&\qquad\qquad\qquad\qquad  - (\cV_t^\infty)^{-1} \left( W_x \oplus J  W_x\right)  \Big) \Big\Vert_{L^2\oplus JL^2}  \nn \\
		&\leq C\Lambda^{-\gamma/4} \sum_{|\beta|\leq2}\Big(\Vert (1+y^2 -\Delta)^{1/4} \partial^\beta W\Vert_2  + \Vert \partial^\beta W\Vert_\infty\Big) \,. \label{eq:V-to-V-infty-convergence-rate-on-x-dependent-state-time-deriv}
	\end{align}
\end{lemma}
\begin{proof}
Let $x\in \BR^3$ be fixed. We start with the proof of \eqref{eq:V-to-V-infty-convergence-rate-on-x-dependent-state}. We have
\begin{align}
	&  (1+x^2)^{-1/4}\Big\Vert  (\cV_t^\Lambda)^{-1}  \rme^{-\rmi\nu S}\rme^{\rmi\nu S}\left( Q_t^\Lambda (\partial^\beta W_x) \varphi_t^\Lambda \oplus J Q_t^\Lambda (\partial^\beta W_x) \varphi_t^\Lambda\right) \nn \\
		&\qquad\qquad\qquad - (\cV_t^\infty)^{-1} \left(\partial^\beta W_x \oplus J \partial^\beta W_x\right)   \Big\Vert \nn \\
		&\leq  (1+x^2)^{-1/4} \Big\Vert  \left((\cV_t^\Lambda)^{-1}\rme^{-\rmi\nu S} - (\cV_t^\infty)^{-1}\right)  \left(\partial^\beta W_x \oplus J \partial^\beta W_x\right) \Big\Vert  \label{eq:aux-eq-114} \\
		&+ C  (1+x^2)^{-1/4}\Big\Vert \rme^{\rmi\nu S} \left( Q_t^\Lambda (\partial^\beta W_x) \varphi_t^\Lambda \oplus J Q_t^\Lambda (\partial^\beta W_x) \varphi_t^\Lambda\right)  \nn \\\
		&\qquad\qquad\qquad\qquad  -  \partial^\beta W_x \oplus J \partial^\beta W_x   \Big\Vert \,. \label{eq:aux-eq-115}
\end{align}
The first term is estimated with \cref{lem:Approx-V-by-V-infty-Conv-Rates} and the Heinz-Kato interpolation theorem \cite[Chapter~IX, Proposition~9]{ReSi75ii}
\begin{align}
	\eqref{eq:aux-eq-114}\leq C   \Lambda^{-\gamma/4}  (1+x^2)^{-1/4} \big\Vert \big(h\oplus J hJ^*\big)^{1/4}  \left(\partial^\beta W_x \oplus J \partial^\beta W_x\right)  \big\Vert \,, \label{eq:aux-eq-116}
\end{align}
where  $h=1+y^2-\Delta$. Similarly, for the second term, in analogy with \cref{lem:Approx-V-by-V-infty-Conv-Rates-2}, using the Heinz-Kato theorem, we get 
\begin{align}
	\eqref{eq:aux-eq-115}\leq C \Lambda^{-\gamma/4} \left( (1+x^2)^{-1/4}\big\Vert (1+y^2)^{1/4} \partial^\beta W_x\big\Vert_2 + \Vert  \partial^\beta W\Vert_\infty \right) \,. \label{eq:aux-eq-117}
\end{align}
Using $W_x(y)=W(x-y)$ and a simple substitution argument, we get $(1+x^2)^{-1}\Vert h  \partial^\beta W_x\Vert_2\leq C \Vert h \partial^\beta W\Vert_2$. Then by $(1+x^2)^{-1}\Vert h  \partial^\beta W_x\Vert_2= (1+x^2)^{-1}\Vert T_x^* h  T_x \partial^\beta W\Vert_2$, where $T_x$ denotes the translation operator by $x$, and the Heinz-Kato interpolation theorem
	\begin{align}
		(1+x^2)^{-1/4}\Vert h^{1/4}  \partial^\beta W_x\Vert_2 \leq C \Vert h^{1/4} \partial^\beta W\Vert_2 \,, \label{eq:x-independent-estimate-W}
	\end{align}
	which gives an $x$-independent estimate.
	Combining \eqref{eq:aux-eq-116}, \eqref{eq:aux-eq-117} and \eqref{eq:x-independent-estimate-W}  proves \eqref{eq:V-to-V-infty-convergence-rate-on-x-dependent-state}. Note that the constants $C$ appearing above are all $x$-independent.
	
	Next, we prove \eqref{eq:V-to-V-infty-convergence-rate-on-x-dependent-state-time-deriv}. Using the evolution equations $\rmi\partial_t(\cV_t^\infty)^{-1}F = (\cV_t^\infty)^{-1} (-\cA^\infty)F$ and $\rmi\partial_t(\cV_t^\Lambda)^{-1}F = (\cV_t^\Lambda)^{-1} (-\cA^\Lambda(t))F$, we have
	\begin{align}
		& (1+x^2)^{-1/4}\Big\Vert  \partial_t\Big(  (\cV_t^\Lambda)^{-1} \left( Q_t^\Lambda  W_x \varphi_t^\Lambda \oplus J Q_t^\Lambda  W_x \varphi_t^\Lambda\right) \nn \\
		&\qquad\qquad\qquad  - (\cV_t^\infty)^{-1} \left( W_x \oplus J  W_x\right)  \Big) \Big\Vert_{L^2\oplus JL^2} \nn \\
		&\leq (1+x^2)^{-1/4} \Big\Vert  (\cV_t^\Lambda)^{-1} \big(-\cA^\Lambda(t)\big)  \left( Q_t^\Lambda  W_x \varphi_t^\Lambda \oplus J Q_t^\Lambda  W_x \varphi_t^\Lambda\right) \nn \\
		&\qquad\qquad\qquad - (\cV_t^\infty)^{-1}\big(-\cA^\infty- \nu S\big) \left( W_x \oplus J  W_x\right)  \Big\Vert \label{eq:V-derivative-approx}  \\
		&\quad  + (1+x^2)^{-1/4} \Big\Vert  (\cV_t^\Lambda)^{-1}\rmi \partial_t \left( Q_t^\Lambda  W_x \varphi_t^\Lambda \oplus J Q_t^\Lambda  W_x \varphi_t^\Lambda\right)
		\nn \\
		& \qquad\qquad\qquad\qquad - (\cV_t^\infty)^{-1} \nu S \left( W_x \oplus J  W_x\right)   \Big\Vert \,. \label{eq:state-derivative-approx} 
	\end{align}
	The first term, \eqref{eq:V-derivative-approx}, is bounded using  $\Vert \cV_t^\Lambda\Vert_{\rm op}\leq C$ (see \eqref{eq:Bog-Map-bounded}):
	\begin{align}
		\eqref{eq:V-derivative-approx}\leq  (1+x^2)^{-1/4}\big\Vert  \big( (\cV_t^\Lambda)^{-1} \rme^{-\rmi t \nu S} - (\cV_t^\infty)^{-1} \big) \big(-\cA^\infty- \nu S\big)  \nn\\
		\times\left( W_x \oplus J  W_x\right) \big\Vert \label{eq:aux-eq-118}  \\
		+ C (1+x^2)^{-1/4}\Big\Vert \Big( \rme^{\rmi t \nu S}\cA^\Lambda(t)\rme^{-\rmi t \nu S} \rme^{\rmi t \nu S}   \left( Q_t^\Lambda  W_x \varphi_t^\Lambda \oplus J Q_t^\Lambda  W_x \varphi_t^\Lambda\right) \nn \\
		\ - \big(\cA^\infty+ \nu S\big) \left( W_x \oplus J  W_x\right)  \Big) \Big\Vert \,, \label{eq:aux-eq-119}
	\end{align}
	where \eqref{eq:aux-eq-118} is estimated with the help of \cref{lem:Approx-V-by-V-infty-Conv-Rates} and the Heinz-Kato theorem 
	\begin{align}
		&\eqref{eq:aux-eq-118} \leq C\Lambda^{-\gamma/4} (1+x^2)^{-1/4} \big\Vert (h\oplus JhJ^*)^{1/4}  \big(\cA^\infty + \nu S\big)   \left( W_x \oplus J  W_x\right) \big\Vert  \,.  \label{eq:aux-eq-120}
	\end{align}
	 We see from the definition of $\cA^\infty$ that $\cA^\infty (W_x\oplus JW_x) =  -\Delta W_x\oplus J\Delta W_x$, where the $K_1^\infty$ and $K_2^\infty$ contributions cancel, since  $W_x$ is real valued. Using this and an argument similar to \eqref{eq:x-independent-estimate-W}, we get
	\begin{align}
	 \eqref{eq:aux-eq-120} &\leq C\Lambda^{-\gamma/4} \big\Vert (h\oplus JhJ^*)^{1/4}  (-\Delta+\nu)  W_x\oplus J(\Delta -\nu)W_x \big\Vert \nn \\
	& \leq C\Lambda^{-\gamma/4} \big(\Vert h^{1/4} (-\Delta) W \Vert_2 + \Vert h^{1/4} W\Vert_2 \big) \,.  \label{eq:aux-eq-69}
	\end{align}
	We continue with the \eqref{eq:aux-eq-119} estimate
		\begin{align}
		\eqref{eq:aux-eq-119}&\leq  (1+x^2)^{-1/4} \Big\Vert  \big(\rme^{\rmi t \nu S}\cA^\Lambda(t)\rme^{-\rmi t \nu S} -\cA^\infty- \nu S \big)  \left( W_x \oplus J  W_x\right)  \Big\Vert  \label{eq:aux-eq-121}  \\
		&\quad + (1+x^2)^{-1/4}\Big\Vert  \rme^{\rmi t \nu S}\cA^\Lambda(t)\rme^{-\rmi t \nu S} \Big( Q_t^\Lambda  W_x \varphi_t^\Lambda \rme^{\rmi t \nu } \oplus J Q_t^\Lambda  W_x \varphi_t^\Lambda  \rme^{\rmi t \nu } \nn \\
		&\qquad\qquad\qquad\qquad\qquad\qquad\qquad\qquad\qquad - W_x \oplus J  W_x \Big) \Big\Vert\,.  \label{eq:aux-eq-122}
	\end{align}
	The bound for \eqref{eq:aux-eq-121} follows from \cref{lem:Approx-V-by-V-infty-Conv-Rates-2}, the Heinz-Kato theorem, and \eqref{eq:x-independent-estimate-W}
	\begin{align}
		\eqref{eq:aux-eq-121} &\leq C \Lambda^{-\gamma/4} (1+x^2)^{-1/4} \big\Vert \big((1+y^2)\oplus (1+y^2)\big)^{1/4} (W_x \oplus J  W_x) \big\Vert  \nn \\
		&\leq C \Vert (1+y^2)^{1/4} W\Vert_2 \,. \label{eq:aux-eq-123}
	\end{align}
	For the last term, \eqref{eq:aux-eq-122}, using the fact that all operators, except the Laplacian, appearing in the definition of $\cA^\Lambda(t)$ (see \eqref{eq:A(t)-Def}) are uniformly bounded in $\Lambda$, we first observe that
	\begin{align}
		\eqref{eq:aux-eq-122}\leq C  (1+x^2)^{-1/4}\left\Vert 
	(-\Delta +1 )\left(Q_t^\Lambda  W_x \varphi_t^\Lambda \rme^{\rmi t \nu } - W_x \right)    \right\Vert \,. \label{eq:aux-eq-140}
	\end{align}
	By separating the terms where all derivatives act on $W_x$ from those where they act on $\varphi_t^\Lambda$ or $Q_t^\Lambda$, we get
	\begin{align}
		\eqref{eq:aux-eq-140} &\leq C \sum_{|\beta|\in\{0, 2\}} (1+x^2)^{-1/4}\Big\Vert   Q_t^\Lambda  (\partial^\beta W_x) \varphi_t^\Lambda\rme^{\rmi t \nu } -\partial^\beta W_x  \Big\Vert \nn \\
		&\quad +C \Lambda^{-1/3} \sum_{|\beta|\leq1} (1+x^2)^{-1/4} \Big\Vert \partial^\beta W_x\Big\Vert_2 \,. \label{eq:aux-eq-124}
	\end{align}
	
	Applying \cref{lem:Approx-V-by-V-infty-Conv-Rates-2} together with the Heinz-Kato theorem and again using \eqref{eq:x-independent-estimate-W}, yields
	\begin{align}
		\eqref{eq:aux-eq-124} \leq C \Lambda^{-\gamma/4} \sum_{|\beta|\leq2} \big\Vert   (1+y^2)^{1/4} \partial^\beta W\big\Vert_2  + C\Lambda^{-1/6} \sum_{|\beta|\leq2} \Vert \partial^\beta W\Vert_\infty \,.\label{eq:aux-eq-125}
	\end{align}
	Collecting all contributions, we conclude
	\begin{align}
		\eqref{eq:V-derivative-approx} &\leq \eqref{eq:aux-eq-118} +\eqref{eq:aux-eq-119} \leq \eqref{eq:aux-eq-69} +\eqref{eq:aux-eq-121} +\eqref{eq:aux-eq-122} \nn \\
		&\leq \eqref{eq:aux-eq-69} + \eqref{eq:aux-eq-123}	+\eqref{eq:aux-eq-125} \nn \\
		&\leq C \Lambda^{-\gamma/4} \sum_{|\beta|\leq2} \left( \Vert h^{1/4} \partial^\beta W\Vert_2 + \Vert \partial^\beta W\Vert_\infty \right)  \,. \label{eq:aux-eq-141}
	\end{align}
	The second estimate required for \eqref{eq:V-to-V-infty-convergence-rate-on-x-dependent-state-time-deriv} concerns \eqref{eq:state-derivative-approx}, namely, the case where the time derivative acts on the state. 
	This bound follows by a straightforward application of \cref{lem:Approx-V-by-V-infty-Conv-Rates-2} and \cref{lem:Approx-V-by-V-infty-Conv-Rates}, together with the Heinz–Kato theorem:
	\begin{align}
		\eqref{eq:state-derivative-approx}\leq C \Lambda^{-\gamma/4} \big(\Vert h^{1/4} W\Vert_2 + \Vert W\Vert_\infty\big) \,. \label{eq:aux-eq-126}
	\end{align}
	Combining \eqref{eq:aux-eq-141} with \eqref{eq:aux-eq-126}, we conclude the claimed convergence rate in \eqref{eq:V-to-V-infty-convergence-rate-on-x-dependent-state-time-deriv}.
\end{proof}

We now use the convergence established in \cref{lem:Approx-V-by-V-infty-on-specific-state} to prove uniform boundedness with respect to both $x$ and $\Lambda$. In fact, these uniform bounds remain valid even after multiplying with a Bogoliubov map $\cZ_0^\Lambda$. The estimates provided in the following lemma play a crucial role in the tracer localization argument (see \cref{lem:Conditons-LNS15-Theorem-8}).

\begin{lemma}\label{lem:uniform-boundedness-F-Lambda}
Let $\beta\in \BN_0^3$ and $\gamma$ be given by \eqref{eq:def-gamma}.
	Assume the conditions of \cref{lem:Approx-V-by-V-infty-Conv-Rates-2}, \cref{con:Initial condition derivative condensate}$_{4}$ and \cref{Assumption:Initial-datum-and-potential}$_{\max\{|\beta|,2\}}$.
		For all densities $\rho\geq1$ and volumes $\Lambda\geq1$ let $\cZ_0^\Lambda\in \cL(L^2\oplus JL^2)$ be a unitarily implementable Bogoliubov map 	such that $\exists C>0$ and $0<\epsilon\leq  \gamma/4$ with $\forall
 	\Lambda,\rho\geq1$
 	\begin{gather}
 		\Vert \widehat{\cZ_0^\Lambda}\cT^{-1} (\tau \oplus J\tau J^*)\Vert_{\cL(L^2\oplus JL^2)} \leq C \,, \quad \Vert \cZ_0^\Lambda \Vert_{\cL(L^2\oplus JL^2)}\leq C\Lambda^{\epsilon} \,. \label{eq:Z-Lambda-regularized-bound} 
 	\end{gather}	 
	  Then for all times $T\geq0$, there exists a $C>0$ such that for all $\rho,\Lambda\geq1$, $x\in \BR^3$ and $-T\leq t\leq T$
		\begin{align}
			(1+x^2)^{-1/4}\left\Vert \cZ_0^\Lambda (\cV_t^\Lambda)^{-1} \left( Q_t^\Lambda (\partial^\beta W_x) \varphi_t^\Lambda \oplus J Q_t^\Lambda (\partial^\beta W_x) \varphi_t^\Lambda \right) \right\Vert_{L^2\oplus JL^2} &\leq C \,, \label{eq:uniform-bound-Z-Lambda-applied-to-state} \\
			(1+x^2)^{-1/4}\left\Vert \partial_t\left( \cZ_0^\Lambda (\cV_t^\Lambda)^{-1} \left( Q_t^\Lambda  W_x \varphi_t^\Lambda \oplus J Q_t^\Lambda  W_x \varphi_t^\Lambda \right)\right) \right\Vert_{L^2\oplus JL^2} &\leq C \,. \label{eq:uniform-bound-time-derivative-Z-Lambda-applied-to-state}
		\end{align}
\end{lemma}
\begin{proof}
	We begin with the proof of \eqref{eq:uniform-bound-Z-Lambda-applied-to-state}. We have
	\begin{align}
		&\eqref{eq:uniform-bound-Z-Lambda-applied-to-state}\leq \left\Vert \cZ_0^\Lambda (\cV_t^\infty)^{-1}   (\partial^\beta W_x\oplus J\partial^\beta W_x) \right\Vert \label{eq:aux-eq-127}  \\
		&+  (1+x^2)^{-1/4} \Vert \cZ_0^\Lambda\Vert_{\rm op}\Big\Vert   (\cV_t^\Lambda)^{-1} \left( Q_t^\Lambda (\partial^\beta W_x) \varphi_t^\Lambda \oplus J Q_t^\Lambda (\partial^\beta W_x) \varphi_t^\Lambda \right) \nn \\
		&\qquad\qquad\qquad\qquad\qquad  - (\cV_t^\infty)^{-1}  (\partial^\beta W_x\oplus J\partial^\beta W_x) \Big\Vert \,. \label{eq:aux-eq-128}
	\end{align}
	We first estimate \eqref{eq:aux-eq-127}. For this, we  prove for $f\in L^2(\BR^3)$ 
	\begin{align}
		\Vert \widehat{\cZ_0^\Lambda} \widehat{(\cV_t^\infty)^{-1}} (f \oplus J\cC R f) \Vert_{L^2\oplus JL^2} \leq C \Vert f\Vert_2 \,. \label{eq:Z-Lambda-V-infty-bound}
	\end{align}
	Recall the explicit form of the Bogoliubov map $\cV_t^\infty$ in \eqref{eq:-V-infty-Explicit}:
	\begin{gather}
	\widehat{\cV_t^\infty} =   \begin{pmatrix}
	L(t) &  M(t)^* \cC R J^*\\
	J\cC R M(t) & J L(t)J^*
\end{pmatrix} \label{eq:-V-infty-Explicit-1} \,, \\
		L(t)  = \cos (\omega t ) -\rmi \frac{\tfrac{p^2}{2}+ (2\pi)^{3/2}\widehat{V}}{\omega} \sin(\omega t)\,, \quad
	M(t) = -\rmi \frac{(2\pi)^{3/2}\widehat{V}}{\omega} \sin(\omega t) \,. \nn
\end{gather}
We remark that from \eqref{eq:Z-Lambda-regularized-bound}, we can conclude
	\begin{align}
		\Vert  \widehat{\cZ_0^\Lambda} (f \oplus J\cC R f) \Vert_{\cL(L^2\oplus JL^2)} \leq \Vert  \widehat{\cZ_0^\Lambda}\cT^{-1} (\tau\oplus J\tau J^*)\Vert \Vert (f \oplus J\cC R f)\Vert \label{eq:Z-Lambda-symmetric-bound}
	\end{align}
	  by inserting  $\cT^{-1}\cT$.
Now, since $\widehat{(\cV_t^\infty)^{-1}}= \widehat{\cV_{-t}^\infty}$,  using \eqref{eq:Z-Lambda-symmetric-bound} and \eqref{eq:Z-Lambda-regularized-bound}, we get
\begin{align}
	&\eqref{eq:Z-Lambda-V-infty-bound} = \left\Vert \widehat{\cZ_0^\Lambda} \begin{pmatrix}
		\big(L(-t) + M^*(-t)\big) f \\
		J\cC R \big( M(-t) + L^*(-t)\big)f 
	\end{pmatrix} \right\Vert   \nn \\
	&\leq   
	\left\Vert 
	\widehat{\cZ_0^\Lambda}\begin{pmatrix}
	\cos(\omega t) f \\
	J\cC R \cos(\omega t) f
	\end{pmatrix} 
	\right\Vert 
	 + 
	  \left\Vert \widehat{\cZ_0^\Lambda}
	 \begin{pmatrix}
	 	\rmi \tfrac{p^2}{2\omega} \sin(\omega t) f \\
	 	J \cC R (-\rmi ) \tfrac{p^2}{2\omega} \sin(\omega t) f
	 \end{pmatrix}
	 \right\Vert \nn\\
	 &\leq  C  
	\left\Vert 
	f
	\right\Vert _2
	 + C
	  \left\Vert \begin{pmatrix}
	  \tau^{-1} & \\
	  0&J\tau^{-1} J^*
\end{pmatrix}	   \cT
	 \begin{pmatrix}
	 	\rmi \tfrac{p^2}{2\omega} \sin(\omega t) f \\
	 	J \cC R (-\rmi ) \tfrac{p^2}{2\omega} \sin(\omega t) f
	 \end{pmatrix}
	 \right\Vert \,. \label{eq:aux-eq-129}
\end{align}
Applying the definition of $\cT$ (see \eqref{def:T-infty-volume}) to the second term in \eqref{eq:aux-eq-129}, we find
	\begin{align}
		&\left\Vert \begin{pmatrix}
	  \tau^{-1}  & \\
	  0&J\tau^{-1} J^*
\end{pmatrix}	   \cT
	 \begin{pmatrix}
	 	\rmi \tfrac{p^2}{2\omega} \sin(\omega t) f \\
	 	J \cC R (-\rmi ) \tfrac{p^2}{2\omega} \sin(\omega t) f
	 \end{pmatrix}
	 \right\Vert \nn \\
	 & = C \Vert \tau^{-2} \tfrac{p^2}{2\omega} \sin(\omega t) f\Vert_2 
	 \leq C\Vert f\Vert_2 \,.  \label{eq:aux-eq-130}
	\end{align}
	From \eqref{eq:aux-eq-129} and \eqref{eq:aux-eq-130} we conclude \eqref{eq:Z-Lambda-V-infty-bound}, and hence obtain the bound $\eqref{eq:aux-eq-127}\leq C$. Together with  the estimate $\eqref{eq:aux-eq-128}\leq C$, obtained using \cref{lem:Approx-V-by-V-infty-on-specific-state}, this yields \eqref{eq:uniform-bound-Z-Lambda-applied-to-state}.
	
	The estimate of \eqref{eq:uniform-bound-time-derivative-Z-Lambda-applied-to-state} follows analogously using \cref{lem:Approx-V-by-V-infty-on-specific-state}, \eqref{eq:Z-Lambda-regularized-bound} and $\epsilon\leq \gamma/4$:
	\begin{align}
			& (1+x^2)^{-1/4}\left\Vert \partial_t\left( \cZ_0^\Lambda (\cV_t^\Lambda)^{-1} \left( Q_t^\Lambda  W_x \varphi_t^\Lambda \oplus J Q_t^\Lambda  W_x \varphi_t^\Lambda \right)\right) \right\Vert_{L^2\oplus JL^2} \nn \\
			&\leq C \Lambda^{\epsilon-\gamma/4} +  \left\Vert \cZ_0^\Lambda (\cV_t^\infty)^{-1} (-\cA^\infty)(W_x\oplus JW_x) \right\Vert \nn \\
			&\leq C + \left\Vert  \cZ_0^\Lambda (\cV_t^\infty)^{-1} \big( \tfrac{p^2}{2} W_x\oplus J-\tfrac{p^2}{2}W_x\big) \right\Vert \label{eq:aux-eq-131}
	\end{align}
	The second term in \eqref{eq:aux-eq-131} can now be estimated using \eqref{eq:-V-infty-Explicit-1}, \eqref{eq:Z-Lambda-regularized-bound} and \eqref{eq:Z-Lambda-symmetric-bound}:
	\begin{align}
		\eqref{eq:aux-eq-131} 
	 &\leq  C + C\left\Vert 
	\tau^{-2}p^2\cos(\omega t) W_x
	\right\Vert 
	 +  C\left\Vert
	 	\rmi \tfrac{p^2}{2\omega}\left( \tfrac{p^2}{2} + (2\pi)^{3/2}\widehat{V}\right) \sin(\omega t) W_x 
	 \right\Vert \nn\\
	 &\leq C +C\Vert W\Vert_{H^2} (1+|t|)(1 + \Vert \widehat{V}\Vert_\infty) \,, \label{eq:aux-eq-132}
	\end{align}
	 where in the last term we estimated small and large momenta separately. Now, \eqref{eq:uniform-bound-time-derivative-Z-Lambda-applied-to-state} follows from \eqref{eq:aux-eq-132}.
\end{proof}

\section{An Explicit Example of $\cZ_0^\Lambda$} \label{sec:Construction-of-Z0}
In this section, we construct a unitarily implementable Bogoliubov map $\cZ_0^\Lambda$ to approximate $\cZ_0^\infty=\widecheck{\cT}$ in the sense of \cref{con:Infinite-volume-Bog-maps}. For the reader's convenience we repeat the definition of $\cT$:
 \begin{align}
  	\cT= \frac{1}{2} 
  	\begin{pmatrix}
  		\tau +\tau^{-1} & (\tau -\tau^{-1}) R \cC  J^*\\
  		J \cC  R (\tau -\tau^{-1}) & J(\tau +\tau^{-1})   J^*
  	\end{pmatrix}  	
  	 \,, \label{def:T-infty-volume-1}
\end{align}
with $\cC\psi=\psi^*$, $R\psi(p)=\psi(-p)$, $\tau=\frac{1}{(1+T)^{1/4}}$. The operator $T$ is given by $T= |p|^{-1} (2\pi)^{3/2}\widehat{V}(p) |p|^{-1}$, where we assume $\widehat{V}(0)\geq0$.
The main difficulty in constructing $\cZ_0^\Lambda$ arises from the fact that, in general, $\cT$ is neither a bounded operator nor unitarily implementable, whereas $\cZ_0^\Lambda$ must satisfy both properties.
To overcome this issue, we introduce an infrared cutoff by replacing $|p|^{-1}$ in the definition of $T$ with $(p^2 + \Lambda^{-\epsilon})^{-1/2}$ for some $\epsilon > 0$.
 The resulting operator is bounded but still not unitarily implementable.
 To ensure unitary implementability, we further replace $(2\pi)^{3/2} \widehat{V}(p)$ by its finite-volume approximation
 $\cK_1^\Lambda:= \tilde{K}_1^\Lambda \rme^{2\rmi t \nu}$, $\nu=\int V -\mu^\infty$. 

The operator obtained in this way is then dressed with $\widehat{Q_0^\Lambda}$ so as to leave the excitation space invariant. More precisely, our approximation of $\tau$ is defined by 
\begin{align}
	\tau^{\Lambda}&= (1+T^{\Lambda})^{-\frac{1}{4}} \,, \\
	T^\Lambda &= \widehat{Q_0^{\Lambda}} \left(p^2+\Lambda^{-\epsilon}\right)^{-1/2} \widehat{\cK_1^{\Lambda}} \left(p^2+\Lambda^{-\epsilon}\right)^{-1/2} \widehat{Q_0^{\Lambda}} \,. \label{eq:A-Definition}
\end{align}
And we set
\begin{align}
	\widehat{\cZ_0^\Lambda}= 
	\begin{pmatrix}
	\tau^{\Lambda} + (\tau^{\Lambda})^{-1}  & \left(\tau^{\Lambda} - (\tau^{\Lambda})^{-1} \right)\cC R J^* \\
	 J\left(\tau^{\Lambda} - (\tau^{\Lambda})^{-1}\right)\cC R & J \left(\tau^{\Lambda} + (\tau^{\Lambda})^{-1}\right) J^*
\end{pmatrix}	
\,. \label{eq:def-Z-Lambda}
\end{align}
In order for $\cZ_0^\Lambda$ to define a Bogoliubov map, it is necessary to assume that the initial condensate is real-valued.
In the following, we verify that this choice of $\cZ_0^\Lambda$ satisfies \cref{con:Infinite-volume-Bog-maps} (see \cref{lem:Z-Lambda-Properties} for a precise statement).

\begin{rem}
	On the Torus, an infrared regularization of the operator diagonalizing the Bogoliubov dynamics is not necessary, as the ground state of the Laplacian (corresponding to the condensate) is separated by a gap scaling like $\Lambda^{-2/3}$, which naturally provides an infrared cutoff. The cutoff introduced here, $\Lambda^{-\epsilon}$ with $\epsilon>0$ sufficiently small, would then correspond to allowing only excitations with momenta that are not too small in the initial data. 
\end{rem}

The following lemma allows us to control the difference of $\tau^\Lambda$ and $\tau$ coming from the regularized momentum terms $(p^2+\Lambda^{-\epsilon})^{-1/2}$ and $\widehat{Q_0^\Lambda}$.
\begin{lemma}\label{lem:Q-and-Momentum-properties}
	Let $\widehat{V}\geq0$ and assume \cref{con:Initial condition}.
	\begin{itemize}
		\item[i)] Then $ \widehat{Q_0^\Lambda}D(|p|)\subset D(|p|)$ , $\widehat{\tilde{K}}_1^\Lambda D(|p|)\subset D(|p|)$ and $\exists C>0$ such that $\forall \Lambda\geq1$ and $f\in D(|p|)$
		\begin{align}
		\Vert \widehat{Q_0^\Lambda} |p| - |p| \widehat{Q_0^\Lambda} \Vert_{\cL(L^2)}  +\Vert \widehat{\tilde{K}}_1^\Lambda |p| - |p| \widehat{\tilde{K}}_1^\Lambda \Vert_{\cL(L^2)}  \leq C \Lambda^{-1/3} \,.\label{eq:Q-and-K-abs-p-commutator-estimate}
		\end{align}
		
		\item[ii)]  Then we have for $\psi \in D(|p|)$, $2/3>\epsilon>0$ that
		\begin{align}
			\Big\Vert \Big( \frac{1}{(p^2+\Lambda^{-\epsilon})^{1/2}} - \frac{1}{|p|} \Big) |p|\psi \Big\Vert_2 &\leq \frac{1}{2\Lambda^{\epsilon/2}}  \Vert \psi\Vert_2\,, \label{eq:p-convergence}\\
			\Big\Vert \Big( \frac{1}{(p^2+\Lambda^{-\epsilon})^{1/2}} \widehat{Q_0^\Lambda} - \frac{1}{|p|} \Big) |p|\psi \Big\Vert_2 &\to 0\,, \ \Lambda\to \infty\,.  \label{eq:Q-dependent-p-convergence}
		\end{align}
	\end{itemize}
\end{lemma}
\begin{proof}[Proof of \cref{lem:Q-and-Momentum-properties}]

For the proof of the first part let $f\in D(|p|)$. Then we have with $\widehat{Q_0^\Lambda}=1 - \Lambda^{-1}\widehat{\varphi_0^\Lambda} \big\langle \widehat{\varphi_0^\Lambda}, \,.\,\big\rangle$ and $\big\Vert |p|\widehat{\varphi_0^\Lambda}\big\Vert_2^2 = \sum_{i=1}^3 \big\Vert p_i\widehat{\varphi_0^\Lambda}\big\Vert_2^2$  that
\begin{align}
	\left\Vert \left( \widehat{Q_0^\Lambda} |p| - |p| \widehat{Q_0^\Lambda} \right) f \right\Vert \leq 2 \Lambda^{-1/2} \big\Vert |p|\widehat{\varphi_0^\Lambda}\big\Vert_2 \Vert f \Vert_2 \leq C \Lambda^{-1/3}\Vert f \Vert_2\,,
\end{align}
which proves the estimate of $\widehat{Q_0^\Lambda}$ in \eqref{eq:Q-and-K-abs-p-commutator-estimate}. Now, since $\tilde{K}_1^\Lambda f= \varphi_0^\Lambda V\ast ((\varphi_0^\Lambda)^* f)$ we have $\widehat{\tilde{K}}_1^\Lambda f = \int K(p,r) f(r) d r$ with $ K(p,r)= (2\pi)^{-3/2} \int \widehat{\varphi_0^\Lambda}(p-q) \widehat{V}(q) \widehat{(\varphi_0^\Lambda)^*}(q-r) d q$. Then
\begin{align}
	&\Vert \big(\widehat{\tilde{K}}_1^\Lambda |p| - |p| \widehat{\tilde{K}}_1^\Lambda\big) f \Vert_2^2 = \int d p \left\vert \int d r (|r| -|p|) K(p,r) f(r) \right\vert^2 \nn \\
	&\leq  \int d p  \left(\int d r  \sum_{i=1}^3 \frac{|r_i-p_i|}{(2\pi)^{3/2}}  \left\vert  \int d q \widehat{\varphi_0^\Lambda}(p-q) \widehat{V}(q) \widehat{(\varphi_0^\Lambda)^*}(q-r) f(r) \right\vert \right)^2 \,. \label{eq:aux-eq-138}
\end{align}
Splitting $|r_i-p_i| \leq |r_i-q_i| +|q_i-p_i|$ gives
\begin{align}
	\eqref{eq:aux-eq-138}^{1/2}&\leq C  \sum_{i=1}^3 \Big\{\left\Vert  \vert  \widehat{\partial_i\varphi_0^\Lambda} \vert \ast \left( |\widehat{V} | \cdot |  \widehat{(\varphi_0^\Lambda)^*}| \ast |f| \right) \right\Vert_2 \nn\\
	&\quad+  \left\Vert  \vert  \widehat{\varphi_0^\Lambda} \vert \ast \left( |\widehat{V}| \cdot |  \widehat{(\partial_i\varphi_0^\Lambda)^*}| \ast |f| \right) \right\Vert_2\Big\} \nn \\
	&\leq C  \sum_{i=1}^3 \Vert \widehat{\partial_i\varphi_0^\Lambda}\Vert_1 \Vert \widehat{\varphi_0^\Lambda}\Vert_1 \Vert \widehat{V}\Vert_\infty \Vert f\Vert_2 \,, \nn 
\end{align}
which proves \eqref{eq:Q-and-K-abs-p-commutator-estimate} using \cref{con:Initial condition}.

The second part follows from
\begin{align}
	&\Big\Vert \big( (p^2+\Lambda^{-\epsilon})^{-1/2} \widehat{Q_0^\Lambda} - |p|^{-1} \big)|p|\psi \Big\Vert_2 \nn \\
	&\leq 
	\left\Vert  (p^2+\Lambda^{-\epsilon})^{-1/2} \big[\widehat{Q_0^\Lambda}, |p|\big]  \psi \right\Vert + \Big\Vert \big( (p^2+\Lambda^{-\epsilon})^{-1/2}|p| \widehat{Q_0^\Lambda} - 1 \big)\psi \Big\Vert_2 \nn \\
	&\leq  C \Lambda^{\epsilon/2- 1/3} \Vert \psi\Vert_2 	
	+\Vert ((p^2+\Lambda^{\epsilon} )^{-1/2}|p| -1 ) \psi \Vert_2
 + \Vert  (\widehat{Q_0^\Lambda} -1) \psi \Vert_2
	 \,, \label{eq:aux-eq-137}
\end{align}
where we used \eqref{eq:Q-and-K-abs-p-commutator-estimate} and $(p^2+\Lambda^{-\epsilon})^{-1/2} |p|\leq1$. We conclude the convergence 
 \eqref{eq:Q-dependent-p-convergence} from \eqref{eq:aux-eq-137} and \cref{cor:Strong-convergence-infinite-volume}. The estimate \eqref{eq:p-convergence} is proven analogously.
\end{proof}

We show that, in an appropriate sense, $\tau^{\Lambda}$ and its inverse converge to $\tau$ and $\tau^{-1}$, respectively. In contrast to $(\tau^{\Lambda})^{-1}$, the limit $\tau^{-1}$ is unbounded, so to obtain convergence we regularize both with $\tau^2$.

\begin{lemma}[Convergence of $\tau^{\Lambda}$] \label{lem:Tau-Lambda-Convergence}
Let $\widehat{V}\geq0$, $\widehat{V}(0)>0$ and $0<\epsilon<1/3$. Assume the conditions of \cref{lem:Approx-V-by-V-infty-Conv-Rates-2} and \cref{con:Initial condition derivative condensate}$_{2}$. Then we have that
	\begin{align}
		(\tau^{\Lambda})^{-1}\tau^2 \to \tau  \,,\quad 
		\tau^\Lambda \to \tau \,, \quad \text{as } \Lambda\to \infty
	\end{align}
	strongly as operators on $L^2(\BR^3)$.
\end{lemma}
\begin{proof}[Proof of \cref{lem:Tau-Lambda-Convergence}]
We begin by proving that for $\psi \in D(|p|^2)$ one has
\begin{align}
	\big\Vert \big( T^\Lambda -T \big) |p|^2 \psi \big\Vert_2 \to 0 \,, \ \Lambda\to \infty \,. \label{eq:ALambda-A-convergence}
\end{align}
Using the definitions of $T^\Lambda$ and $T$, we obtain
\begin{align}
	&\big\Vert \big( T^\Lambda -T \big) |p|^2 \psi \big\Vert_2 \nn\\
	&\left\Vert \left( \widehat{Q_0^{\Lambda}} \left(p^2+\Lambda^{-\epsilon}\right)^{-1/2} \widehat{\cK_1^{\Lambda}} \left(p^2+\Lambda^{-\epsilon}\right)^{-1/2} \widehat{Q_0^{\Lambda}} - |p|^{-1} (2\pi)^{3/2} \widehat{V} |p|^{-1} \right) |p|^2 \psi \right\Vert \nn\\
	&\leq 
	\left\Vert \left(\widehat{Q_0^{\Lambda}} - 1 \right)  (2\pi)^{3/2} \widehat{V}    \psi \right\Vert \label{eq:aux-eq-92}\\
	& + \left\Vert \left( \left(p^2+\Lambda^{-\epsilon}\right)^{-1/2} \widehat{\cK_1^{\Lambda}} \left(p^2+\Lambda^{-\epsilon}\right)^{-1/2} \widehat{Q_0^{\Lambda}} - |p|^{-1} (2\pi)^{3/2} \widehat{V} |p|^{-1} \right)|p|^2 \psi \right\Vert \,. \label{eq:aux-eq-93}
	\end{align}
	The term in \eqref{eq:aux-eq-92} converges to zero as $\Lambda \to \infty$ by \cref{cor:Strong-convergence-infinite-volume}.
Next, we estimate \eqref{eq:aux-eq-93} by
\begin{align}
	&\eqref{eq:aux-eq-93}\leq 
	\left\Vert \left( \left(p^2+\Lambda^{-\epsilon}\right)^{-1/2} - |p|^{-1} \right) |p| \widehat{V} \psi \right\Vert  \label{eq:aux-eq-94} \\
	& + \left\Vert  \left(p^2+\Lambda^{-\epsilon}\right)^{-1/2}  \left( \widehat{\cK_1^{\Lambda}} \left(p^2+\Lambda^{-\epsilon}\right)^{-1/2} \widehat{Q_0^{\Lambda}} - (2\pi)^{3/2} \widehat{V} |p|^{-1} \right) |p|^2 \psi \right\Vert \,. \label{eq:aux-eq-95}
\end{align}
The contribution in \eqref{eq:aux-eq-94} converges to zero by \cref{lem:Q-and-Momentum-properties}. To treat \eqref{eq:aux-eq-95}, we commute $|p|$ through the operator to compensate the left factor $(p^2+\Lambda^{-\epsilon})^{-1/2}$:
\begin{align}
	&\eqref{eq:aux-eq-95} \leq   
	\left\Vert  \left(p^2+\Lambda^{-\epsilon}\right)^{-1/2}   \widehat{\cK_1^{\Lambda}} \left(p^2+\Lambda^{-\epsilon}\right)^{-1/2} \left[\widehat{Q_0^{\Lambda}},|p| \right] |p| \psi \right\Vert  \label{eq:aux-eq-96} \\
	& +\left\Vert   \left(p^2+\Lambda^{-\epsilon}\right)^{-1/2}  \big[\widehat{\cK_1^{\Lambda}},|p|\big] \left(p^2+\Lambda^{-\epsilon}\right)^{-1/2} |p| \widehat{Q_0^{\Lambda}}  |p| \psi \right\Vert \label{eq:aux-eq-98} \\
	&+ \left\Vert  \left(p^2+\Lambda^{-\epsilon}\right)^{-\frac{1}{2}} |p|\right\Vert_\infty \left\Vert \left( \widehat{\cK_1^{\Lambda}} \left(p^2+\Lambda^{-\epsilon}\right)^{-1/2} \widehat{Q_0^{\Lambda}} - (2\pi)^{3/2} \widehat{V} |p|^{-1} \right) |p| \psi \right\Vert \,. \label{eq:aux-eq-99}	
\end{align}
By \cref{lem:Q-and-Momentum-properties}, the commutators satisfy $\Vert[\widehat{Q_0^{\Lambda}},|p| ]\Vert+ \Vert [\widehat{\cK_1^{\Lambda}},|p|]\Vert\leq C\Lambda^{-1/3}$
and therefore $\eqref{eq:aux-eq-96}\leq C\Lambda^{\epsilon -1/3} \to 0$ and $\eqref{eq:aux-eq-98} \to 0$ for $\epsilon<1/3$.
Finally, we bound \eqref{eq:aux-eq-99} by
\begin{align}
	\eqref{eq:aux-eq-99}&\leq 
	\Vert \widehat{\cK_1^{\Lambda}}\Vert_{\rm op}  
	 \left\Vert \left(  \left(p^2+\Lambda^{-\epsilon}\right)^{-1/2} \widehat{Q_0^{\Lambda}} - |p|^{-1} \right) |p| \psi \right\Vert \label{eq:aux-eq-100}	\\
	 &\quad + \left\Vert \left( \widehat{\cK_1^{\Lambda}}  - (2\pi)^{3/2} \widehat{V} \right)|p|^{-1} |p| \psi \right\Vert \,,\label{eq:aux-eq-101}	
\end{align}
where $\eqref{eq:aux-eq-100}\to0$ by \cref{lem:Q-and-Momentum-properties}, and $\eqref{eq:aux-eq-101}\to0$ by \cref{cor:Strong-convergence-infinite-volume}. Collecting all estimates, we conclude \eqref{eq:ALambda-A-convergence}.

We proceed with the proof of the strong convergence  of $(\tau^{\Lambda})^{-1}\tau^2$ to $\tau$. Let $\psi \in D(T^{1/2})$. Then $\tau^2\psi\in D(T)$, and  by the functional calculus we obtain
\begin{align}
	&\left\Vert \left( (\tau^{\Lambda})^{-1} -\tau^{-1}\right)\tau^2 \psi \right\Vert = \left\Vert \left( (1+T^\Lambda)^{1/4} -(1+T)^{1/4}\right)\tau^2 \psi \right\Vert \nn \\
	&=  \left\Vert \frac{\sin(\pi/4)}{\pi} \int_0^\infty \lambda^{-1/4} \left( \frac{1+T^\Lambda}{1+ \lambda (1+T^\Lambda)} - \frac{1+T}{1 +\lambda (1+T)}\right)\tau^2 \psi \,d\lambda \right\Vert\,. \label{eq:aux-eq-102}	
	\end{align}
	Using the identity $B/(1+\lambda B)= \lambda^{-1} - \lambda^{-1}(1+\lambda B)^{-1}$, we rewrite this as
	\begin{align}
	\eqref{eq:aux-eq-102} &=C \left\Vert  \int_0^\infty \lambda^{-5/4} \left( \frac{1}{1+ \lambda (1+T)} - \frac{1}{1 +\lambda (1+T^\Lambda)}\right)\tau^2 \psi \,d\lambda \right\Vert \nn \\
	&= C \left\Vert  \int_0^\infty \lambda^{-1/4} \frac{1}{1 +\lambda (1+T^\Lambda)} \left( T^\Lambda - T \right) \frac{1}{1+ \lambda (1+T)} \tau^2 \psi \,d\lambda \right\Vert \nn \\
	&\leq  C   \int_0^\infty  \frac{\lambda^{-1/4}}{1 +\lambda }  \left\Vert\left( T^\Lambda - T \right) \frac{1}{1+ \lambda (1+T)} \tau^2 \psi  \right\Vert d\lambda \,. \label{eq:Inverse-tauLambda-tau-Integral-estimate}
\end{align}
We prove that \eqref{eq:Inverse-tauLambda-tau-Integral-estimate} converges to zero as $\Lambda \to \infty$ by the dominated convergence theorem. First, we show pointwise convergence of the integrand. In fact,
\begin{align}
	&\left\Vert\left( T^\Lambda - T \right) \frac{1}{1+ \lambda (1+T)} \tau^2 \psi \right\Vert 
	=\left\Vert\left( T^\Lambda - T \right) |p|^2 \frac{|p|^{-1}}{1+ \lambda (1+T)} \frac{|p|^{-1}}{(1+T)^{1/2}} \psi \right\Vert \label{eq:Convergence-Integrand-ALambda-A}
\end{align}
and we have the bounds $\Vert\frac{|p|^{-1}}{1+ \lambda (1+T)}\Vert_{\rm op}\leq C(\lambda^{-1/2}+1)$ and $\Vert\frac{|p|^{-1}}{(1+T)^{1/2}}\Vert_{\rm op}\leq C$. This follows since for small momenta $p$ we have the bound $1/(p^2+\lambda\widehat{V}(p))^{1/2}\leq 2/(\lambda \widehat{V}(0))^{1/2}$, which holds since $\widehat{V}(0)>0$ and $\widehat{V}$ is continuous. The representation in \eqref{eq:Convergence-Integrand-ALambda-A}, together with the previously established convergence of $T^\Lambda$ to $T$ in \eqref{eq:ALambda-A-convergence}, implies that the integrand converges to zero.
Moreover, by an argument analogous to the proof of the convergence of $T^\Lambda$ to $T$, one shows that $\Vert T^\Lambda |p|^2\Vert_{\rm op} + \Vert T |p|^2\Vert_{\rm op}\leq C$. Consequently,
\begin{align}
	\frac{\lambda^{-1/4}}{1 +\lambda }  \left\Vert\left( T^\Lambda - T \right) \frac{1}{1+ \lambda (1+T)} \tau^2 \psi  \right\Vert \leq C  \frac{\lambda^{-1/4}(\lambda^{-1/2}+1)}{1 +\lambda }  \in L^1((0,\infty)_\lambda, \BR) \,, \label{eq:Bound-Dominated-Convergence-tauLambda-tau}
\end{align}
which provides an integrable dominating function. The dominated convergence theorem and \eqref{eq:Inverse-tauLambda-tau-Integral-estimate} therefore yield
$\Vert \left((\tau^{\Lambda})^{-1} - \tau^{-1} \right) \tau^2 \psi \Vert \to 0$ for all $\psi \in D(T^{1/2})$. To extend this convergence to the whole space $L^2$, it remains to show that $(\tau^{\Lambda})^{-1}\tau^2$ is uniformly bounded.
 By an argument analogous to that used in the proof of \eqref{eq:ALambda-A-convergence}, one obtains for all $\phi\in L^2$
\begin{align}
	\Vert (1 +T^\Lambda) \frac{1}{(1+T)^{2}} \phi\Vert \leq C \Vert \phi\Vert 
\end{align}
and thus $\Vert (1 +T^\Lambda)  \phi\Vert \leq C \Vert (1+T)^{2}\phi\Vert$.
By operator monotonicity of the fourth root, this implies the uniform bound $\Vert (\tau^{\Lambda})^{-1}\tau^2\Vert_{\rm op}\leq C$.
Combining this with the bound $ \Vert \tau\Vert_{\rm op}\leq 1$, we conclude that $\Vert \left((\tau^{\Lambda})^{-1} - \tau^{-1} \right) \tau^2 \psi \Vert \to 0$ for all $\psi \in L^2$. 

Next, for $\psi \in L^2$ we compute
\begin{align*}
	\Vert (\tau^\Lambda  - \tau) \tau \psi\Vert = \Vert \tau^{\Lambda} \left( \tau^{-1} - (\tau^\Lambda)^{-1} \right) \tau^2 \psi \Vert 
	\leq \Vert  \left( \tau^{-1} - (\tau^\Lambda)^{-1} \right) \tau^2 \psi \Vert \to 0\,.
\end{align*} 
Since $\tau$ is self-adjoint and injective, its range is dense in $L^2$. Together with the uniform bound $\Vert \tau^{\Lambda}\Vert_{\rm op} + \Vert \tau \Vert_{\rm op}\leq 2 $, this implies that $\tau^\Lambda$ converges strongly to $\tau$ on all of $L^2$.

\end{proof}

The following lemma shows that $\cZ_0^\Lambda$ defined in \eqref{eq:def-Z-Lambda} satisfies \cref{con:Infinite-volume-Bog-maps}.

\begin{lemma} \label{lem:Z-Lambda-Properties}
Let $\epsilon>0$. Let $\widehat{V}\geq0$, $\widehat{V}(0)>0$ and $0<\epsilon<1/3$. Assume the conditions of \cref{lem:Approx-V-by-V-infty-Conv-Rates-2} and \cref{con:Initial condition derivative condensate}$_{2}$. Assume $\varphi_0(y)\in\BR$ for all $y\in\BR^3$.
Then $\tau^{\Lambda}$ commutes with $\cC R$, and $\cZ_0^\Lambda$ is a unitarily implementable Bogoliubov map with
	\begin{align}
		\Vert \cZ_0^\Lambda\Vert_{\rm op} &\leq C \Lambda^{\epsilon/4} \,, \label{eq:Z-OP-bound} \\
		\Vert \cZ_0^\Lambda(\cZ_0^\Lambda)^* -1 \Vert_{\rm HS} &\leq C \Lambda^{1/2 +\epsilon} \,. \label{eq:Z-HS-bound}
	\end{align}
For all $f,g\in L^2$ we have the following convergence
 \begin{align}
 	\widehat{\cZ_0^\Lambda} \cT^{-1} \left(\tau f \oplus J\tau g\right) \to \tau f \oplus J\tau g\,. \label{eq:Z-convergence-to-T-infty-constructive}
 \end{align}
The commutator of $\cZ_0^\Lambda$  with translations $T_x$ converges to zero: 
 \begin{align}
	 	\sup_{x\in\BR^3} \frac{1}{(1+x^2)^{1/2}}\Big\Vert  \left[ \cZ_0^\Lambda,T_x \oplus J T_x J^*\right] F\Big\Vert_2 \to 0\,, \quad \text{as } \Lambda\to \infty \label{eq:Z-commutator-with-translations}
	 \end{align} for all $F\in L^2\oplus JL^2$.
For the corresponding unitary we have the following invariances 
\begin{align}
	U_{\cZ_0^\Lambda}^*\cF\left( \{\varphi_0^\Lambda\}^{\perp} \right)&\subset \cF\left( \{\varphi_0^\Lambda\}^{\perp} \right) \,, \label{eq:Z-Excitation-space-Invariance} \\
	U_{\cZ_0^\Lambda}^* Q(\rmd \Gamma (-\Delta+1))&\subset Q(\rmd \Gamma (-\Delta+1))\,. \label{eq:Z-dGamma-Delta-Invariance}
\end{align}
\end{lemma}
\begin{proof}[Proof of \cref{lem:Z-Lambda-Properties}]
Since $\varphi_0^\Lambda(y)\in\mathbb{R}$, for all $y\in\BR^3$, $\tau^\Lambda$ commutes with $\mathcal{C}R$, and hence $\mathcal{Z}_0^\Lambda$ is a Bogoliubov map.
 Implementability follows from the Hilbert–Schmidt bound \eqref{eq:Z-HS-bound}, which we now verify. Using the functional calculus and an orthonormal basis $\{e_k\}$ of $L^2$, we obtain
\begin{align}
	& \sum_{k=0}^\infty \Vert\big(\tau^{\Lambda} - (\tau^{\Lambda})^{-1}\big) e_k\Vert^2 = \sum_{k=0}^\infty \big\Vert\frac{1}{(1+T^\Lambda)^{1/4} } \big( 1 - (1+T^\Lambda)^{1/2} \big) e_k \big\Vert^2 \nn \\
	&=\sum_{k=0}^\infty \Big\Vert\frac{1}{(1+T^\Lambda)^{1/4} } \frac{1}{2} \int_0^1 \frac{T^\Lambda}{(1+t T^\Lambda)^{1/2} }e_k d t \Big\Vert^2 \leq 1/2\Vert T^\Lambda \Vert_{\rm HS} \,. \label{eq:Z-Off-diagonal-estimate}
\end{align} 
By the definition of $T^\Lambda$ in \eqref{eq:A-Definition} and \eqref{eq:K1-Def-2},
\begin{align*}
	\Vert T^\Lambda \Vert_{\rm HS}\leq \Lambda^{\epsilon} \big\Vert \widehat{\cK_1^\Lambda}\big\Vert_{\rm HS}\leq C\Lambda^{1/2+\epsilon} \,,
\end{align*}
and thus $\Vert \tau^{\Lambda} - (\tau^{\Lambda})^{-1}\Vert_{\rm HS}\leq C\Lambda^{1/2+\epsilon}$, proving \eqref{eq:Z-HS-bound}. Moreover, since $\big\Vert \widehat{\cK_1^\Lambda}\big\Vert_{\rm op}\leq C$, we have $\Vert (1+T^\Lambda) \Vert_{\rm op}\leq C \Lambda^{\epsilon}$. Taking the fourth root gives us $\Vert(\tau^{\Lambda})^{-1}\Vert_{\rm op}\leq \Lambda^{\epsilon/4}$, yielding \eqref{eq:Z-OP-bound}.

Next, we prove \eqref{eq:Z-Excitation-space-Invariance}. Since $T^\Lambda \widehat{P_0^\Lambda}=0$, where $P_0^\Lambda= 1/\Lambda |\varphi_0^\Lambda\rangle \langle \varphi_0^\Lambda|$,  it is easy to see that $(\tau^{\Lambda})^{\pm1}\widehat{P_0^\Lambda}=\widehat{P_0^\Lambda}$. Thus 
\begin{align}
		a(\varphi_0^\Lambda)U_{\cZ_0^\Lambda}^* \Omega=U_{\cZ_0^\Lambda}^* a\big(\varphi_0^\Lambda\big)\Omega =0 \,, \label{eq:Z-invariance-on-Vac}
	\end{align}
	so $U_{\cZ_0^\Lambda}^* \Omega\in \cF(\{\varphi_0^\Lambda\}^{\perp} )$. Now \eqref{eq:Z-Excitation-space-Invariance} follows from $[(\tau^{\Lambda})^{\pm1}, \widehat{Q_0^\Lambda}]=0$, the density of $\bigcup_{n=0}^\infty\{A(F_1)\cdot A(F_n)\Omega \,|\, F_j\in \{\varphi_0\}^{\perp} \oplus J \{\varphi_0\}^{\perp} \}$ in $\cF(\{\varphi_0^\Lambda\}^{\perp} )$ and the continuity of $U_{\cZ_0^\Lambda}^*$ in $\cF(L^2)$.

Next, we discuss \eqref{eq:Z-dGamma-Delta-Invariance}. For an orthonormal basis $\{e_k\}\subset H^1$, we compute using \eqref{eq:Bog-creation-annihilation}
\begin{align}
	\pscal{U_{\cZ_0^\Lambda}^*\psi, \rmd \Gamma (-\Delta) U_{\cZ_0^\Lambda}^*\psi} &=   \sum_{k} \Vert a((-\Delta)^{1/2} e_k) U_{\cZ_0^\Lambda}^*\psi \Vert^2  \nn \\
	 &=   \sum_{k} \Vert A\big( \cZ_0^\Lambda((-\Delta)^{1/2} e_k \oplus J0)\big) \psi \Vert^2 \,. \label{eq:aux-eq-133}
\end{align}
Using the definition of $\cZ_0^\Lambda$, \eqref{eq:def-Z-Lambda}, this splits into the diagonal term \eqref{eq:aux-eq-142} and the off-diagonal term \eqref{eq:aux-eq-134}: 
\begin{align}
	\eqref{eq:aux-eq-133}&\leq C  \sum_{k} \Vert a\big( \cF^{-1}(\tau^\Lambda+(\tau^\Lambda)^{-1}) |p| \widehat{e_k} \big) \psi \Vert^2  \label{eq:aux-eq-142} \\
	&\quad + 	\sum_{k} \Vert a^*\big( \cF^{-1}(\tau^\Lambda - (\tau^\Lambda)^{-1})  \cC R |p| \widehat{e_k} \big) \psi \Vert^2  \,.
	 \label{eq:aux-eq-134}
\end{align}
Since $(\tau^\Lambda - (\tau^\Lambda)^{-1})    |p|$ is Hilbert-Schmidt, which follows analogously to \eqref{eq:Z-Off-diagonal-estimate} using the commutator bound \eqref{eq:Q-and-K-abs-p-commutator-estimate}, the second term is bounded by
\begin{align}
	\eqref{eq:aux-eq-134}\leq \Vert (\tau^\Lambda - (\tau^\Lambda)^{-1})   |p| \Vert_{\rm HS} \Vert (\cN+1)\psi\Vert<\infty \,. \label{eq:aux-eq-139}
\end{align}
The first term \eqref{eq:aux-eq-142} can be rewritten as 
\begin{align}
	\eqref{eq:aux-eq-142}&= \pscal{ \widehat{\psi} , \rmd\Gamma \big(  (\tau^\Lambda+(\tau^\Lambda)^{-1}) |p|^2 (\tau^\Lambda+(\tau^\Lambda)^{-1}) \big) \widehat{\psi} } \nn \\
	&= \sum_{n=0}^\infty \sum_{k=1}^n \left\Vert |p|_k (\tau^\Lambda+(\tau^\Lambda)^{-1})_k  \psi^{(n)} \right\Vert_2^2 \,.\label{eq:aux-eq-135}
\end{align}
But with the integral representations for  $0<\alpha<1$, $\phi\in \sH$
	\begin{align}
		(1+T^\Lambda)^\alpha \phi &= \frac{\sin(\alpha\pi)}{\pi}  \int_0^\infty \frac{t^{-\alpha} (1+T^\Lambda)}{1+ t (1+T^\Lambda) } \phi \,d t \,, \label{eq:Integral-rep-fractional-op} \\
		 (1+T^\Lambda)^{-\alpha} \phi &= \frac{\sin(\alpha\pi)}{\pi}  \int_0^\infty \frac{t^{-\alpha} }{t+  (1+T^\Lambda) } \phi \,d t \label{eq:Integral-rep-fractional-inverse-op}
	\end{align}
 and \cref{lem:Q-and-Momentum-properties} it is easily shown that $\Vert |p| (\tau^\Lambda)^{\pm1} \phi \Vert_2 \leq C_\Lambda \Vert \phi\Vert + C \Vert |p| \phi \Vert$ for all $\phi\in H^1$, and thus 
\begin{align}
	\eqref{eq:aux-eq-135} \leq C_\Lambda\Vert (\cN+1)^{1/2}\psi\Vert^2 + C\pscal{\psi ,\rmd\Gamma(-\Delta) \psi } <\infty \,. \label{eq:aux-eq-136}
\end{align}
Now, we conclude \eqref{eq:Z-dGamma-Delta-Invariance} from \eqref{eq:aux-eq-139} and \eqref{eq:aux-eq-136}.

Next we prove the convergence of the commutator of $\cZ_0^\Lambda$ with translation $T_x$ in \eqref{eq:Z-commutator-with-translations}. We start by showing that for all $x\in \BR^3$ and $f\in L^2$:
\begin{align}
	\big\Vert \big[ Q_0^\Lambda , T_x \big] f\big\Vert_2 \leq \Lambda^{-1/3} |x| \Vert f\Vert_2 \,, \label{eq:Commutator-Q-0-with-translations}\\
		\big\Vert \big[ \tilde{K}_1^\Lambda , T_x \big] f\big\Vert_2 \leq \Lambda^{-1/3} |x| \Vert f\Vert_2 \,. \label{eq:Commutator-K-1-with-translations}
\end{align}
To this end, using the definition of $Q_0^\Lambda$, we obtain
\begin{align}
	&\eqref{eq:Commutator-Q-0-with-translations} = \Lambda^{-1} \Vert \varphi_0^\Lambda - T_x \varphi_0^\Lambda\Vert_2 |\pscal{ \varphi_0^\Lambda, T_x f}| + \Lambda^{-1} \Vert T_x \varphi_0^\Lambda\Vert_2 |\pscal{\varphi_0^\Lambda, (T_x-1)f}| \nn \\
	&\leq 2 \Lambda^{-1} \Vert f\Vert_2 \Vert (T_x-1)\varphi_0^\Lambda \Vert_2 \leq 2 \Lambda^{-1} \Vert f\Vert_2 \big\Vert \sin(px/2)\widehat{\varphi_0^\Lambda} \big\Vert_2  \nn \\
	&\leq  \Lambda^{-1}|x| \Vert f\Vert_2 \big\Vert |p|^2\widehat{\varphi_0^\Lambda} \big\Vert_2^{1/2} \big\Vert \widehat{\varphi_0^\Lambda} \big\Vert_2^{1/2} \leq C \Vert f\Vert_2 \Lambda^{-1/3} |x| \,,
\end{align}
which proves \eqref{eq:Commutator-Q-0-with-translations}. For the commutator with $\tilde{K}_1^\Lambda(0)$ we use its definition \eqref{eq:K1-Def-2} and the mean-value theorem to get
\begin{align}
	&\eqref{eq:Commutator-K-1-with-translations} = \Vert \tilde{K}_1^\Lambda T_x f - T_x \tilde{K}_1^\Lambda f \Vert_2 \nn \\
	& = \left\Vert \varphi_0^\Lambda V\ast \big(( T_x\varphi_0^\Lambda)^* T_xf\big)  - T_x\varphi_0^\Lambda   V\ast \big( (\varphi_0^\Lambda)^* T_xf \big) \right\Vert \nn \\
	&\leq 2\Vert (T_x-1)\varphi_0^\Lambda \Vert_\infty \Vert V\Vert_1 \Vert \varphi_0^\Lambda\Vert_\infty \Vert f\Vert_2 \nn \\
	&\leq C \Vert f\Vert_2 \Vert \nabla \varphi_0^\Lambda\Vert_\infty |x|\leq C \Vert f\Vert_2 \Lambda^{-1/3} |x| \,,
\end{align}
which proves \eqref{eq:Commutator-K-1-with-translations}.
With this at hand it is readily checked that 
\begin{align}
	\left\Vert \big[T^\Lambda, \widehat{T_x}\big] f \right\Vert \leq C\Lambda^{\epsilon-1/3} \Vert f\Vert_2 |x|
\end{align}
and with the use of the integral formulas \eqref{eq:Integral-rep-fractional-op} and \eqref{eq:Integral-rep-fractional-inverse-op} also 
\begin{align}
	\left\Vert \big[ (1+T^\Lambda)^{\pm1/4}, \widehat{T_x} \big] f\right\Vert_2 \leq C\Lambda^{\epsilon-1/3} \Vert f\Vert_2 |x| \,.
\end{align}
Thus it follows for all $x\in\BR^3$ and $F\in L^2\oplus JL^2$ that
\begin{align}
	\sup_{x\in\BR^3} \frac{1}{(1+x^2)^{1/2}} \Big\Vert \left[ \cZ_0^\Lambda,T_x \oplus J T_x J^*\right] F\Big\Vert \leq C\Lambda^{\epsilon-1/3} \Vert F\Vert
\end{align}
and then also \eqref{eq:Z-commutator-with-translations}.

It remains to prove the convergence of $\widehat{\cZ_0^\Lambda}$ to $\cT$ in the sense of \eqref{eq:Z-convergence-to-T-infty-constructive}. This follows from the identity
\begin{align}
	&\widehat{\cZ_0^\Lambda}\cT^{-1} (\tau \oplus J\tau J^*)  \nn \\
	&= \frac{1}{2} \begin{pmatrix}
		(\tau^{\Lambda})^{-1}\tau^2 + \tau^\Lambda & \left(  \tau^\Lambda - (\tau^{\Lambda})^{-1}\tau^2 \right) \cC R J^* \\
		J \cC R\left(  \tau^\Lambda - (\tau^{\Lambda})^{-1}\tau^2 \right) & J\left( (\tau^{\Lambda})^{-1}\tau^2 + \tau^\Lambda \right) J^*
	\end{pmatrix} 
\end{align}
together with the convergence of $\tau^\Lambda$ to $\tau$ established in \cref{lem:Tau-Lambda-Convergence}.
\end{proof}


\bibliography{Literature}

\end{document}